\documentclass[a4paper,aps,pra,showpacs,twocolumn,superscriptaddress]{revtex4-2} 

\usepackage[utf8]{inputenc}  
\usepackage[T1]{fontenc}     
\usepackage[american]{babel}  
\usepackage[svgnames,dvipsnames]{xcolor}
\usepackage[colorlinks=true,citecolor=DarkBlue,linkcolor=DarkBlue,urlcolor=DarkBlue,linktocpage,unicode]{hyperref}
\usepackage{amsmath,amssymb,amsthm, bm,mathtools,amsfonts,mathrsfs,bbm,dsfont} 
\usepackage{centernot}
\newtheorem{theorem}{Theorem}
\newtheorem{observation}[theorem]{Observation}

\newcommand{\tr}{{\mathrm{tr}}}

\newcommand{\eins}{\mathbbm{1}}
\newcommand{\swap}{\mathbb{S}}
\renewcommand{\vr}{\ensuremath{\varrho}}
\renewcommand{\vec}[1]{\ensuremath{\boldsymbol{#1}}}
\newcommand{\trans}{{\rm T}}
\usepackage{braket}
\usepackage[normalem]{ulem}
\usepackage{comment}

\newcommand{\inv}{^{\text{-}1}}

\begin{document}
\title{Hierarchy of saturation conditions for multiparameter quantum metrology bounds}
\author{Satoya Imai\hyperlink{email1}{\textsuperscript{*}}}
\affiliation{Institute of Systems and Information Engineering, University of Tsukuba, Tsukuba, Ibaraki 305-8573, Japan}
\affiliation{Center for Artificial Intelligence Research (C-AIR), University of Tsukuba, Tsukuba, Ibaraki 305-8577, Japan}

\author{Jing Yang\hyperlink{email2}{\textsuperscript{\textdagger}}}
\affiliation{Institute of Fundamental and Transdisciplinary Research, Institute of Quantum Sensing, and Institute for Advanced Study in Physics, Zhejiang University, Hangzhou 310027, China}
\affiliation{Nordita, KTH Royal Institute of Technology and Stockholm University, Hannes Alfv\'ens v\"ag 12, 106 91 Stockholm, Sweden.}

\author{Luca Pezzè\hyperlink{email3}{\textsuperscript{\textdaggerdbl}}}
\affiliation{Istituto Nazionale di Ottica del Consiglio Nazionale delle Ricerche (CNR-INO), Largo Enrico Fermi 6, 50125 Firenze, Italy}
\affiliation{European Laboratory for Nonlinear Spectroscopy (LENS), Via N. Carrara 1, 50019 Sesto Fiorentino, Italy}

\date{\today}
\begin{abstract}
The quantum Cramér-Rao (QCR) bound sets the ultimate local precision limit for unbiased multiparameter estimation. Unlike in the single-parameter case, however, its saturability is not generally guaranteed and is often analyzed through a hierarchy of commutativity-based conditions. Here, we resolve the logical structure of these conditions for unitary parameter-encoding transformations. We identify strict separations among the conditions, reveal previously overlooked gaps in their implications, and construct explicit counterintuitive examples that expose the boundaries among distinct classes. In particular, we show that commutativity of the parameter-encoding generators alone does not guarantee saturation of the QCR bound when realistic noise leads to mixed probe states. Our results provide a systematic classification of saturation conditions in multiparameter quantum metrology and clarify fundamental precision limits of noisy distributed quantum sensing beyond idealized pure-state regimes.
\end{abstract}

\maketitle

\section{Introduction}\label{sec:Introduction}

\begin{figure}[t]
    \centering
    \includegraphics[width=0.95\linewidth]{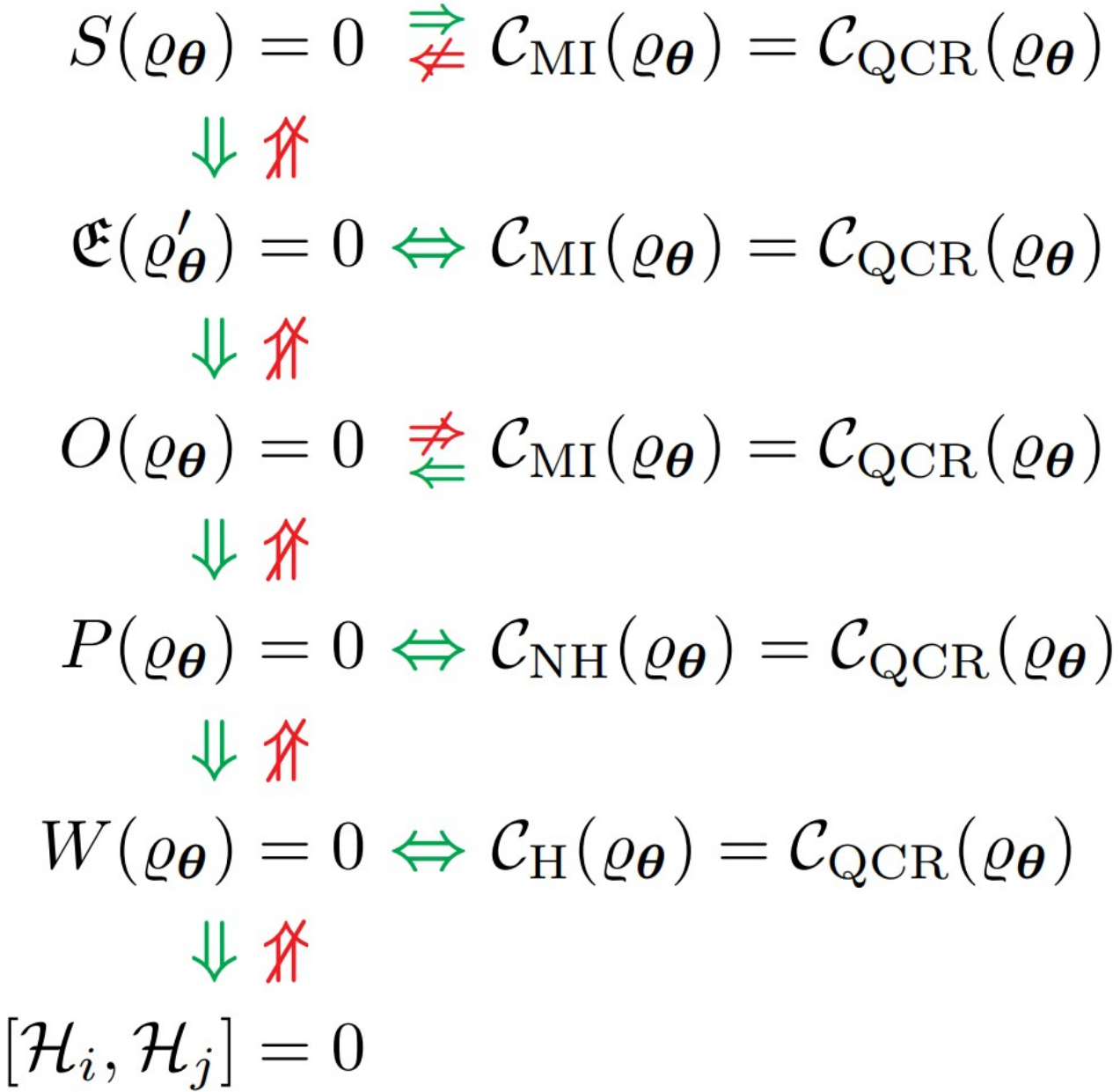}
    \caption{Schematic summary of the hierarchical chain in multiparameter quantum metrology. Here, the quantities $\mathcal{C}_{\rm MI} (\vr_{\vec{\theta}}), \mathcal{C}_{\rm QCR} (\vr_{\vec{\theta}}), \mathcal{C}_{\rm NH} (\vr_{\vec{\theta}})$, and $\mathcal{C}_{\rm H} (\vr_{\vec{\theta}})$ represent the most informative (MI), quantum Cramér-Rao (QCR), Nagaoka-Hayashi (NH), and Holevo (H) bounds, respectively defined in Section~\ref{sec:II.B}. For simplicity, we respectively denote $S(\vr_{\vec{\theta}})=0, \mathfrak{E}(\vr_{\vec{\theta}}^\prime) = 0,  O(\vr_{\vec{\theta}})=0, P(\vr_{\vec{\theta}})=0$, and $W(\vr_{\vec{\theta}})=0$ as the strong commutativity (SC), extended commutativity (EC), one-sided commutativity (OC), partial commutativity (PC), and weak commutativity (WC) conditions, explained in Section~\ref{sec:II.C}. The operator $\mathcal{H}_i$ is a parameter-encoding generator in Eq.~(\ref{eq:generator}). Further details are given in Section~\ref{sec:II.D}.}
    \label{fig1}
\end{figure}

Multiparameter quantum metrology aims to develop optimal strategies for simultaneously estimating several unknown parameters encoded in a quantum state~\cite{albarelli2020perspective,suzuki2020quantum,demkowicz2020multi,pezze2025advances}. This problem is practically relevant, as many quantum technologies involve multiparameter estimation tasks, including clock networks~\cite{ye2024essay}, biological imaging~\cite{aslam2023quantum}, learning quantum systems~\cite{gebhart2023learning}, designing quantum gates and algorithms~\cite{bharti2022noisy}, communication networks~\cite{wehner2018quantum}, and benchmarking quantum simulators~\cite{eisert2020quantum}. Despite recent advances in estimation protocols, several theoretical challenges remain, including computing relevant uncertainty bounds for the covariance matrix of locally-unbiased estimators, identifying conditions under which these bounds are attainable, and developing optimal estimation strategies.

The ultimate precision limit is characterized by the so-called most informative (MI) bound~\cite{albarelli2020perspective,suzuki2020quantum,demkowicz2020multi,pezze2025advances}, defined as the minimum of a chosen multiparameter figure of merit over all estimators and generalized measurements. In practice, however, evaluating the MI bound is generally infeasible. A common alternative approach is therefore to derive more tractable bounds and determine when they coincide with the MI bound. Within the theory of local asymptotic normality~\cite{kahn2009local,yamagata2013quantum,yang2019attaining,demkowicz2020multi}, the uncertainty bound introduced by Holevo~\cite{holevo2011probabilistic} approaches the MI bound for collective measurements on infinitely many copies of the state. Nevertheless, the Holevo bound remains computationally demanding~\cite{sidhu2021tight,albarelli2019evaluating,hayashi2023tight,hayashi2024finding} and experimentally challenging because it requires collective measurements~\cite{roccia2018entangling,conlon2023approaching}.

A more tractable, though less fundamental, alternative is the quantum Cramér-Rao (QCR) bound introduced by Helstrom~\cite{helstrom1969quantum,helstromBOOK1976}. The QCR bound has a simple expression in terms of the quantum state and its symmetric logarithmic derivative (SLD) operators. Still, it does not generally coincide with the MI bound, even in the asymptotic regime. This has motivated many efforts to characterize the conditions under which the QCR bound is saturable at the single-copy level. These saturation conditions are related to the commutation properties of the SLD operators and the quantum state, as summarized below (see Section~\ref{sec:II.C} for precise formulations):
\begin{itemize}
    \item[(i)] The weak commutativity (WC) condition~\cite{matsumoto2002new,ragy2016compatibility} is necessary for saturating the QCR bound and also becomes sufficient for pure states.

    \item[(ii)] The partial commutativity (PC) condition~\cite{yang2019optimal} is also necessary for saturating the QCR bound. It coincides with the WC condition for pure states and with the SC condition for full-rank states.

    \item[(iii)] The ``one-sided'' commutativity (OC) condition, as termed in this manuscript and motivated by Ref.~\cite{matsumoto2005geometricalPhD}, is necessary for saturating the QCR bound and is stronger than the PC condition.

    \item[(iv)] The extended commutativity (EC) condition on an enlarged Hilbert space~\cite{suzuki2020quantum,conlon2022gap} is equivalent to the saturation condition for the QCR bound, which requires identifying a suitable isometry into the extended space. Since it does not specify how to choose the extended SLD operators, establishing a constructive and testable characterization for general rank-deficient single-copy states remains an open problem in quantum information theory~\cite{horodecki2022five}.
    
    \item[(v)] The strong commutativity (SC) condition~\cite{nagaoka1987OnFisher,amari2000methods} is sufficient for saturating the QCR bound and also becomes necessary for full-rank states.
\end{itemize}

Multiparameter estimation exhibits a rich structure, involving pure or mixed states, separable or collective measurement strategies, commuting or incompatible generators of parameter encoding, and classical or quantum correlations in probe states. Despite extensive research, the overall picture remains fragmented, and the relations among the commutativity conditions above are still not fully understood. Clarifying these connections is essential for determining when the QCR bound can be saturated, especially in scenarios involving rank-deficient states and separable measurements. Such progress would facilitate the optimization of probe states and measurement schemes.

The goal of this manuscript is to clarify the hierarchical relations among the commutativity conditions associated with the saturation of the QCR bound. We identify gaps in this hierarchy, revealing both coincidences and distinctions between different conditions through several counterintuitive examples in physically relevant settings. Our results reveal the structure of commutativity relations, as summarized in Fig.~\ref{fig1}: green arrows denote established implications, while red arrows indicate forbidden relations. 

A key implication arises when parameter encoding is implemented through a unitary transformation generated by multiple Hamiltonians. In this setting, the WC condition becomes both necessary and sufficient to saturate the QCR bound for pure states. This implies the widely recognized fact: \textit{For pure states, the QCR bound is saturable when all parameter-encoding generators commute}~\cite{matsumoto2002new}. This observation is relevant to distributed quantum sensing applications, where commuting generators appear naturally~\cite{humphreys2013quantum,proctor2018multiparameter,ge2018distributed,guo2020distributed, liu2021distributed,malia2022distributed,malitesta2023distributed,kim2024distributed,hong2025quantum, pezze2025distributed,li2026multiparameter,minati2025multiparameter,yan2025scalable}. In this manuscript, we show that mixed separable and entangled probes can fail to saturate the QCR bound even when the generators commute.
This suggests that determining the ultimate sensitivity of realistic distributed sensors beyond idealized pure-state models is subtler than previously thought. 

This manuscript is organized as follows: In Section~\ref{sec:Background}, we first recall the general framework of multiparameter quantum metrology, following several in-depth review articles~\cite{albarelli2020perspective,suzuki2020quantum,demkowicz2020multi,pezze2025advances}. We then discuss the hierarchical structure of metrology bounds and outline the conditions under which they can be saturated. Also, we discuss the OC condition in detail (Observation~\ref{ob:OCisnecessary}) and the EC condition; for clarity and completeness of this manuscript, we will provide a comprehensive proof in Appendix~\ref{ap:commutativity_extension}. In Section~\ref{sec:III}, we focus on the WC condition as a simple route toward our goal. We present a general expression for determining the WC condition (Observation~\ref{ob:Gformsummary}) and provide several counterintuitive examples of the WC condition for commuting parameter-encoding generators (Observations~\ref{ob:commutingbutnotweak},~\ref{ob:local_ham_weak_always},~and~\ref{ob:local_ham_weak_always_twoqubit}). In Section~\ref{sec:IV}, we consider other commutativity conditions and give general forms for checking them (Observation~\ref{ob:GE:UPUandUMUandUSU}). We discuss whether, when a certain commutativity condition holds, the corresponding converse implication also holds for commuting Hamiltonians (Observation~\ref{ob:conversearrow}). Finally, in Section~\ref{sec:V}, we conclude and discuss the consequences of our work.

\section{Multiparameter quantum metrology}\label{sec:Background}
\subsection{General framework}\label{sec:II.A}
Multiparameter quantum metrology consists of four stages: state preparation, parameter encoding, measurement, and estimation. Let $\vr$ be an initial probe state in a preparation stage, and $\Lambda_{\vec{\theta}}$ be a quantum channel (completely positive and trace-preserving map) that encodes $m$ real parameters $\vec{\theta} = \{ \theta_1, \ldots, \theta_m \} \in \mathbb{R}^m$ into the probe state. We indicate $\vr_{\vec{\theta}} = \Lambda_{\vec{\theta}} (\vr)$ as the parameter-encoded state. More generally, one may have access to $\nu$ independent and identical copies of the state, $\vr_{\vec{\theta}}^{\otimes \nu}$.

At the measurement stage, a positive operator-valued measure (POVM) $\mathsf{E} = \{ E_\omega \}$ is applied, where $0 \leq E_\omega \leq \eins$ and $\sum_\omega E_\omega = \eins$. Upon measurement, the outcome $\omega$ occurs with probability $p(\omega | \vec{\theta}) = \tr(\vr_{\vec{\theta}}^{\otimes \nu} E_\omega)$. Repeating this procedure $n$ times yields a sequence of independent outcomes $\vec{\omega} = \{ \omega_1, \ldots, \omega_n \} \in \mathbb{R}^n$, with joint probability distribution $p(\vec{\omega} | \vec{\theta}) = \prod_{k=1}^{n} p(\omega_k | \vec{\theta})$. Note that $p(\vec{\omega} | \vec{\theta})$ depends on $\nu$ but we omit its dependency.

In the multi-copy setting ($\nu > 1$), several measurement strategies are employed (for details, see Ref.~\cite{suzuki2020quantum}): (i)~Local measurements (also known as uncorrelated or repetitive measurements), where the same measurement is independently applied to each copy of the state; (ii)~Adaptive measurements, where each copy is measured locally, but the choice of future measurement may depend on the outcomes of previous measurements based on feedback operations; (iii)~Separable measurements, where classical mixtures of local positive operators across the $\nu$ copies are considered~\cite{chitambar2014everything}; (iv)~Collective measurements (also called global or joint measurements), are performed on joint state $\vr_{\vec{\theta}}^{\otimes \nu}$ as a whole, using entangled measurements across the $\nu$ copies. These four measurement strategies differ in implementation complexity and achievable estimation precision, with collective measurements generally providing higher precision than the other strategies.

At the parameter-estimation stage, an estimator $\Tilde{\vec{\theta}} = \{ \Tilde{\theta}_1, \ldots, \Tilde{\theta}_m \} \in \mathbb{R}^m$ is considered. Here, each $\Tilde{\theta}_i = \Tilde{\theta}_i(\vec{\omega})$ is a function of the observed data and aims to estimate the corresponding true parameter $\theta_i$. In particular, an estimator is called locally-unbiased if $\braket{\Tilde{\theta}_i} = \theta_i$ and $\partial_i \braket{\Tilde{\theta}_j} = \delta_{ij}$ for each $i,j \in [1, m]$, where $\braket{\Tilde{\theta}_i} \equiv \sum_{\vec{\omega}} p (\vec{\omega} | \vec{\theta}) \Tilde{\theta}_i (\vec{\omega})$ is the statistical mean value of the estimator and $\partial_i \equiv \partial/\partial \theta_i$. In what follows, we consider locally-unbiased estimators and the corresponding sensitivity bounds.

A central goal in the quantum metrology framework is to develop optimal strategies that enhance the precision of parameter estimation. This uncertainty is typically characterized by the $m \times m$ covariance matrix of the locally-unbiased estimators, denoted as ${\rm Cov}(\vr_{\vec{\theta}}^{\otimes \nu}, \mathsf{E}, \Tilde{\vec{\theta}})$ with elements
\begin{equation} \label{eq:covariance}
    [{\rm Cov}(\vr_{\vec{\theta}}^{\otimes \nu}, \mathsf{E}, \Tilde{\vec{\theta}})]_{ij}
    \equiv
    \sum_{\vec{\omega}}
    p (\vec{\omega} | \vec{\theta}) 
    (\theta_i - \Tilde{\theta}_i)
    (\theta_j - \Tilde{\theta}_j).
\end{equation}
This covariance matrix depends on the state $\vr_{\vec{\theta}}^{\otimes \nu}$, the POVM measurement $\mathsf{E}$, and the estimator $\Tilde{\vec{\theta}}$. One major challenge is that, in general, different parameters or their linear combinations cannot be estimated simultaneously with optimal sensitivity. Rather than optimizing the covariance matrix in Eq.~(\ref{eq:covariance}), one focuses on the scalar quantity
\begin{equation} \label{eq:scalarC}
    \mathcal{C}(\vr_{\vec{\theta}}^{\otimes \nu}, \mathsf{E}, \Tilde{\vec{\theta}})
    \equiv \tr[M {\rm Cov}(\vr_{\vec{\theta}}^{\otimes \nu}, \mathsf{E}, \Tilde{\vec{\theta}})],
\end{equation}
where $M$ is a generic weight matrix that expresses the desired estimation trade-off of different parameters or linear combinations of them. Here $M$ is an $m \times m$ real, symmetric, and positive-definite matrix. Note that $\mathcal{C}$ depends on the weight matrix $M$, but we omit this dependency.

For clarity, let us here summarize the notations used in this manuscript: $m$ denotes the number of parameters to be estimated; $n$ is the number of protocol repetitions; and $\nu$ is the number of identical copies of the state. 

\subsection{Hierarchy of metrology bounds}\label{sec:II.B}
The following hierarchy of scalar inequalities holds:
\begin{equation} \label{eq:hierarchy_classical}
    \mathcal{C}(\vr_{\vec{\theta}}^{\otimes \nu}, \mathsf{E}, \Tilde{\vec{\theta}})
    \geq
    \frac{\mathcal{C}_{\rm CCR} (\vr_{\vec{\theta}}^{\otimes \nu}, \mathsf{E})}{n}
    \geq
    \frac{\mathcal{C}_{\rm MI} (\vr_{\vec{\theta}}^{\otimes \nu})}{n}.
\end{equation}
Note that $\mathcal{C}_{\rm CCR}$ and $\mathcal{C}_{\rm MI}$ depend on the weight matrix $M$, but we omit that dependency. Here,
\begin{equation}
    \mathcal{C}_{\rm CCR} (\vr_{\vec{\theta}}^{\otimes \nu}, \mathsf{E})
    \equiv \tr[M F_C^{-1}(\vr_{\vec{\theta}}^{\otimes \nu}, \mathsf{E})]
\end{equation}
is called the classical Cramér-Rao (CCR) bound~\cite{cramer1999mathematical,rao1945information} and $F_C(\vr_{\vec{\theta}}^{\otimes \nu}, \mathsf{E})$ denotes the classical Fisher information matrix (CFIM). Its elements are
\begin{equation} \label{eq:CFIM_form}
    [F_C(\vr_{\vec{\theta}}^{\otimes \nu}, \mathsf{E})]_{ij}
    = \sum_{\omega} \frac{1}{p(\omega | \vec{\theta})} [\partial_i p(\omega | \vec{\theta})] [\partial_j p(\omega | \vec{\theta})],
\end{equation}
for the outcome $\omega$ occurring with probability $p(\omega | \vec{\theta}) = \tr(\vr_{\vec{\theta}}^{\otimes \nu} E_\omega)$. The left-hand inequality in Eq.~(\ref{eq:hierarchy_classical}) can be saturated by appropriate locally-unbiased estimators such as the maximum-likelihood in the limit $n \to \infty$ (see Ref.~\cite{kay1993fundamentals,lehmann1998theory}).

Furthermore, 
\begin{equation}
    \mathcal{C}_{\rm MI} (\vr_{\vec{\theta}}^{\otimes \nu})
    \equiv
    \min_{\mathsf{E}} \mathcal{C}_{\rm CCR} (\vr_{\vec{\theta}}^{\otimes \nu}, \mathsf{E})
\end{equation}
is called the most informative (MI) bound. Clearly, $\mathcal{C}_{\rm MI} (\vr_{\vec{\theta}}^{\otimes \nu}) = n \min_{\Tilde{\vec{\theta}}, \mathsf{E}} \mathcal{C}(\vr_{\vec{\theta}}^{\otimes \nu}, \mathsf{E}, \Tilde{\vec{\theta}})$. The MI bound represents the ultimate attainable bound in multiparameter quantum metrology, minimized over all possible estimators and all POVMs, for the given $\vr_{\vec{\theta}}^{\otimes \nu}$.

The MI bound is not additive (see below for details), which justifies considering the collective measurement strategy. Such an approach can outperform local/separable measurements by achieving higher estimation precision and enabling tighter bounds on achievable sensitivity. On the other hand, the direct computation of the MI bound is highly nontrivial in general, since the optimal POVM measurement can itself depend on the parameters. Developing effective algorithms and practical procedures for evaluating the MI bound is highly desirable, but this remains a central challenge in multiparameter quantum estimation.

A viable approach to bypass this challenge is to introduce lower bounds to the MI bound and to consider conditions under which these bounds can be saturated. Such bounds are independent of POVM operators and only depend on the state; we refer to them as \textit{quantum bounds}. In fact, the hierarchy in Eq.~(\ref{eq:hierarchy_classical}) is further extended to a sequence of quantum bounds:
\begin{equation} \label{eq:hierarchy_quantum}
    \mathcal{C}_{\rm MI} (\vr_{\vec{\theta}}^{\otimes \nu})
    \geq
    \mathcal{C}_{\rm NH} (\vr_{\vec{\theta}}^{\otimes \nu})
    \geq
    \mathcal{C}_{\rm H} (\vr_{\vec{\theta}}^{\otimes \nu})
    \geq
    \mathcal{C}_{\rm QCR} (\vr_{\vec{\theta}}^{\otimes \nu}),
\end{equation}
which holds for arbitrary weight matrices $M$. Here, $\mathcal{C}_{\rm NH}$ is the Nagaoka-Hayashi bound (originally proposed for two parameters in Refs.~\cite{nagaoka1989ieice,nagaoka2005generalization} reprinted in~\cite{hayashi2005asymptotic} and extended to more parameters in Refs.~\cite{hayashi1999simultaneous,conlon2021efficient}), $\mathcal{C}_{\rm H}$ is the Holevo bound~\cite{holevo2011probabilistic} (sometimes also called the Holevo-Nagaoka bound~\cite{hayashi2024alexander}), and $\mathcal{C}_{\rm QCR}$ is the quantum Cramér-Rao (QCR) bound, which was originally introduced by Helstrom~\cite{helstrom1969quantum,helstromBOOK1976}. These bounds will be defined more explicitly later.

We note that a metrology bound $\mathfrak{c}(\vr_{\vec{\theta}})$ is called \textit{additive} if $\mathfrak{c}(\vr_{\vec{\theta}}^{\otimes \nu}) = (1/\nu) \mathfrak{c}(\vr_{\vec{\theta}})$ for any integer $\nu$, while it is called \textit{subadditive} if $\mathfrak{c}(\vr_{\vec{\theta}}^{\otimes \nu}) \leq (1/\nu) \mathfrak{c}(\vr_{\vec{\theta}})$. It is known that the Holevo bound $\mathcal{C}_{\rm H}(\vr_{\vec{\theta}})$ and the QCR bound $\mathcal{C}_{\rm QCR}(\vr_{\vec{\theta}})$ are additive (see Ref.~\cite{hayashi2008asymptotic}). On the other hand, the MI bound $\mathcal{C}_{\rm MI} (\vr_{\vec{\theta}})$ and the Nagaoka-Hayashi bound $\mathcal{C}_{\rm NH} (\vr_{\vec{\theta}})$ are subadditive in general~\cite{conlon2022gap}.

\subsubsection{Nagaoka-Hayashi and Holevo bounds}
Let $\mathsf{X} = \{X_i\}_{i=1}^m$ be a collection of Hermitian operators $X_i$ and $\mathcal{X}_{\vec{\theta}} =\{ \mathsf{X}| \tr[(\partial_i \vr_{\vec{\theta}}) X_j] = \delta_{ij} \}$ be the set to which $\mathsf{X}$ belongs. The Nagaoka-Hayashi bound and Holevo bound are respectively given by 
\begin{subequations}
\begin{align}
    \label{eq:def:NHbnound}
    \mathcal{C}_{\rm NH} (\vr_{\vec{\theta}})
    &\equiv \min_{\mathsf{X} \in \mathcal{X}_{\vec{\theta}} }
    \{ \tr[M {\rm Re}(\mathsf{Z}_{\mathsf{X}})] + \zeta_{\rm NH} (\mathsf{X}|M) \},
    \\
    \label{eq:def:Hbnound}
    \mathcal{C}_{\rm H} (\vr_{\vec{\theta}})
    &\equiv  \min_{\mathsf{X} \in \mathcal{X}_{\vec{\theta}} }
    \{ \tr[M {\rm Re}(\mathsf{Z}_{\mathsf{X}})] + \zeta_{\rm H} (\mathsf{X}|M) \},
\end{align}
\end{subequations}
where $\mathsf{Z}_{\mathsf{X}}$ is the Hermitian matrix that depends on $\mathsf{X}$, with elements $[\mathsf{Z}_{\mathsf{X}}]_{ij} = \tr(\vr_{\vec{\theta}} X_i X_j)$. Here, $\zeta_{\rm NH}$ is defined as $\zeta_{\rm NH} (\mathsf{X}|M) \equiv \min_{V} \tr(V)$ subject to $V \geq 0$ and $\mathcal{A}(V) + \sqrt{M \otimes \vr_{\vec{\theta}}} \mathcal{A}(\mathsf{X} \mathsf{X}^\trans) \sqrt{M \otimes \vr_{\vec{\theta}}} = 0$, where $\mathcal{A} (V)$ represents the anti-symmetrized version of a matrix $V$ such that $[\mathcal{A} (V)]_{ij} = (V_{ij} - V_{ji})/2$~\cite{hayashi1999development,conlon2022gap}. Also, $\zeta_{\rm H}$ is defined as $\zeta_{\rm H} (\mathsf{X}|M) = \tr[\lvert \sqrt{M} {\rm Im}(\mathsf{Z}_{\mathsf{X}}) \sqrt{M}  \rvert]$~\cite{nagaoka1989ieice,hayashi2005asymptotic}, where $\tr[\lvert X \rvert]$ is the norm of $\lvert X \rvert = \sqrt{X^\dagger X}$ for a matrix $X$. The bounds $\mathcal{C}_{\rm NH} (\vr_{\vec{\theta}})$ and $\mathcal{C}_{\rm H} (\vr_{\vec{\theta}})$ are obtained by nontrivial constrained minimizations over a set of operators. This makes exact computation very challenging, though numerical evaluation has been investigated based on semidefinite programming~\cite{albarelli2019evaluating,conlon2021efficient} and conic linear programming~\cite{hayashi2023tight}.

In the two-parameter case ($m=2$), the bounds $\mathcal{C}_{\rm NH} (\vr_{\vec{\theta}})$ in Eq.~(\ref{eq:def:NHbnound}) and $\mathcal{C}_{\rm H} (\vr_{\vec{\theta}})$ in Eq.~(\ref{eq:def:Hbnound}) are further simplified, since $\zeta_{\rm NH}(\mathsf{X}|M)$ and $\zeta_{\rm H}(\mathsf{X}|M)$ with $\mathsf{X} = \{X_1, X_2\}$ for a weight matrix $M$ are rewritten as
\begin{subequations}
    \begin{align}
        \zeta_{\rm NH}
        &= f_M \, \tr[\lvert \sqrt{\vr_{\vec{\theta}}} (X_1 X_2 - X_2 X_1) \sqrt{\vr_{\vec{\theta}}} \rvert],
        \\
        \zeta_{\rm H}
        &= f_M \, \lvert \tr[\vr_{\vec{\theta}} (X_1 X_2 - X_2 X_1)] \rvert,
    \end{align}
\end{subequations}
where $f_M = 2 \sqrt{\det{(M)}}$. We remark that the inequality $\tr[\lvert X \rvert] \geq \lvert \tr[X] \rvert$ for a matrix $X$ directly leads to $\zeta_{\rm NH} \geq \zeta_{\rm H}$~\cite{nagaoka1989ieice,hayashi2005asymptotic}.

\subsubsection{QCR bound}
The QCR bound is defined as
\begin{equation}
    \mathcal{C}_{\rm QCR} (\vr_{\vec{\theta}})
    \equiv \tr[M F_Q^{-1} (\vr_{\vec{\theta}})],
\end{equation}
where $F_Q(\vr_{\vec{\theta}})$ is the quantum Fisher information matrix (QFIM) with the elements 
\begin{equation} \label{eq:QFImatrix}
    [F_Q(\vr_{\vec{\theta}})]_{ij} = \frac{1}{2} \tr [\vr_{\vec{\theta}} (L_i L_j + L_j L_i) ].
\end{equation}
Here, $L_i \equiv L_i (\vr_{\vec{\theta}})$ is the symmetric logarithmic derivative (SLD) operator defined as
\begin{equation} \label{eq:SLD_definition_derivative}
    \partial_i \vr_{\vec{\theta}} = \frac{1}{2} ( L_i \vr_{\vec{\theta}} + \vr_{\vec{\theta}} L_i ),
\end{equation}
with $\tr(\vr_{\vec{\theta}} L_i) = 0$~\cite{helstrom1969quantum, helstromBOOK1976}. More details about the SLD operator are given in Section~\ref{subsubsec:remarks}.

While so far we have referred to the scalar quantities $\mathcal{C}_{\rm CCR} (\vr_{\vec{\theta}}, \mathsf{E})$ and $\mathcal{C}_{\rm QCR} (\vr_{\vec{\theta}})$ as the CCR and QCR bounds, respectively, they also hold in matrix form~\cite{helstrom1969quantum,helstromBOOK1976}:
\begin{equation}
    \label{eq:chainCCRQCR}
    {\rm Cov}(\vr_{\vec{\theta}}, \mathsf{E}, \Tilde{\vec{\theta}})
    \geq
    \frac{F_C^{-1}(\vr_{\vec{\theta}}, \mathsf{E})}{n}
    \geq
    \frac{F_Q^{-1}(\vr_{\vec{\theta}})}{n}.
\end{equation}
As we will discuss in Appendix~\ref{ap:derivation_FC<=FQ}, Eq.~(\ref{eq:chainCCRQCR}) can lead to
\begin{equation}
    \label{eq:FC<=FQ}
    F_C(\vr_{\vec{\theta}}, \mathsf{E})
    \leq F_Q(\vr_{\vec{\theta}}),
\end{equation}
which holds for all POVMs and quantum states (see, e.g., Refs.~\cite{pezze2017optimal,yang2019optimal}), with equality saturated for all states in the single-parameter case $m=1$~\cite{braunstein1994statistical}.

\subsubsection{Tightness} \label{SubSec.tightness}
For two metrology bounds $\mathfrak{c}_1$ and $\mathfrak{c}_2$, an inequality $\mathfrak{c}_1\leq\mathfrak{c}_2$ is called \textit{tight} if there is some case such that $\mathfrak{c}_1=\mathfrak{c}_2$ holds. Let us make four remarks on the tightness in Eq.~(\ref{eq:hierarchy_quantum}) for an arbitrary weight matrix $M$. 

First, the inequality $\mathcal{C}_{\rm MI} \geq \mathcal{C}_{\rm H}$ is always tight in the asymptotic limit by collective measurements over infinitely many copies of the state, i.e., $\mathcal{C}_{\rm MI} (\vr_{\vec{\theta}}^{\otimes \nu}) = \mathcal{C}_{\rm H} (\vr_{\vec{\theta}}^{\otimes \nu})$ for $\nu \to \infty$~\cite{nagaoka1989ieice,hayashi2005asymptotic,kahn2009local,yamagata2013quantum,yang2019attaining}, but cannot be tight by separable measurements in general (see Ref.~\cite{conlon2021efficient}).

Next, in the two-parameter case, the inequality $\mathcal{C}_{\rm NH} \geq \mathcal{C}_{\rm H}$ is always tight in the asymptotic limit, i.e., $\mathcal{C}_{\rm NH} (\vr_{\vec{\theta}}^{\otimes \nu}) = \mathcal{C}_{\rm H} (\vr_{\vec{\theta}}^{\otimes \nu})$ for $\nu \to \infty$~\cite{conlon2022gap}. Third, in certain metrology settings, the inequality $\mathcal{C}_{\rm MI} \geq \mathcal{C}_{\rm NH}$ can be tight by separable measurements on a finite number of copies~\cite{conlon2021efficient}, but there are also settings such that the inequality $\mathcal{C}_{\rm MI} \geq \mathcal{C}_{\rm NH}$ cannot be tight by separable measurements. This means that the inequality $\mathcal{C}_{\rm MI} \geq \mathcal{C}_{\rm NH}$ is not always tight~\cite{hayashi2023tight,conlon2025role}.

Finally, let us introduce the quantity
\begin{equation} \label{eq:mathcalQ}
    \mathcal{Q}(\vr_{\vec{\theta}}^{\otimes \nu})
    \equiv \mathcal{C}_{\rm MI} (\vr_{\vec{\theta}}^{\otimes \nu}) - \mathcal{C}_{\rm QCR} (\vr_{\vec{\theta}}^{\otimes \nu}).
\end{equation}
It has been shown that $\mathcal{Q}(\vr_{\vec{\theta}}) = 0$ holds if and only if $\mathcal{Q}(\vr_{\vec{\theta}}^{\otimes \nu}) = 0$ holds for any finite $\nu$~\cite{conlon2022gap}.
This result is generally indicated as the gap persistence theorem~\cite{conlon2022gap}. We note that $\mathcal{Q}(\vr_{\vec{\theta}}) = 0$ holds if and only if there exist a POVM $\mathsf{E}$ and an estimator $\Tilde{\vec{\theta}}$ that achieve the equality ${\rm Cov}(\vr_{\vec{\theta}}, \mathsf{E}, \Tilde{\vec{\theta}}) = (1/n) F_Q^{-1} (\vr_{\vec{\theta}})$, for sufficiently large $n$.

\subsection{Commutativity conditions for the saturation of quantum bounds}\label{sec:II.C}
From the perspective of attaining the MI bound, it is essential to understand conditions for the saturation of the different quantum bounds in Eq.~(\ref{eq:hierarchy_quantum}) and their relations. Here we first make some remarks on SLD operators. Then we outline relevant saturation conditions associated with the chain in Eq.~(\ref{eq:hierarchy_quantum}).

\subsubsection{Remarks}\label{subsubsec:remarks}
The saturation conditions of quantum bounds have been completely determined for full-rank and pure states (see below for details). For mixed rank-deficient (neither pure nor full-rank) states, instead, the situation becomes substantially more involved. The main difficulty comes from the structure of the SLD operators.

Specifically, we denote as $\Pi_{\vr_{\vec{\theta}}}$ and $\Pi_{\vr_{\vec{\theta}}}^\perp$ the projectors onto the support and kernel spaces of $\vr_{\vec{\theta}}$, respectively: $\Pi_{\vr_{\vec{\theta}}} \vr_{\vec{\theta}} \Pi_{\vr_{\vec{\theta}}} = \vr_{\vec{\theta}}$ and $\Pi_{\vr_{\vec{\theta}}}^\perp \vr_{\vec{\theta}} \Pi_{\vr_{\vec{\theta}}}^\perp = 0$. Any SLD operator $L_i$ can be decomposed into contributions acting within and across these subspaces. In block matrix form, $L_i$ is written as
\begin{equation} \label{eq:SLDcomponents}
    L_i
    = \begin{bmatrix}
    L_i^{{\rm ss}} & L_i^{{\rm sk}} \\
    L_i^{{\rm ks}} & L_i^{{\rm kk}}
    \end{bmatrix}.
\end{equation}
Here, $L_i^{{\rm ss}} = \Pi_{\vr_{\vec{\theta}}} L_i^{{\rm ss}} \Pi_{\vr_{\vec{\theta}}}$ acts within the support, $L_i^{{\rm sk}} = \Pi_{\vr_{\vec{\theta}}} L_i^{{\rm sk}} \Pi_{\vr_{\vec{\theta}}}^\perp$ maps from the kernel to the support, $L_i^{{\rm ks}} = \Pi_{\vr_{\vec{\theta}}}^\perp L_i^{{\rm ks}} \Pi_{\vr_{\vec{\theta}}}$ maps from the support to the kernel with $(L_i^{{\rm ks}})^\dagger = L_i^{{\rm sk}}$, and $L_i^{{\rm kk}} = \Pi_{\vr_{\vec{\theta}}}^\perp L_i^{{\rm kk}} \Pi_{\vr_{\vec{\theta}}}^\perp$ acts within the kernel. These terms can be expressed in terms of support and kernel states with appropriate coefficients.

It is essential to clarify that $L_i^{{\rm kk}}$ is not uniquely determined by the relation in Eq.~(\ref{eq:SLD_definition_derivative}). In fact, $L_i^{{\rm kk}} \vr_{\vec{\theta}} = \vr_{\vec{\theta}} L_i^{{\rm kk}} =  0$, meaning that for any Hermitian operator $K$ whose support is ${\rm ker} (\vr_{\vec{\theta}})$, an operator $L_i + K$ is also qualified as SLD according to Eq.~(\ref{eq:SLD_definition_derivative}). That is, for a rank-deficient state, the SLD operator belongs to an equivalence class modulo Hermitian operators supported on the kernel space. In general, any Hermitian operator $\mathfrak{L}_i$ is called the SLD operator of $\vr_{\vec{\theta}}$ if and only if it satisfies the condition
\begin{equation} \label{eq:def_SLD_another}
     \mathfrak{L}_i \vr_{\vec{\theta}} = L_i \vr_{\vec{\theta}},
\end{equation}
where $L_i$ satisfies Eq.~(\ref{eq:SLD_definition_derivative}).

As we will explain later, computing commutators of $L_i$ and $L_j$ is relevant for the saturation conditions. Due to the orthogonality between the support and kernel subspaces, when computing the product $L_i L_j$ (eventually also multiplying by $\Pi_{\vr_{\vec{\theta}}}$), many components in Eq.~(\ref{eq:SLDcomponents}) vanish in the computation. For instance, in the weak commutativity and partial commutativity condition given in Eqs.~(\ref{eq:def_W}, \ref{eq:def_P}), only the contributions $L_i^{{\rm ss}} L_j^{{\rm ss}}$ and $L_i^{{\rm sk}} L_j^{{\rm ks}}$ remain, independently of $L_i^{{\rm kk}}$. In contrast, the one-sided commutativity condition given in Eq.~(\ref{eq:def_O}) also includes the additional terms $L_i^{{\rm ks}} L_j^{{\rm ss}}$ and $L_i^{{\rm kk}} L_j^{{\rm ks}}$, while the strong and extended commutativity conditions in Eqs.~(\ref{eq:def_S}, \ref{eq:def_Sprime}) have further terms.

The appearance of terms that depend explicitly on the kernel structure highlights the subtle role of rank deficiency and complicates the formulation of commutativity conditions in this case. More explicit differences between these commutativity conditions will be discussed in Observation~\ref{ob:GE:UPUandUMUandUSU} below. 

Finally, let us comment on pure states. The SLD operator has a compact form
\begin{equation} \label{SLDpurestates}
    L_i (\ket{\psi_{\vec{\theta}}})
    = 2 ( \ket{\partial_i \psi_{\vec{\theta}}}\! \bra{\psi_{\vec{\theta}}} + \ket{\psi_{\vec{\theta}}}\! \bra{\partial_i \psi_{\vec{\theta}}})
    + L_i^{{\rm kk}},
\end{equation}
which is obtained from Eq.~(\ref{eq:SLD_definition_derivative}) and $\vr_{\vec{\theta}}^2 = \vr_{\vec{\theta}}$. In this case, the analysis of the commutativity conditions is greatly simplified.

\subsubsection{Weak commutativity}
The weak commutativity (WC) condition is given by $W(\vr_{\vec{\theta}}) = 0$, where
\begin{equation}\label{eq:def_W}
    [W(\vr_{\vec{\theta}})]_{ij}
    = \tr [\vr_{\vec{\theta}} (L_i L_j - L_j L_i) ].
\end{equation}
The WC condition has been introduced in Ref.~\cite{matsumoto2002new} and shown to be both necessary and sufficient for the equivalence between the MI bound and the QCR bound for arbitrary weight matrices $M$ when a quantum state is pure, i.e., $\mathcal{Q}(\ket{\psi_{\vec{\theta}}}) = 0$ if and only if $W(\ket{\psi_{\vec{\theta}}}) = 0$. In addition, the WC condition is both necessary and sufficient for the equivalence of the Holevo bound and the QCR bound for arbitrary $M$ at the single-copy (not necessarily pure) state $\vr_{\vec{\theta}}$, i.e., $\mathcal{C}_{\rm H} (\vr_{\vec{\theta}}) = \mathcal{C}_{\rm QCR} (\vr_{\vec{\theta}})$ holds if and only if $W(\vr_{\vec{\theta}}) = 0$~\cite{ragy2016compatibility}. 

More specifically, the gap between the Holevo bound and the QCR bound is characterized as follows~\cite{carollo2019quantumness,tsang2020quantum}:
\begin{equation} \label{eq:gap_HQCR}
    \mathcal{C}_{\rm QCR} (\vr_{\vec{\theta}}) \leq 
    \mathcal{C}_{\rm H} (\vr_{\vec{\theta}}) \leq [1+ \mathcal{R}(\vr_{\vec{\theta}})]
    \mathcal{C}_{\rm QCR} (\vr_{\vec{\theta}}),
\end{equation}
where $\mathcal{R}(\vr_{\vec{\theta}}) \equiv (1/2) \lVert F_Q^{-1}(\vr_{\vec{\theta}}) W(\vr_{\vec{\theta}}) \rVert_{\infty} \in [0,1]$ and $\lVert X \rVert_{\infty}$ denotes the largest absolute value of eigenvalue of a matrix $X$. The quantity $\mathcal{R}(\vr_{\vec{\theta}})$ has been used to quantify the incompatibility in multiparameter quantum metrology~\cite{carollo2019quantumness} (see also Refs.~\cite{razavian2020quantumness,belliardo2021incompatibility}). The condition $\mathcal{R}(\vr_{\vec{\theta}}) = 0$ is equivalent to $W(\vr_{\vec{\theta}}) = 0$. In fact, if $\mathcal{R}(\vr_{\vec{\theta}}) = 0$, then $\mathcal{C}_{\rm H} (\vr_{\vec{\theta}}) = \mathcal{C}_{\rm QCR}$ and the WC condition holds; conversely, if $W(\vr_{\vec{\theta}}) = 0$ then $\mathcal{C}_{\rm H} (\vr_{\vec{\theta}}) = \mathcal{C}_{\rm QCR}$ and $\mathcal{R}(\vr_{\vec{\theta}}) = 0$.

From a geometric perspective, the WC condition corresponds to the vanishing of the mean Uhlmann curvature as a generalization of the Berry curvature to mixed states~\cite{uhlmann1986parallel,carollo2020geometry}. In other words, the WC condition is associated with the fact that the imaginary part of the quantum geometric tensor $G(\vr_{\vec{\theta}})$ vanishes: $G(\vr_{\vec{\theta}}) = F_Q(\vr_{\vec{\theta}}) + (1/2) W(\vr_{\vec{\theta}})$, where the real part of $G(\vr_{\vec{\theta}})$ is the QFIM. In particular, the quantity $\mathcal{R}(\vr_{\vec{\theta}})$ introduced above has recently been discussed in a geometric framework through the notion of semi-classical geometric tensor as a measurement-dependent counterpart of the quantum geometric tensor~\cite{imai2026semiclassical}.

Finally, we comment on the gap between the MI bound and the QCR bound. In the single-copy case ($\nu=1$), the WC condition is necessary to have $\mathcal{Q}(\vr_{\vec{\theta}}) = 0$ in Eq.~(\ref{eq:mathcalQ}), but it is not sufficient because the equivalence between the MI bound and the Holevo bound does not hold in general. This lack of sufficiency holds for any finite $\nu$ by the gap persistence theorem~\cite{conlon2022gap}. Only in the asymptotic limit ($\nu \to \infty$), the WC condition becomes both necessary and sufficient for the equivalence between the MI bound and the QCR bound. This follows from the two facts: one is that the Holevo bound coincides with the MI bound in this limit (see discussion in Section~\ref{SubSec.tightness}); the other is that the WC condition is both necessary and sufficient for the equivalence between the QCR bound and the Holevo bound, as mentioned above.

\subsubsection{Partial commutativity}
The partial commutativity (PC) condition is given by $P(\vr_{\vec{\theta}}) = 0$, where
\begin{equation}\label{eq:def_P}
    [P(\vr_{\vec{\theta}})]_{ij}
    = \Pi_{\vr_{\vec{\theta}}} (L_i L_j - L_j L_i) \Pi_{\vr_{\vec{\theta}}}.
\end{equation}
Here, $\Pi_{\vr_{\vec{\theta}}}$ is the projector on the support of the state $\vr_{\vec{\theta}}$, defined as $\Pi_{\vr_{\vec{\theta}}} = \sum_{k} \ket{\psi_k(\vec{\theta})}\! \bra{\psi_k(\vec{\theta})}$ for $\ket{\psi_k(\vec{\theta})}$ being the eigenstate of $\vr_{\vec{\theta}}$ with nonzero eigenvalue. The PC condition is both necessary and sufficient for the equivalence between the Nagaoka-Hayashi bound and the QCR bound for arbitrary $M$ at the single-copy state~\cite{conlon2022gap}, i.e., $P(\vr_{\vec{\theta}}) = 0$ holds if and only if $\mathcal{C}_{\rm NH} (\vr_{\vec{\theta}}) = \mathcal{C}_{\rm QCR} (\vr_{\vec{\theta}})$. 

In the single-copy case, the PC condition is necessary to have $\mathcal{Q}(\vr_{\vec{\theta}}) = 0$~\cite{yang2019optimal}. However, the PC condition cannot be sufficient because the equivalence between the MI bound and the Nagaoka-Hayashi bound does not hold in general~\cite{hayashi2023tight,conlon2025role}. Very recently, Ref.~\cite{yang2026geometric} has arrived at this conclusion from a geometric criterion based on the so-called hollowization condition. We note that whether the lack of this sufficiency remains true or not for any finite $\nu>1$ remains an open question~\cite{conlon2022gap}. Similarly to the WC condition, the PC condition becomes both necessary and sufficient for the equivalence between the MI bound and the QCR bound only in the asymptotic limit ($\nu \to \infty$). This follows from the two facts: one is that in this limit, the Nagaoka-Hayashi bound coincides with the Holevo bound, and the Holevo bound coincides with the MI bound (see discussion in Section~\ref{SubSec.tightness}); the other is that the PC condition is both necessary and sufficient for the equivalence between the QCR bound and the Nagaoka-Hayashi bound, as mentioned above.

The connection between the PC and WC conditions is clarified as follows~\cite{chen2022information,chen2022incompatibility}:
\begin{equation}
    \label{eq:equivalence_PC_WC}
    \lim_{\nu \to \infty}
    \frac{1}{\nu}
    [\mathfrak{p} (\vr_{\vec{\theta}}^{\otimes \nu})]_{ij}
    =
    \left \lvert
    [W(\vr_{\vec{\theta}})]_{ij}
    \right \rvert,
\end{equation}
where $[W(\vr_{\vec{\theta}})]_{ij}$ is given in Eq.~(\ref{eq:def_W}) and
\begin{equation}
    [\mathfrak{p} (\vr_{\vec{\theta}}^{\otimes \nu})]_{ij}
    \!=\! \tr \left[
    \left \lvert
    \sqrt{\vr_{\vec{\theta}}^{\otimes \nu}}
    \left( L_{i \nu} L_{j \nu} \!-\! L_{j \nu} L_{i \nu} \right)
    \sqrt{\vr_{\vec{\theta}}^{\otimes \nu}}
    \right \rvert
    \right].
\end{equation}
Here we denoted $\tr[\lvert X \rvert]$ as the norm of $\lvert X \rvert = \sqrt{X^\dagger X}$ for a matrix $X$, and we define $L_{i \nu} \equiv L_i (\vr_{\vec{\theta}}^{\otimes \nu})$ as the SLD operator of the $\nu$-copy state $\vr_{\vec{\theta}}^{\otimes \nu}$. For $\nu = 1$, one can see that $[\mathfrak{p} (\vr_{\vec{\theta}})]_{ij} = 0$ holds if and only if $[P(\vr_{\vec{\theta}})]_{ij} = 0$ holds, as in Eq.~(\ref{eq:def_P}). This equivalence follows from the linear independence of the eigenstates corresponding to distinct positive eigenvalues.

\subsubsection{Strong commutativity}
The strong commutativity (SC) condition is formulated as follows: There exists at least one choice of $\{L_i^{{\rm kk}}\}$ acting within the kernel space of $\vr_{\vec{\theta}}$ in Eq.~(\ref{eq:SLDcomponents}) such that $S(\vr_{\vec{\theta}}) = 0$, where
\begin{equation}\label{eq:def_S}
    [S(\vr_{\vec{\theta}})]_{ij}
    = L_i L_j - L_j L_i.
\end{equation}
The SC condition is sufficient for the equivalence between the MI bound and the QCR bound for arbitrary $M$ at the single-copy state~\cite{nagaoka1987OnFisher,amari2000methods}, i.e., if there exist $\{L_i^{{\rm kk}}\}$ such that $S(\vr_{\vec{\theta}}) = 0$, then $\mathcal{Q}(\vr_{\vec{\theta}}) = 0$. In this case, the optimal POVM is given by the simultaneous spectral decomposition of each SLD operator~\cite{nagaoka1987OnFisher,amari2000methods}.

We emphasize that the choice of $\{L_i^{{\rm kk}}\}$ in the SC condition, understood as the existence of commuting SLD operators, is preserved under arbitrary kernel additions of operators $K$ satisfying $K = \Pi_{\vr_{\vec{\theta}}}^\perp K \Pi_{\vr_{\vec{\theta}}}^\perp$. In contrast, the commutator $S(\vr_{\vec{\theta}})$ itself is not preserved under such additions. This fact does not contradict the physical insight that the saturation condition for $\mathcal{Q}(\vr_{\vec{\theta}}) = 0$ is invariant under any changes of the kernel. In practice, a common and straightforward way to verify the SC condition is to choose $L_i^{{\rm kk}} = 0$. On the other hand, it should be noted that even if such a choice yields $S(\vr_{\vec{\theta}}) \neq 0$, there might remain a possibility that $S(\vr_{\vec{\theta}}) = 0$ holds for other choices of $\{L_i^{{\rm kk}}\}$~\cite{conlon2025role}.

\subsubsection{One-sided commutativity}
The one-sided commutativity (OC) condition is formulated as follows: There exists at least one choice of $\{L_i^{{\rm kk}}\}$ acting within the kernel space of $\vr_{\vec{\theta}}$ in Eq.~(\ref{eq:SLDcomponents}) such that $O(\vr_{\vec{\theta}}) = 0$, where
\begin{equation}\label{eq:def_O}
    [O(\vr_{\vec{\theta}})]_{ij}
    = (L_i L_j - L_j L_i) \Pi_{\vr_{\vec{\theta}}}.
\end{equation}
Then one can obtain the following:
\begin{observation}\label{ob:OCisnecessary}
    The OC condition is necessary for the equivalence between the MI bound and the QCR bound for arbitrary $M$ at the single-copy state, i.e., if $\mathcal{Q}(\vr_{\vec{\theta}}) = 0$, then there exist $\{L_i^{{\rm kk}}\}$ such that $O(\vr_{\vec{\theta}}) = 0$ holds.
\end{observation}
We have several remarks. First, the statement of Observation~\ref{ob:OCisnecessary} is motivated by Ref.~\cite{matsumoto2005geometricalPhD}, and its proof is postponed below. Second, similar to the SC condition, the OC condition is invariant under arbitrary kernel additions, but this alone does not provide a recipe for choosing $\{L_i^{{\rm kk}}\}$. In fact, it has been reported that, even when $\mathcal{Q}(\vr_{\vec{\theta}}) = 0$ holds, simply choosing $L_i^{{\rm kk}} = 0$ does not guarantee that $O(\vr_{\vec{\theta}}) = 0$, see Example~B in Ref.~\cite{conlon2025role}. Third, it might be natural to ask whether the OC condition also becomes sufficient to have $\mathcal{Q}(\vr_{\vec{\theta}}) = 0$. Although this sufficiency is true for pure and full-rank states as discussed in Section~\ref{sec:II.D} below, this cannot be true for rank-deficient states, see Example~E in Ref.~\cite{conlon2025role}. Finally, we note that the condition $[O(\vr_{\vec{\theta}})]_{ij} = 0$ holds if and only if $(\partial_i L_j - \partial_j L_i ) \vr_{\vec{\theta}} = 0$ holds, due to the identity relation $\partial_i \partial_j \vr_{\vec{\theta}} - \partial_j \partial_i \vr_{\vec{\theta}} = 0$.

\subsubsection{Extended commutativity on an enlarged Hilbert space}\label{subsubsec:ECcondition}
Let $\mathcal{H}$ be the Hilbert space where the state $\vr_{\vec{\theta}}$ is defined, and $\mathcal{H}^\prime$ be an extended Hilbert space with $\mathcal{H}^\prime \supseteq \mathcal{H}$. Let $\mathcal{V}:\mathcal{H} \to \mathcal{H}^\prime$ be an isometry operator, so that $\mathcal{V}^\dagger \mathcal{V}=\eins$ on $\mathcal{H}$, and $\Pi_{\mathcal{V}}=\mathcal{V}\mathcal{V}^\dagger$ be the orthogonal projector onto $\mathcal{V}(\mathcal{H}) \subseteq \mathcal{H}^\prime$. Here, $\mathcal{V}(\mathcal{H})$ is the image (range) of the operator $\mathcal{V}$, i.e., the set of all outputs of $\mathcal{V}$ when applied to elements of $\mathcal{H}$. The state $\vr_{\vec{\theta}}^\prime = \mathcal{V} \vr_{\vec{\theta}} \mathcal{V}^\dagger \in \mathcal{V}(\mathcal{H})$ is thus the isometric extension of $\vr_{\vec{\theta}}$. In analogy with Eq.~(\ref{eq:def_SLD_another}), we say that a Hermitian operator $\mathfrak{L}_i^\prime \in \mathcal{H}^\prime$ is a SLD operator in the extended space if and only if
\begin{equation} \label{eq:def_SLD_another_extension}
    \mathfrak{L}_i^\prime \vr_{\vec{\theta}}^\prime = \mathcal{V} L_i \mathcal{V}^\dagger \vr_{\vec{\theta}}^\prime.
\end{equation}

The extended commutativity (EC) condition on an enlarged Hilbert space~\cite{suzuki2020quantum,conlon2022gap} is formulated as follows: There exist an isometry $\mathcal{V}$ and a set of SLD operators $\{ \mathfrak{L}_i^\prime \}$ on the extended Hilbert space $\mathcal{H}^\prime$, satisfying Eq.~(\ref{eq:def_SLD_another_extension}), such that $\mathfrak{E} (\vr_{\vec{\theta}}^\prime) = 0$, where 
\begin{equation}\label{eq:def_Sprime}
    [\mathfrak{E} (\vr_{\vec{\theta}}^\prime)]_{ij}
    = \mathfrak{L}_i^\prime \mathfrak{L}_j^\prime - \mathfrak{L}_j^\prime \mathfrak{L}_i^\prime.
\end{equation}
The EC condition is both necessary and sufficient for the equivalence between the MI bound and the QCR bound for arbitrary $M$ at the single-copy state, i.e., $\mathcal{Q}(\vr_{\vec{\theta}}) = 0$ if and only if there exist $\mathcal{V}$ and $\{\mathfrak{L}_i^\prime \}$ such that $\mathfrak{E} (\vr_{\vec{\theta}}^\prime) = 0$~\cite{suzuki2020quantum,conlon2022gap}. For clarity and completeness of this manuscript, we will provide a comprehensive proof in Appendix~\ref{ap:commutativity_extension}. 

Accordingly, one can prove Observation~\ref{ob:OCisnecessary}:
\begin{proof}[Proof of Observation~\ref{ob:OCisnecessary}]
    Let us begin by decomposing the extended Hilbert space $\mathcal{H}^\prime$ as
    \begin{equation}
        \label{eq:H'space}
        \mathcal{H}^\prime
        = \mathcal{V}(\mathcal{H}) \oplus (\mathcal{V}(\mathcal{H}))^\perp
        = \mathcal{V}(\mathcal{H}_{\rm S}) \oplus \mathcal{V}(\mathcal{H}_{\rm K}) \oplus \mathcal{H}_{\rm A},
    \end{equation}
    where $\mathcal{H}_{\rm S} = {\rm supp} (\vr_{\vec{\theta}})$ and $\mathcal{H}_{\rm K} = {\rm ker} (\vr_{\vec{\theta}})$ respectively denote the support and kernel subspaces of $\vr_{\vec{\theta}}$, and $\mathcal{H}_{\rm A} = (\mathcal{V}(\mathcal{H}))^\perp$ represents an ancillary (additional) Hilbert space as an orthogonal complement to $\mathcal{V}(\mathcal{H})$. Then the SLD operator $\mathfrak{L}_i^\prime \in \mathcal{H}^\prime$ satisfying Eq.~(\ref{eq:def_SLD_another_extension}) can be written as
    \begin{equation} \label{eq:SLD'extendedwithK'}
        \mathfrak{L}_i^\prime
        = \mathcal{V} L_i \mathcal{V}^\dagger + X_i,
    \end{equation}
    where $X_i$ is an Hermitian operator whose support is ${\rm ker} (\vr_{\vec{\theta}}^\prime) = \mathcal{V}(\mathcal{H}_{\rm K}) \oplus \mathcal{H}_{\rm A}$. Using the notations in Eq.~(\ref{eq:SLDcomponents}), we can decompose Eq.~(\ref{eq:SLD'extendedwithK'}) into components acting within and across these subspaces. In the block matrix form, $\mathfrak{L}_i^\prime$ is written as
    \begin{equation} \label{eq:SLD'_block}
    \mathfrak{L}_i^\prime
    =
    \begin{bmatrix}
    \mathcal{V} L_i^{\rm ss} \mathcal{V}^\dagger & \mathcal{V} L_i^{\rm sk} \mathcal{V}^\dagger & 0 \\
    \mathcal{V} L_i^{\rm ks} \mathcal{V}^\dagger & \mathcal{V} L_i^{\rm kk} \mathcal{V}^\dagger + X_i^{\rm kk} & X_i^{\rm ka} \\
    0 & X_i^{\rm ak} & X_i^{\rm aa}
    \end{bmatrix},
    \end{equation}
    where the script $\rm a$ represents the ancillary part of $X$, and $\mathfrak{L}_i^\prime$ has no coupling between $\mathcal{V}(\mathcal{H}_{\rm S})$ and $\mathcal{H}_{\rm A}$, due to Eq.~(\ref{eq:def_SLD_another_extension}) and $\partial_i \vr_{\vec{\theta}}^\prime =\mathcal{V} \partial_i \vr_{\vec{\theta}} \mathcal{V}^\dagger \in \mathcal{V}(\mathcal{H})$. 
    
    Denoting $\Pi_{\vr_{\vec{\theta}}}$ and $\Pi_{\vr_{\vec{\theta}}}^\perp$ as the projectors on $\mathcal{H}_{\rm S}$ and $\mathcal{H}_{\rm A}$, we have
    \begin{subequations}
        \begin{align}
        \Pi_{\vr_{\vec{\theta}}} \mathcal{V}^\dagger
        \mathfrak{L}_i^\prime \mathfrak{L}_j^\prime
        \mathcal{V} \Pi_{\vr_{\vec{\theta}}}
        &= \Pi_{\vr_{\vec{\theta}}}  L_i L_j \Pi_{\vr_{\vec{\theta}}},
        \\
        \Pi_{\vr_{\vec{\theta}}}^\perp \mathcal{V}^\dagger
        \mathfrak{L}_i^\prime \mathfrak{L}_j^\prime
        \mathcal{V} \Pi_{\vr_{\vec{\theta}}}
        &=\Pi_{\vr_{\vec{\theta}}}^\perp
        ( L_i L_j +  Y_{ij})
        \Pi_{\vr_{\vec{\theta}}},
    \end{align}
    \end{subequations}
    where $Y_{ij} = \mathcal{V}^\dagger X_i^{\rm kk} \mathcal{V} L_j^{\rm ks}$. Summing these and using $\Pi_{\vr_{\vec{\theta}}} + \Pi_{\vr_{\vec{\theta}}}^\perp = \eins_{\mathcal{H}}$ reduces to 
    \begin{equation}
        \label{eq:L'Lrel}
        \mathcal{V}^\dagger
        [\mathfrak{E} (\vr_{\vec{\theta}}^\prime)]_{ij}
        \mathcal{V} \Pi_{\vr_{\vec{\theta}}}
        = [O(\vr_{\vec{\theta}})]_{ij} + \Pi_{\vr_{\vec{\theta}}}^\perp (Y_{ij} -Y_{ji}) \Pi_{\vr_{\vec{\theta}}},
    \end{equation}
    where $O(\vr_{\vec{\theta}})$ and $\mathfrak{E} (\vr_{\vec{\theta}}^\prime)$ are defined in Eq.~(\ref{eq:def_O},\ref{eq:def_Sprime}).

    Finally, let us suppose that $\mathcal{Q}(\vr_{\vec{\theta}}) = 0$. According to the EC condition, there exist $\mathcal{V}$ and $\{\mathfrak{L}_i^\prime \}$ such that $\mathfrak{E} (\vr_{\vec{\theta}}^\prime) = 0$. Keeping such a SLD operator $\mathfrak{L}_i^{\prime}$ unchanged, we can redefine $L_i^{\mathrm{kk}}$ as $\tilde{L}_i^{\mathrm{kk}}=L_i^{\mathrm{kk}}-\mathcal{V}^\dagger X_i^{\mathrm{kk}} \mathcal{V}$ such that $\tilde{X}_i^{\mathrm{kk}}=0$. This is always possible because $\Pi_{\mathcal{V}} X_i^{\mathrm{kk}} \Pi_{\mathcal{V}} = X_i^{\mathrm{kk}} \in \mathcal{V}(\mathcal{H}_{\rm K})$ and the fact that the EC condition is invariant under any kernel shifts as discussed. Thus, Eq.~(\ref{eq:L'Lrel}) directly implies the OC condition. Hence, we can complete the proof.
\end{proof}

\subsection{Hierarchy of commutativity conditions}\label{sec:II.D}
Here, we summarize the logical connection between the commutativity conditions explained above. The following chain of implications holds:
\begin{widetext}
\begin{equation}\label{eq:several_hierarchy}
    \{ \exists \vec{L}^{\rm kk}| S(\vr_{\vec{\theta}}) = 0\}
    \ \stackrel{\rm (I)}{\Rightarrow} \
    \{ \exists \vec{\mathfrak{L}}^\prime| \mathfrak{E}(\vr_{\vec{\theta}}^\prime) = 0\}
    \ \stackrel{\rm (II)}{\Rightarrow} \
    \{ \exists \vec{L}^{\rm kk}| O(\vr_{\vec{\theta}}) = 0\}
    \ \stackrel{\rm (III)}{\Rightarrow} \
    P (\vr_{\vec{\theta}}) = 0
    \ \stackrel{\rm (IV)}{\Rightarrow} \
    W (\vr_{\vec{\theta}}) = 0,
\end{equation}
\end{widetext}
where we denote
\begin{equation}
    \vec{L}^{\rm kk}
    = \{L_i^{\rm kk} \ {\rm in} \ {\rm Eq.}~(\ref{eq:SLDcomponents}) \},
    \ \
    \vec{\mathfrak{L}}^\prime
    = \{ \mathfrak{L}_i^\prime \ {\rm in} \ {\rm Eq.}~(\ref{eq:def_SLD_another_extension}) \}.
\end{equation}
Note that $\mathcal{Q}(\vr_{\vec{\theta}}), W (\vr_{\vec{\theta}}), P (\vr_{\vec{\theta}}),  S(\vr_{\vec{\theta}}), O (\vr_{\vec{\theta}})$, and $\mathfrak{E}(\vr_{\vec{\theta}}^\prime)$ are respectively given in Eqs.~(\ref{eq:mathcalQ}, \ref{eq:def_W}, \ref{eq:def_P}, \ref{eq:def_S}, \ref{eq:def_O}, \ref{eq:def_Sprime}). Here, the implication (III)~follows from $P (\vr_{\vec{\theta}}) = \Pi_{\vr_{\vec{\theta}}} O(\vr_{\vec{\theta}})$, and the implication (IV)~follows from $W(\vr_{\vec{\theta}}) = \tr[\vr_{\vec{\theta}} P(\vr_{\vec{\theta}})]$.

We have six remarks on Eq.~(\ref{eq:several_hierarchy}). For simplicity, we here denote ($\rm N \inv$) as the converse of an implication ($\rm N$) for $\rm N = I, II, III, IV$. First, for general (both pure and mixed) states, the EC condition is both necessary and sufficient for $\mathcal{Q}(\vr_{\vec{\theta}}) = 0$, i.e., 
\begin{equation}
    \mathcal{Q}(\vr_{\vec{\theta}}) = 0
    \Leftrightarrow
    \{ \exists \vec{\mathfrak{L}}^\prime| \mathfrak{E}(\vr_{\vec{\theta}}^\prime) = 0\}.
\end{equation}
Note that the EC condition alone cannot provide criteria for choosing extended SLD operators from the state.

Second, for pure states, i.e., $\vr_{\vec{\theta}} = \ket{\psi_{\vec{\theta}}}\!\bra{\psi_{\vec{\theta}}}$, the WC condition also becomes sufficient for $\mathcal{Q}(\ket{\psi_{\vec{\theta}}}) = 0$~\cite{matsumoto2005geometricalPhD}:
\begin{equation}
    \label{eq:P=W=0_pure}
    \mathcal{Q}(\ket{\psi_{\vec{\theta}}}) = 0
    \Leftrightarrow
    W(\ket{\psi_{\vec{\theta}}}) = 0.
\end{equation}
That is, all the converse implications (I$\inv$), (II$\inv$), (III$\inv$), and (IV$\inv$) hold for pure states. The proof of (I$\inv$) is given in Appendix~\ref{ap:proof_Iinv}, while the others can be seen since $O (\ket{\psi_{\vec{\theta}}}) = P (\ket{\psi_{\vec{\theta}}}) =  W(\ket{\psi_{\vec{\theta}}}) \ket{\psi_{\vec{\theta}}}\! \bra{\psi_{\vec{\theta}}}$ with $L_i^{\rm kk} = 0$. This can be straightforwardly checked by taking the SLD operator in Eq.~(\ref{SLDpurestates}) and noticing that $\braket{\partial_i \psi_{\vec{\theta}}|\psi_{\vec{\theta}}} + \braket{\psi_{\vec{\theta}}|\partial_i\psi_{\vec{\theta}}} = 0$.

Third, for full-rank states, i.e., $\det{(\vr_{\vec{\theta}})} > 0$, the SC condition also becomes necessary for $\mathcal{Q}(\vr_{\vec{\theta}}) = 0$~\cite{nagaoka1987OnFisher}:
\begin{equation}
    \label{eq:S=P=0_fullrank}
    \mathcal{Q}(\vr_{\vec{\theta}}) = 0
    \Leftrightarrow
    S(\vr_{\vec{\theta}}) = 0.
\end{equation}
Notice that the kernel space of $\vr_{\vec{\theta}}$ does not exist in this case, namely $\Pi_{\vr_{\vec{\theta}}} = \eins$, where $\Pi_{\vr_{\vec{\theta}}}$ is the projector on the support space. That is, the converse implications (I$\inv$), (II$\inv$), and (III$\inv$) hold, since $\mathfrak{E}(\vr_{\vec{\theta}}) = S (\vr_{\vec{\theta}}) = O (\vr_{\vec{\theta}}) = P (\vr_{\vec{\theta}})$. On the other hand, (IV$\inv$) does not hold; for details, see Observation~\ref{ob:conversearrow} below.

Fourth, for single-qubit states (either pure or full-rank), according to Eqs.~(\ref{eq:P=W=0_pure}, \ref{eq:S=P=0_fullrank}), both necessary and sufficient conditions for the equivalence $\mathcal{Q}(\vr_{\vec{\theta}}) = 0$ are established. The simplified form of $W(\vr_{\vec{\theta}})$ will be discussed in Eq.~(\ref{eq:siglequbit}). On the other hand, a complete and comprehensive understanding of the equivalence beyond single qubits, such as higher-dimensional systems or bipartite/multipartite systems, remains elusive. 

Fifth, for rank-deficient states, the kernel structure further complicates the analysis of commutativity conditions. In this case, all the converse implications (I$\inv$), (II$\inv$), (III$\inv$), and (IV$\inv$) do not hold. In particular, (III$\inv$) requires an additional condition: $O(\vr_{\vec{\theta}}) = 0$ holds if and only if $P (\vr_{\vec{\theta}}) = 0$ and $\Pi_{\vr_{\vec{\theta}}}^\perp (L_i L_j - L_j L_i) \Pi_{\vr_{\vec{\theta}}} = 0$ hold, where $\Pi_{\vr_{\vec{\theta}}}^\perp = \eins - \Pi_{\vr_{\vec{\theta}}}$ is the projector on the kernel space of $\vr_{\vec{\theta}}$. In Section~\ref{sec:IV.B}, we will systematically analyze this gap and provide explicit counterexamples. Note that there exist models that demonstrate that the converse implications (I$\inv$) and (II$\inv$) do not hold; see Appendix~\ref{ap:Iinv_counterexample} in this manuscript and Example~E in Ref.~\cite{conlon2025role}, respectively. 

Finally, consider a pure state with unitary parameter-encoding: $\ket{\psi_{\vec{\theta}}} = U_{\vec{\theta}} \ket{\psi}$ for a unitary $U_{\vec{\theta}}$ with $\mathcal{H}_i = - i (\partial_i U_{\vec{\theta}}^\dagger) U_{\vec{\theta}} = i U_{\vec{\theta}}^\dagger (\partial_i U_{\vec{\theta}})$ being the Hermitian generator associated with $\theta_i$. From Eq.~(\ref{eq:P=W=0_pure}), we then have:
\begin{equation} \label{eq:relation_pure_QCRB}
    \mathcal{Q}(\ket{\psi_{\vec{\theta}}}) = 0
    \Leftrightarrow
    \braket{\psi|\mathcal{H}_i \mathcal{H}_j - \mathcal{H}_j \mathcal{H}_i|\psi}
    = 0,
    \quad
    \forall i,j.
\end{equation}
This yields a crucial result: \textit{For pure states, $\mathcal{Q}(\ket{\psi_{\vec{\theta}}}) = 0$ holds when all parameter-encoding generators $\mathcal{H}_i$ commute}. One of the main goals of this manuscript is to determine whether this remains true for mixed states as well.

\subsection{Solved and open problems}\label{sec:II.E}
For readers' convenience, let us highlight important solved and open problems in the field and related to this manuscript: 
\begin{itemize}
    \item 
    Problem~(A): Provide general necessary and sufficient conditions for an unspecified single-copy state $\vr_{\vec{\theta}}$ to attain the equivalence $\mathcal{C}_{\rm MI} (\vr_{\vec{\theta}}) = \mathcal{C}_{\rm QCR} (\vr_{\vec{\theta}})$, i.e., equivalently $\mathcal{Q}(\vr_{\vec{\theta}}) = 0$. This problem asks about the \textit{existence} of POVM measurements to achieve the equivalence. This is considered one of the five open problems in quantum information theory (see Problem 3 and discussion in Ref.~\cite{horodecki2022five}).

    \item 
    Problem~(B): Provide general both necessary and sufficient conditions for POVM measurements $\mathsf{E}$ acting on an unspecified single-copy state $\vr_{\vec{\theta}}$ to attain the equivalence $\mathcal{C}_{\rm CCR} (\vr_{\vec{\theta}}, \mathsf{E}) = \mathcal{C}_{\rm QCR} (\vr_{\vec{\theta}})$, i.e., equivalently $F_C(\vr_{\vec{\theta}}, \mathsf{E}) = F_Q(\vr_{\vec{\theta}})$. This problem seeks a \textit{recipe} for constructing optimal POVM measurements that achieve the equivalence.
\end{itemize}

First, solving Problem~(B) can be more involved than solving Problem~(A). For pure states, the WC condition is the solution to Problem~(A) in Ref.~\cite{matsumoto2002new}, and Problem~(B) was solved in Ref.~\cite{pezze2017optimal}. Note that Refs.~\cite{imai2026semiclassical,wang2025tight} recently addressed Problem~(B) more quantitatively by identifying the gap between the CFIM and QFIM. For full-rank states, to our knowledge, it has been reported that both Problems~(A)~and~(B) were solved by Nagaoka (e.g., see Ref.~\cite{nagaoka1987OnFisher} and Theorem~11 in Ref.~\cite{matsumoto2005geometrical}), where the SC condition is the solution to Problem~(A).

For general (rank-deficient) states, complete solutions to both Problems~(A)~and~(B) are believed to remain open. Although the EC condition has addressed Problem~(A) (see Theorem~14 in Ref.~\cite{conlon2022gap} and Appendix~B1 in Ref.~\cite{suzuki2020quantum}), this alone cannot identify criteria for choosing extended SLD operators from a given state. Very recently, Ref.~\cite{yang2026geometric} has addressed Problem~(B) from a geometric perspective. Yet, a comprehensive understanding of these connections remains elusive.

\section{Weak commutativity}\label{sec:III}
In the following, we focus on the case
\begin{equation} \label{eq:unitary}
    \vr_{\vec{\theta}} = U_{\vec{\theta}} \vr U_{\vec{\theta}}^\dagger,
\end{equation}
where $\vr$ is the initial probe state and $U_{\vec{\theta}}$ is a unitary parameter-encoding transformation with Hermitian generators 
\begin{equation} \label{eq:generator}
\mathcal{H}_i = - i (\partial_i U_{\vec{\theta}}^\dagger) U_{\vec{\theta}} = i U_{\vec{\theta}}^\dagger (\partial_i U_{\vec{\theta}}).
\end{equation}
We also denote with $\{\lambda_k, \ket{\psi_k}\}$ the (positive) eigenvalues and the corresponding eigenstates of the state $\vr$ as
\begin{equation} \label{eq:rho}
    \vr = \sum_k \lambda_k \ket{\psi_k} \! \bra{\psi_k},
\end{equation}
and define the projector $\Pi_k = \ket{\psi_k} \! \bra{\psi_k}$.

We are concerned with the question of whether or not the equivalence $\mathcal{Q}(\vr_{\vec{\theta}}) = 0$ holds for mixed quantum states when all parameter-generating Hamiltonians commute. A straightforward strategy to address this issue is to consider the weak commutativity (WC) condition explained in the previous section as a necessary condition for $\mathcal{Q}(\vr_{\vec{\theta}}) = 0$. Let us first present the general expression of $W(\vr_{\vec{\theta}})$ and then discuss relations between the WC condition and commuting generators. Finally, we provide several counterintuitive examples of the WC condition for commuting Hamiltonians.

\subsection{General expression of weak commutativity}\label{sec:III.A}
Let us first present the following:
\begin{observation} \label{ob:Gformsummary}
    Consider a rank-$r$ quantum state $\vr$ transforming according to Eq.~(\ref{eq:unitary}). Then, it holds that
    \begin{equation} \label{eq:Gformdecomp}
    W(\vr_{\vec{\theta}})
    = \Gamma(\vr_{\vec{\theta}}) + \Delta(\vr_{\vec{\theta}}),
    \end{equation}
    where $W(\vr_{\vec{\theta}})$ is defined in Eq.~(\ref{eq:def_W}), and
    \begin{subequations}
    \begin{align}
        \label{eq:Gamma_rho}
        [\Gamma(\vr_{\vec{\theta}})]_{ij}
        &= 4 \tr [\vr (\mathcal{H}_i \mathcal{H}_j - \mathcal{H}_j \mathcal{H}_i) ],
        \\
        \label{eq:Delta_rho}
        [\Delta (\vr_{\vec{\theta}})]_{ij}
        \!&=\! 4 \sum_{ k < l}^r  \gamma_{kl}
        \tr[ \swap \Pi_{kl} (\mathcal{H}_i \! \otimes \! \mathcal{H}_j -  \mathcal{H}_j \! \otimes \! \mathcal{H}_i ) ].
    \end{align}
    \end{subequations}
    Here, $\gamma_{kl} \equiv - 4 (\lambda_k -\lambda_{l})\lambda_k \lambda_{l}/(\lambda_k +\lambda_{l})^2$, $\swap$ denotes the SWAP operator with $\swap \ket{x} \otimes \ket{y} = \ket{y} \otimes \ket{x}$, and $\Pi_{kl} = \Pi_k \otimes \Pi_l$ with $\Pi_k = \ket{\psi_k} \! \bra{\psi_k}$. We recall $\ket{\psi_k}$ are eigenstates and $\lambda_k$ are eigenvalues of $\vr$ as in Eq.~(\ref{eq:rho}).  
\end{observation}
\begin{proof}
    We begin by rewriting Eq.~(\ref{eq:def_W}) as $[W(\vr_{\vec{\theta}})]_{ij} = \tr [\vr (\mathcal{L}_i \mathcal{L}_j - \mathcal{L}_i \mathcal{L}_j) ]$ with $\mathcal{L}_i \equiv U_{\vec{\theta}}^\dagger L_i U_{\vec{\theta}}$. For a rank-$r$ state $\vr$ in the Hilbert space with dimension $D$, we have
    \begin{equation} \label{eq:mathcal_L_i}
    \mathcal{L}_i
    = 2 i 
    \sum_{\substack{k, l = 1
    \\ \lambda_k + \lambda_l > 0}}^{D}
    \frac{\lambda_{k} - \lambda_{l}}{\lambda_k + \lambda_l}
    \mathfrak{h}_{kl}^{(i)}
    \ket{\psi_k} \! \bra{\psi_l},
    \end{equation} 
    where the sum runs over $k,l$ such that $\lambda_k + \lambda_l > 0$ and
    \begin{equation}
        \mathfrak{h}_{kl}^{(i)} \equiv \braket{\psi_k|\mathcal{H}_i|\psi_l}.
    \end{equation}
    Here, without loss of generality, we can set $L_i^{\rm kk}$ to zero, as discussed in Section~\ref{sec:II.C}. Inserting Eq.~(\ref{eq:mathcal_L_i}) into $\tr(\vr \mathcal{L}_i \mathcal{L}_j)$ leads to
    \begin{equation}
        \tr(\vr \mathcal{L}_i \mathcal{L}_j)
        = 4 \sum_{\substack{k, l = 1 \\ \lambda_k + \lambda_l > 0}}^D
        \frac{(\lambda_k - \lambda_{l})^2}{(\lambda_k +\lambda_{l})^2}
        \mathfrak{h}_{kl}^{(i)} \mathfrak{h}_{lk}^{(j)}
        \lambda_k,
    \end{equation}
    where we used $(\mathfrak{h}_{kl}^{(j)})^* = \mathfrak{h}_{lk}^{(j)}$. Then,
    \begin{equation} \label{eq:GQform_middle}
        [W(\vr_{\vec{\theta}})]_{ij}
        = 4 \sum_{\substack{k, l = 1\\ \lambda_k + \lambda_l > 0}}^D
        \frac{(\lambda_k - \lambda_{l})^3}{(\lambda_k +\lambda_{l})^2}
        \mathfrak{h}_{kl}^{(i)} \mathfrak{h}_{lk}^{(j)}.
    \end{equation}
    Now we apply the identity, $(\lambda_k - \lambda_{l})^3 = (\lambda_k + \lambda_{l})^2 (\lambda_k - \lambda_{l}) - 4(\lambda_k - \lambda_{l})\lambda_k \lambda_{l}$, which allows to rewrite Eq.~(\ref{eq:GQform_middle}) as
    \begin{equation}
        [W(\vr_{\vec{\theta}})]_{ij} =
        4 \sum_{\substack{k, l = 1\\ \lambda_k + \lambda_l > 0}}^D
        [(\lambda_k - \lambda_l) + \gamma_{kl}]
        \mathfrak{h}_{kl}^{(i)} \mathfrak{h}_{lk}^{(j)}.
    \end{equation}
    Clearly, in the second term, the sum runs over indices $k,l$ such that $\lambda_k \neq 0$, $\lambda_l \neq 0$, and $k \neq l$, due to the form $\gamma_{kl}$.
    
    The term $\Gamma(\vr_{\vec{\theta}})$ in Eq.~(\ref{eq:Gformdecomp}) can be obtained by noticing that $\sum_{k, l = 1}^D \lambda_k  \mathfrak{h}_{kl}^{(i)} \mathfrak{h}_{lk}^{(j)} = \tr(\vr \mathcal{H}_i \mathcal{H}_j)$, where we used $\sum_{l=1}^D \ket{\psi_l}\! \bra{\psi_l} = \eins$. Also, the term $\Delta(\vr_{\vec{\theta}})$ in Eq.~(\ref{eq:Gformdecomp}) can be obtained from the following:
    \begin{align} \nonumber
        \nonumber
        \mathfrak{h}_{kl}^{(i)} \mathfrak{h}_{lk}^{(j)} 
        &= \braket{\psi_k|\mathcal{H}_i|\psi_l} \!
        \braket{\psi_l|\mathcal{H}_j|\psi_k}
        \\
        \nonumber
        &= \braket{\psi_k \otimes \psi_l|\mathcal{H}_i \otimes \mathcal{H}_j |\psi_l \otimes \psi_k}
        \\
        \nonumber
        &= \tr[
        \ket{\psi_l \otimes \psi_k} \! \bra{\psi_k \otimes \psi_l}
        \mathcal{H}_i \otimes \mathcal{H}_j]
        \\
        \nonumber
        &= \tr[
        \swap \ket{\psi_k \otimes \psi_l} \! \bra{\psi_k \otimes \psi_l}
        \mathcal{H}_i \otimes \mathcal{H}_j]
        \\
        &= \tr[\swap \Pi_{kl} \mathcal{H}_i  \otimes  \mathcal{H}_j].
    \end{align}
    Using $\mathfrak{h}_{lk}^{(i)} \mathfrak{h}_{kl}^{(j)} = \mathfrak{h}_{kl}^{(j)} \mathfrak{h}_{lk}^{(i)}$ and $\gamma_{lk} = - \gamma_{kl}$ yields Eq.~(\ref{eq:Gformdecomp}).
\end{proof}

Here we have three main remarks on Observation~\ref{ob:Gformsummary}. First, the form of Eq.~(\ref{eq:Gformdecomp}) allows us to calculate $W(\vr_{\vec{\theta}})$ only using contributions from the support of the initial state $\vr$, excluding contributions from its kernel subspace. Thus, this calculation is analytically tractable and provides a useful method for determining whether a given state $\vr_{\vec{\theta}} = U_{\vec{\theta}} \vr U_{\vec{\theta}}^\dagger$ satisfies the WC condition. 

Second, the expression for $\Delta (\vr_{\vec{\theta}})$ in Eq.~(\ref{eq:Delta_rho}) explicitly involves the SWAP operator $\swap$. As we will see later in Observation~\ref{ob:local_ham_weak_always}, this form is useful to prove that the WC condition holds in certain examples. In those cases, we will employ a key identity, called the SWAP trick: $\tr[\swap (A \otimes B)] = \tr(AB)$ for operators $A$ and $B$.

Third, we notice that $\Delta(\vr_{\vec{\theta}})$ depends on the eigenbasis of the state $\vr$. In the following, we present two basis-independent representations of $W(\vr_{\vec{\theta}})$ using $\mathcal{H}_i \otimes \mathcal{H}_j - \mathcal{H}_j \otimes \mathcal{H}_i$, as shown in Appendix~\ref{ap:proof_another_Gform}. One is given by
\begin{equation} \label{eq:unitaryGform_basisInd}
    [W(\vr_{\vec{\theta}})]_{ij}
    = 4 \tr[
    \swap \, \mathcal{N}_\vr  \mathcal{T}_\vr
    (\mathcal{H}_i \otimes \mathcal{H}_j - \mathcal{H}_j \otimes \mathcal{H}_i) ],
\end{equation}
where $\mathcal{T}_\vr \equiv \int_{0}^{\infty} ds \int_{0}^{\infty} dt \, e^{- \vr (s + t)} \otimes e^{- \vr (s + t)}$ and $\mathcal{N}_\vr \equiv \vr^3 \otimes \eins + \vr \otimes \vr^2$. The other is given by
\begin{equation} \label{eq:unitaryGform_basisInd_hierarchical}
    [W(\vr_{\vec{\theta}})]_{ij}
    = 4 \lim_{\alpha \to \infty} [\mathcal{G}_\alpha(\vr_{\vec{\theta}})]_{ij},
\end{equation}
where $[\mathcal{G}_\alpha (\vr_{\vec{\theta}})]_{ij} = \sum_{a,b=0}^{2 \alpha} \xi_{ab}^{(\alpha)} \tr[\swap (\vr^a \otimes \vr^b) (\mathcal{H}_i  \otimes \mathcal{H}_j - \mathcal{H}_j  \otimes  \mathcal{H}_i)]$, with some coefficients $\xi_{ab}^{(\alpha)}$ and integers $a,b$. In the derivation of Eqs.~(\ref{eq:unitaryGform_basisInd}, \ref{eq:unitaryGform_basisInd_hierarchical}), we employed a technique similar to that considered in Ref.~\cite{rath2021quantum}.

\subsubsection{Pure states}\label{Sec.Rank1}
For a pure state, we have that $\Delta (\ket{\psi_{\vec{\theta}}})=0$ in Eq.~(\ref{eq:Gformdecomp}) by construction.
We thus recover that $W (\ket{\psi_{\vec{\theta}}}) = 0$ if and only if $\Gamma(\ket{\psi_{\vec{\theta}}}) = 0$, i.e., if and only if $\braket{\psi|\mathcal{H}_i \mathcal{H}_j - \mathcal{H}_j \mathcal{H}_i|\psi} = 0$ for all $i,j$, as discussed in Eqs.~(\ref{eq:P=W=0_pure}, \ref{eq:relation_pure_QCRB}).

\subsubsection{Rank-two states}
Consider a general rank-two state $\vr = \lambda \Pi_1 + (1-\lambda) \Pi_2$, where $\Pi_k = \ket{\psi_k} \! \bra{\psi_k}$ with $\lambda_1 = \lambda$ and $\lambda_2 = 1-\lambda$. Then, we can simplify $\Delta(\vr_{\vec{\theta}})$ in Eq.~(\ref{eq:Delta_rho}) as
\begin{equation} \label{eq:rank-two_formula}
    [\Delta(\vr_{\vec{\theta}})]_{ij}
    = 4\gamma_{12} \tr[ \swap \Pi_{12} (\mathcal{H}_i  \otimes  \mathcal{H}_j -  \mathcal{H}_j  \otimes \mathcal{H}_i ) ],
\end{equation}
where $\gamma_{12} = 4 (1-2\lambda)\lambda (1-\lambda)$ and $\Pi_{12} = \Pi_1 \otimes \Pi_2$. This expression provides a compact and more tractable form for evaluating $W(\vr_{\vec{\theta}})$ in the rank-two case. As we will discuss in Observation~\ref{ob:commutingbutnotweak}, we will apply this result to check whether the WC condition holds in specific examples.

\subsubsection{Full-rank states}
We will consider full-rank states in \textit{Example~2} in the next subsection.

\subsection{Weak commutativity and commuting generators}\label{sec:III.B}
Let us consider the case when the generators $\mathcal{H}_i$ in Eq.~(\ref{eq:generator}) commute:
\begin{equation}
        \mathcal{H}_i \mathcal{H}_j - \mathcal{H}_j \mathcal{H}_i = 0, \quad \forall i,j, 
\end{equation}
This condition leads to $\Gamma(\vr_{\vec{\theta}}) = 0$. As mentioned above, for pure states, this condition implies $W(\ket{\psi_{\vec{\theta}}})=0$ and thus $\mathcal{Q}(\ket{\psi_{\vec{\theta}}}) = 0$.

However, for general mixed states, the situation is more involved: If $\mathcal{H}_i \mathcal{H}_j - \mathcal{H}_j \mathcal{H}_i = 0$, then $\Gamma(\vr_{\vec{\theta}}) = 0$ holds, but $\Delta(\vr_{\vec{\theta}})$ does not necessarily vanish. That is, unlike pure states, commuting generators $\mathcal{H}_i$ do not lead to the WC condition, in general: see Observation~\ref{ob:commutingbutnotweak} for more details. The following two examples provide notable exceptions.

\subsubsection{Example 1}
For single-qubit states $\vr \in \mathbb{C}^2$, we can find
\begin{equation} \label{eq:siglequbit}
    W(\vr_{\vec{\theta}})
    = [2 \tr(\vr^2) - 1] \Gamma(\vr_{\vec{\theta}}).
\end{equation}
This can be directly derived using Eq.~(\ref{eq:rank-two_formula}) and the fact that $\Pi_1 + \Pi_2 = \eins$. If all $\mathcal{H}_i$ commute, then $\Gamma(\vr_{\vec{\theta}}) = 0$, and therefore, $W(\vr_{\vec{\theta}}) = 0$ for any single-qubit (either pure or rank-two mixed) state.

\subsubsection{Example 2}
For a mixed state of the form
\begin{equation}
    \vr_p = p \Pi_\psi + (1-p) \eins/d^N
\end{equation}
where $p \in [0,1]$ and $\Pi_\psi = \ket{\psi} \! \bra{\psi}$ is a $N$-qudit pure state, namely $\vr_p\in (\mathbb{C}^d)^{\otimes N}$. We have
\begin{equation} \label{eq:mixed_whitenoise}
    [W(\vr_{\vec{\theta}})]_{ij}
    = \frac{4 p^3 d^{2 N}}{[p (d^N - 2) + 2]^2}
    \braket{\psi|\mathcal{H}_i \mathcal{H}_j - \mathcal{H}_j \mathcal{H}_i|\psi}.
\end{equation}
This can be derived by $\lambda_1 = p + (1-p)/d^N$ with $\ket{\psi_1} = \ket{\psi}$ and $\lambda_2 = (1-p)/d^N$ (with multiplicity $d^N -1$), and using $\sum_{k=2}^{d^N-1} \ket{\psi_k}\! \bra{\psi_k} = \eins - \Pi_\psi$: $[\Gamma (\vr_{\vec{\theta}})]_{ij} = 4 p \tr[\Pi_\psi (\mathcal{H}_i \mathcal{H}_j - \mathcal{H}_j \mathcal{H}_i) ]$ and $[\Delta (\vr_{\vec{\theta}})]_{ij} = 4 \gamma_{12} \tr[\Pi_\psi (\mathcal{H}_i \mathcal{H}_j - \mathcal{H}_j \mathcal{H}_i) ]$. Clearly, if all $\mathcal{H}_i$ commute, then $W(\vr_{\vec{\theta}}) = 0$ in this case. 

\subsection{Weak commutativity and commuting Hamiltonians}\label{sec:III.C}
Here we focus on the unitary $U_{\vec{\theta}}$ given by
\begin{equation} \label{eq:unitary_specific}
    U_{\vec{\theta}} = e^{-i \sum_{i=1}^m \theta_i H_i},
\end{equation}
where $H_i$ are Hamiltonians. In general, the generator $\mathcal{H}_i$ as in Eq.~(\ref{eq:generator}) is not equal to $H_i$, which can be seen from the Wilcox formula~\cite{wilcox1967exponential}: for an operator $X$,
\begin{equation} \label{eq:wilcoxexponentialformula}
    \partial_i e^X = \int_0^1 ds \, e^{sX} (\partial_i X) e^{(1-s)X}.
\end{equation}
In the special case where all Hamiltonians commute, i.e., $H_i H_j - H_j H_i = 0$ for all $i,j$, the generator reduces to $\mathcal{H}_i = H_i$, and thus $\mathcal{H}_i \mathcal{H}_j - \mathcal{H}_j \mathcal{H}_i = 0$. This setting is particularly relevant in practical applications such as distributed quantum sensing~\cite{humphreys2013quantum,proctor2018multiparameter,ge2018distributed,guo2020distributed, liu2021distributed,malia2022distributed,malitesta2023distributed,kim2024distributed,hong2025quantum, pezze2025distributed,li2026multiparameter,minati2025multiparameter,yan2025scalable}, where parameter-encoding operations are performed locally across spatially separated parties. 

While commutativity of the Hamiltonians in the phase-encoding transformation in Eq.~(\ref{eq:unitary_specific}) implies the WC condition for pure states, it does not generally hold for mixed states. More precisely:
\begin{observation} \label{ob:commutingbutnotweak}    
    Even when the Hamiltonians $\{H_i\}$ commute in Eq.~(\ref{eq:unitary_specific}), there exist rank-deficient quantum states such that the WC condition does not hold, i.e.,
    \begin{equation}
        \forall i,j, \
        H_i H_j - H_j H_i = 0        
        \nRightarrow
        W(\vr_{\vec{\theta}}) = 0.
    \end{equation}
\end{observation}
This observation is supported by four explicit examples that we provide below. In particular, we consider rank-two mixed states of the form $\vr = \lambda \ket{\psi_1} \! \bra{\psi_1} + (1-\lambda) \ket{\psi_2} \! \bra{\psi_2}$. In this setting, we notice that $W(\vr_{\vec{\theta}}) = \Delta(\vr_{\vec{\theta}})$, as given in Eq.~(\ref{eq:rank-two_formula}).

\subsubsection{Example 3}
Consider a single-qutrit system $\mathbb{C}^3$ with commuting Hamiltonians given by
\begin{equation}
    H_1
    = \sqrt{\frac{3}{2}}
    \begin{bmatrix}
    0 & -i & 0 \\
    i & 0 & 0 \\
    0 & 0 & 0
    \end{bmatrix},
    \quad
    \label{eq:hamH2qutrit}
    H_2
    = \frac{1}{\sqrt{2}}
    \begin{bmatrix}
    1 & 0 & 0 \\
    0 & 1 & 0 \\
    0 & 0 & -2
    \end{bmatrix}.
\end{equation}
We construct the rank-two state $\vr$, where the eigenstates $\ket{\psi_k}$ for $k=1,2$ are chosen as $\ket{\psi_k} = \sin(\alpha) \sin(\beta_k) \ket{0} + \sin(\alpha) \cos(\beta_k) \ket{1} + \cos(\alpha) \ket{2}$ with $\beta_1 = - \beta_2 = \beta = (1/2) (\pi - \cos^{-1}[\cot^2(\alpha)] )$ chosen to ensure orthogonality: $\braket{\psi_1|\psi_2}=0$. Then we have 
\begin{equation}
[W(\vr_{\vec{\theta}})]_{12} = h_\lambda [1-\cos (4 \alpha)] \sqrt{\cos (2 \alpha + \pi) \csc^4(\alpha)}, 
\end{equation}
where $h_\lambda \equiv -6 i \sqrt{3} (1-\lambda) \lambda (2 \lambda-1)$. One can find values of $\lambda$ and $\alpha$ for which $W(\vr_{\vec{\theta}}) \neq 0$. This result demonstrates that the WC condition can become nontrivial in higher-dimensional systems, even when the Hamiltonians commute, in contrast to the single-qubit case in Eq.~(\ref{eq:siglequbit}).

\subsubsection{Example 4}
Consider a two-qubit system $(\mathbb{C}^2)^{\otimes 2}$ with Hamiltonians: $H_1 = \sigma_x \otimes \eins$ and $H_2 = \eins \otimes \sigma_y$, where $\sigma_i$ denotes the Pauli matrix along the $i$-axis ($i=x,y,z$). Clearly, these Hamiltonians, being local, commute. Let us take the initial state as the separable mixture $\vr = p \ket{00}\! \bra{00} + (1-p) \ket{++}\! \bra{++}$ with $\ket{+} = (1/\sqrt{2})(\ket{0}+\ket{1})$. Then we obtain 
\begin{equation}
    [W(\vr_{\vec{\theta}})]_{12} = -8i (1 - p) p^2,
\end{equation}
with the eigenvalue $\lambda = (1/2) [1 - \sqrt{1 - 3 (1-p) p}]$. Again, we have $W(\vr_{\vec{\theta}}) \neq 0$. This indicates that separable mixed states can violate the WC condition, even when the Hamiltonians are local.

In contrast, consider a separable state \textit{without} classical correlations $\vr_{\vec{\theta}} = \vr_{\vec{\theta}}^A \otimes \vr_{\vec{\theta}}^B$. In this case, one can write $W(\vr_{\vec{\theta}})$ as the additive form, 
\begin{equation}
W(\vr_{\vec{\theta}}) = W(\vr_{\vec{\theta}}^A) + W(\vr_{\vec{\theta}}^B),    
\end{equation}
due to the structure of the SLDs: $L_{i} = L_{i}^A \otimes \eins_B + \eins_A \otimes L_{i}^B$. Since the eigenvalues of $W$ always appear in $\pm$~pairs because of its antisymmetry $W^\trans = -W$, if the subsystem contributions satisfy $W(\vr_{\vec{\theta}}^A) = -W(\vr_{\vec{\theta}}^B)$, then it holds that $W(\vr_{\vec{\theta}}) = 0$. Such a case highlights a fundamental distinction between the WC condition at the level of the whole system and its marginal subsystems: Even when each part individually violates the WC condition, their contributions can cancel in the composite system.

\subsubsection{Example 5}
Consider a three-qubit system $(\mathbb{C}^2)^{\otimes 3}$ with local Hamiltonians: $H_1 = \sigma_z \otimes \eins \otimes \eins$, $H_2 = \eins \otimes \sigma_z \otimes \eins$, and $H_3 = \eins \otimes \eins \otimes \sigma_z$. We choose the eigenstates of the rank-two state $\vr$ as $\ket{\psi_1} = (\ket{001} + \omega \ket{010} + \omega^2 \ket{100})/\sqrt{3}$ and $\ket{\psi_2} = (\ket{001} + \omega^2 \ket{010} + \omega \ket{100})/\sqrt{3}$, where $\omega = e^{2 \pi i /3}$ satisfies $1 + \omega + \omega^2 = 0$. Then we obtain that $[W(\vr_{\vec{\theta}})]_{12} = [W(\vr_{\vec{\theta}})]_{23} = - [W(\vr_{\vec{\theta}})]_{13} = 64i (1-\lambda) \lambda (1-2 \lambda)/(3\sqrt{3})$. This shows that classical correlations in entangled states can also lead to a violation of the WC condition, even when identical Hamiltonians act on distinct subsystems.

\subsubsection{Example 6}
Consider two-qubit systems $(\mathbb{C}^2)^{\otimes 2}$ with Hamiltonians: $H_1 = \sigma_x \otimes \eins$ and $H_2 = \eins \otimes \sigma_x$. Let $\vr_{XY}$ be the two-qubit reduced states of a three-qubit pure state $\ket{\Psi}_{ABC} = (\ket{001} + \omega \ket{010} + \omega^2 \ket{100})/\sqrt{3}$, taken over the subsystems $XY = AB, BC, CA$. For each of these reduced states $\vr_{XY}$, the WC condition does not hold. We stress that this cannot occur if $\ket{\Psi}_{ABC}$ is a pure product state, since all the reduced states of any pure product state remain pure, and thus automatically satisfy the WC condition under commuting Hamiltonians. This example illustrates that, even when the WC condition is satisfied for an entangled whole state (such as $\ket{\Psi}_{ABC}$), the WC condition does not necessarily hold for its reduced subsystems (such as $\vr_{XY}$), which highlights a distinction between whole and marginal behavior under local Hamiltonians.

\subsection{Weak commutativity and two-qudit rank-deficient states}
Here we present another counterintuitive example illustrating the opposite scenario, in which the WC condition always holds. We have already seen such an example as the mixed state with white noise in Eq.~(\ref{eq:mixed_whitenoise}), which is a full-rank state. We have the following:
\begin{observation} \label{ob:local_ham_weak_always}
    For any local Hamiltonian of the form $H_1 = A \otimes \eins$ and $H_2 = \eins \otimes B$ in Eq.~(\ref{eq:unitary_specific}), where $A$ and $B$ are real and Hermitian operators (i.e., $A^\trans = A$ and $B^\trans = B$), there exist two-qudit rank-deficient states such that the WC condition always holds.
\end{observation}
This is shown in the following example. Consider the generalized Bell-diagonal state~\cite{horodecki1996information,chruscinski2010bell,moerland2024bound}, the class of two-qudit mixed states formed by generalized orthonormal Bell states:
\begin{equation} \label{eq:belldiagonalstate}
    \vr_{\rm BD} = \sum_{k} \lambda_k \ket{\psi_k}\! \bra{\psi_k},
    \quad
    \braket{\psi_k|\psi_l} = \delta_{kl},
\end{equation}
where each eigenstate $\ket{\psi_k} \in (\mathbb{C}^d)^{\otimes 2}$ with corresponding eigenvalue $\lambda_k$ is a maximally entangled state, meaning that the reduced state for one particle is maximally mixed. For example, in the qubit case ($d = 2$), the Bell basis consists of the four states:
\begin{subequations}
\begin{align}
    \label{eq:BDS_d=2_v1}
    \!
    \ket{\psi_1} &= \frac{1}{\sqrt{2}} (\ket{00}+\ket{11}),
    \ \
    \ket{\psi_2} = \frac{1}{\sqrt{2}} (\ket{00}-\ket{11}),
    \\
    \label{eq:BDS_d=2_v2}
    \!
    \ket{\psi_3} &= \frac{1}{\sqrt{2}} (\ket{01}+\ket{10}),
    \ \
    \ket{\psi_4} = \frac{1}{\sqrt{2}} (\ket{01}-\ket{10}).
\end{align}
\end{subequations}

Below we prove that the \textit{real} generalized Bell-diagonal state, such that $\vr_{\rm BD}^\trans = \vr_{\rm BD}$, provides an explicit example that satisfies Observation~\ref{ob:local_ham_weak_always}. This result suggests that certain classical correlations in entangled states can always preserve the WC condition, even for an arbitrary real local Hamiltonian:
\begin{proof}
We begin by recalling that for commuting Hamiltonians $H_1$ and $H_2$, it holds that $W(\vr_{\vec{\theta}}) = \Delta(\vr_{\vec{\theta}})$ from Observation~\ref{ob:Gformsummary}. Moreover, $\Delta(\vr_{\vec{\theta}}) = 0$ holds, if the following condition holds:
\begin{equation} \label{eq:genetalizedBelldiagonal}
    \tr[\swap \Pi_{kl}
    (A \otimes \eins \otimes \eins \otimes B - 
    \eins \otimes B \otimes A \otimes \eins )] = 0,
    \forall
    k,l,
\end{equation}
where $\Pi_{kl} = \ket{\psi_k}\! \bra{\psi_k} \otimes \ket{\psi_l}\! \bra{\psi_l}$, and each $\ket{\psi_k}$ is now taken as a maximally entangled state.

Next, we note that for any maximally entangled state $\ket{\psi_k}$, there exists a local unitary $U_A \otimes U_B$ such that $\ket{\psi_k} = (U_A \otimes U_B) \ket{\Psi^+}$, where $\ket{\Psi^+} = (1/\sqrt{d}) \sum_{i=0}^{d-1} \ket{i} \otimes \ket{i}$. For simplicity, we write $\ket{\psi_k} = U_{AB} \ket{\Psi^+}$ and $\ket{\psi_l} = V_{AB} \ket{\Psi^+}$, where $U_{AB} = U_A \otimes U_B$ and $V_{AB} = V_A \otimes V_B$. Using the identity $\ket{\Psi^+}\! \bra{\Psi^+} = (1/d)\swap^{\trans_A}$, where $\swap$ denotes the SWAP operator and $\trans_A$ denotes the partial transpose with respect to the subsystem, we have
\begin{align} \nonumber
    \!\!\!
    &\Pi_{kl} (A \otimes \eins \otimes \eins \otimes B)
    \\
    \!=
    &\frac{1}{d^2}
    [U_{AB} \swap^{\trans_A} U_{AB}^\dagger (A \! \otimes \! \eins)]
    \! \otimes \!
    [V_{AB} \swap^{\trans_A} V_{AB}^\dagger (\eins \! \otimes \! B)].
    \label{eq:trick1}
\end{align}

Applying the SWAP trick $\tr[\swap(X \otimes Y)] = \tr[XY]$ and the relations $\tr(X Y^{\trans_A} Z Y^{\trans_A}) = \tr(X^{\trans_A} Y Z^{\trans_A} Y)$ and $\swap (X \otimes Y) \swap = Y \otimes X$ for operators $X,Y,Z$, we obtain
\begin{align}
    \nonumber
    &\tr \left[\swap \Pi_{kl} (A \otimes \eins \otimes \eins \otimes B) \right]
    \\ \!=\! \ \nonumber
    &\frac{1}{d^2} \tr \left[U_{AB} \swap^{\trans_A} U_{AB}^\dagger (A \otimes \eins)
    V_{AB} \swap^{\trans_A} V_{AB}^\dagger (\eins \otimes B) \right]
    \\ \!=\! \ \nonumber
    &\frac{1}{d^2} \tr \left[
    V_{AB}^\dagger (\eins \otimes B) U_{AB}
    \swap^{\trans_A}
    U_{AB}^\dagger (A \otimes \eins) V_{AB}
    \swap^{\trans_A} \right]
    \\ \!=\! \ \nonumber
    &\frac{1}{d^2} \tr \left\{
    [V_{AB}^\dagger (\eins \otimes B) U_{AB}]^{\trans_A}
    \swap
    [U_{AB}^\dagger (A \otimes \eins) V_{AB}]^{\trans_A}
    \swap \right\}
    \\ \!=\! \ \nonumber
    &\frac{1}{d^2}
    \tr \left[(V_{A}^\dagger \eins U_A)^\trans (U_B^\dagger \eins V_B) \right]
    \tr \left[ (V_{B}^\dagger B U_B) (U_A^\dagger A V_A)^\trans \right]
    \\ \!=\! \ 
    &\frac{1}{d^2} \tr \left[U_A^\trans V_{A}^* U_B^\dagger V_B \right]
    \tr \left[V_{B}^\dagger B U_B V_A^\trans A^\trans U_A^* \right],
    \label{eq:H1H2term}
\end{align}
where we denoted $\trans$, $*$, and $\dagger$ respectively as the transposition, the complex conjugation, and the Hermitian (conjugate transpose). Here we used 
\begin{subequations}
    \begin{align}
        [V_{AB}^\dagger (\eins \otimes B) U_{AB}]^{\trans_A}
        &= (V_{A}^\dagger \eins U_A)^\trans \otimes V_{B}^\dagger B U_B,
        \\
        [U_{AB}^\dagger (A \otimes \eins) V_{AB}]^{\trans_A}
        &= (U_A^\dagger A V_A)^\trans \otimes U_B^\dagger \eins V_B.
    \end{align}
\end{subequations}

To proceed, let us recall the three things that we have considered: (i)~each basis $\ket{\psi_k}$ is orthonormal as in Eq.~(\ref{eq:belldiagonalstate}); (ii)~the generalized Bell-diagonal state $\vr_{\rm BD}$ is real, $\vr_{\rm BD}^\trans = \vr_{\rm BD}^* = \vr_{\rm BD}$; (iii)~$A$ and $B$ are real (i.e., $A^\trans = A$ and $B^\trans = B$) in Observation~\ref{ob:local_ham_weak_always}.

Now, (i) imposes further conditions on the unitaries $U_A, U_B, V_A$, and $V_B$. In fact, one can rewrite $\ket{\psi_k}$ as $\ket{\psi_k} = (U_A \otimes \eins)  \ket{\Psi^+}$, where $U_A$ belongs to a class of Heisenberg-Weyl operators~\cite{chruscinski2010bell,moerland2024bound}: $\{U_{mn} = Z_m X_n\}$, for $Z_m = \sum_j w^{jm} \ket{j}\!\bra{j}$ and $X_n = \sum_j \ket{j-n}\! \bra{j}$ representing phase and shift operations, respectively. Then we can take $U_B = V_B = \eins$. Also, (ii) implies that the above unitary class $\{U_{mn} \}$ is real, i.e., $U_{mn}^* = U_{mn}$. Then we take $U_A^* = U_A$ and $V_A^* = V_A$. Combining these with (iii), we rewrite Eq.~(\ref{eq:H1H2term}) as $\tr \left[\swap \Pi_{kl} (A \otimes \eins \otimes \eins \otimes B) \right] = (1/d^2) \tr \left[U_A^\trans V_{A} \right] \tr \left[B V_A^\trans A U_A \right]$.

Finally, using $\tr[\swap \Pi_{kl} H_2 \otimes H_1] = \tr[\swap \Pi_{kl} H_1 \otimes H_2]^*$, we obtain that $\tr[\swap \Pi_{kl} (A \otimes \eins \otimes \eins \otimes B - \eins \otimes B \otimes A \otimes \eins )] = 0$, where we used $[\tr(X)]^* = \tr(X^\dagger)$ and $\tr(X^\trans) = \tr(X)$ for an operator $X$. According to Eq.~(\ref{eq:genetalizedBelldiagonal}), we can show that the real generalized Bell-diagonal state is an example in Observation~\ref{ob:local_ham_weak_always}.
\end{proof}

In particular, for the two-qubit case, we can provide a stronger statement than Observation~\ref{ob:local_ham_weak_always} by excluding the \textit{real} assumption on Hamiltonians as follows:
\begin{observation} \label{ob:local_ham_weak_always_twoqubit}
    For any local Hamiltonian of the form $H_1 = h_{\vec{u}_A} \otimes \eins$ and $H_2 = \eins \otimes h_{\vec{u}_B}$ in Eq.~(\ref{eq:unitary_specific}), where $h_{\vec{u}} \!=\! \sum_{i=x,y,z} u_i \sigma_i$ for a vector $\vec{u} \!=\! (u_x, u_y,u_z)$, there exist two-qubit rank-deficient states such that the WC condition always holds.
\end{observation}
This can be seen by considering the two-qubit Bell-diagonal state, $\vr_{\rm BD} = \sum_{k=1}^4 \lambda_k \ket{\psi_k}\! \bra{\psi_k}$ for the states $\{\ket{\psi_k}\}$ given in Eqs.~(\ref{eq:BDS_d=2_v1}, \ref{eq:BDS_d=2_v2}). In this case, one can immediately check that $\tr[\swap \Pi_{kl} (h_{\vec{u}_A} \otimes \eins \otimes \eins \otimes h_{\vec{u}_B} - \eins \otimes h_{\vec{u}_B} \otimes h_{\vec{u}_A} \otimes \eins )]$ in Eq.~(\ref{eq:genetalizedBelldiagonal}) vanishes for all $k,l \in [1,4]$. This shows that the two-qubit Bell-diagonal state with any rank obeys the WC condition.

\section{Relation between commutativity conditions}\label{sec:IV}
Here, we first present general expressions relevant to the PC, OC, and SC conditions. Then we examine whether, when a certain commutativity condition within the hierarchy in Eq.~(\ref{eq:several_hierarchy}) holds, the corresponding converse implication also holds for commuting Hamiltonians. We provide several counterintuitive examples to illustrate these points. Below we use the notations in Eqs.~(\ref{eq:unitary}, \ref{eq:generator}, \ref{eq:rho}) and denote $\Pi_\vr = \sum_{k=1}^r \Pi_k$ and $\Pi_{\vr}^\perp = \eins - \Pi_\vr$ as the projectors onto the support and kernel spaces of a rank-$r$ state $\vr$ with $\Pi_k = \ket{\psi_k} \! \bra{\psi_k}$ for $\ket{\psi_k}$ being the eigenstates of $\vr$ with nonzero eigenvalue.

\subsection{General expressions}\label{sec:IV.A}
We note that the PC, OC, and SC conditions as in Eqs.~(\ref{eq:def_P}, \ref{eq:def_S}, \ref{eq:def_O}) are respectively rewritten as $\Pi_{\vr} (\mathcal{L}_i \mathcal{L}_j - \mathcal{L}_j \mathcal{L}_i) \Pi_{\vr} \!=\! 0$, $(\mathcal{L}_i \mathcal{L}_j - \mathcal{L}_j \mathcal{L}_i) \Pi_{\vr} \!=\! 0$, and $\mathcal{L}_i \mathcal{L}_j - \mathcal{L}_j \mathcal{L}_i \!=\! 0$, where $\mathcal{L}_i \equiv U_{\vec{\theta}}^\dagger L_i U_{\vec{\theta}}$ is defined in the proof of Observation~\ref{ob:Gformsummary}. Let us denote $\mathcal{P} (\vr_{\vec{\theta}}), \mathcal{O}(\vr_{\vec{\theta}})$, and $\mathcal{S}(\vr_{\vec{\theta}})$ as matrices with elements defined as
\begin{subequations}
\begin{align}
    [\mathcal{P} (\vr_{\vec{\theta}})]_{ij}
    &\equiv \Pi_{\vr} (\mathcal{L}_i \mathcal{L}_j - \mathcal{L}_j \mathcal{L}_i) \Pi_{\vr},
    \\
    [\mathcal{O}(\vr_{\vec{\theta}})]_{ij}
    &\equiv (\mathcal{L}_i \mathcal{L}_j - \mathcal{L}_j \mathcal{L}_i) \Pi_{\vr},
    \\
    [\mathcal{S}(\vr_{\vec{\theta}})]_{ij}
    &\equiv \mathcal{L}_i \mathcal{L}_j - \mathcal{L}_j \mathcal{L}_i.
\end{align}
\end{subequations}
Note that $\mathcal{P} (\vr_{\vec{\theta}}) \!=\! U_{\vec{\theta}}^\dagger P (\vr_{\vec{\theta}}) U_{\vec{\theta}}$, $\mathcal{O} (\vr_{\vec{\theta}}) \!=\! U_{\vec{\theta}}^\dagger O (\vr_{\vec{\theta}}) U_{\vec{\theta}}$, and $\mathcal{S} (\vr_{\vec{\theta}}) \!=\! U_{\vec{\theta}}^\dagger S (\vr_{\vec{\theta}}) U_{\vec{\theta}}$.

We provide the detailed expressions of these matrices, similarly to Observation~\ref{ob:Gformsummary}:
\begin{observation} \label{ob:GE:UPUandUMUandUSU}
    Consider a rank-$r$ quantum state $\vr$ transforming according to Eq.~(\ref{eq:unitary}), where $\Pi_\vr$ and $\Pi_{\vr}^\perp$ respectively denote the projectors on the support and kernel spaces of $\vr$. Then, it holds that
    \begin{subequations}
    \begin{align}   
        \label{eq:GE:UPU^daggwe}
        \mathcal{P} (\vr_{\vec{\theta}})
        &= \mathcal{I}_{\rm ss}(\vr_{\vec{\theta}}) + \mathcal{I}_{\rm ss}^\prime(\vr_{\vec{\theta}}),
        \\
        \label{eq:GE:UMU^daggwe}
        \mathcal{O} (\vr_{\vec{\theta}})
        &= \mathcal{P} (\vr_{\vec{\theta}})
        + \mathcal{I}_{\rm ks}(\vr_{\vec{\theta}})
        + \mathcal{J}_{\rm ks}(\vr_{\vec{\theta}}),
        \\
        \nonumber
        \mathcal{S} (\vr_{\vec{\theta}})
        &= \mathcal{O} (\vr_{\vec{\theta}})
        + \mathcal{I}_{\rm sk}(\vr_{\vec{\theta}}) 
        + \mathcal{I}_{\rm kk}(\vr_{\vec{\theta}})
        \\
        \label{eq:GE:USU^daggwe}
        &+ \mathcal{J}_{\rm sk}(\vr_{\vec{\theta}}) + \mathcal{J}_{\rm kk} (\vr_{\vec{\theta}}),
    \end{align}
    \end{subequations}
    where each of components represent support-support~(ss), support-kernel~(sk), kernel-support~(ks), and kernel-kernel~(kk) contribution of the state $\vr$ given by
    \begin{subequations}
    \begin{align}
    \label{eq:def_OathcalIss}
    [\mathcal{I}_{\rm ss}(\vr_{\vec{\theta}})]_{ij}
    &= 4 \Pi_\vr (\mathcal{H}_i \mathcal{H}_j - \mathcal{H}_i \mathcal{H}_j) \Pi_\vr,
    \\
    \label{eq:def_OathcalIssprime1}
    [\mathcal{I}_{\rm ss}^\prime(\vr_{\vec{\theta}})]_{ij}
    &=
    -4 \Biggl\{
    \sum_{k=1}^r
    \Pi_k \mathcal{D}_{ij}^k \Pi_k   
    \\ \label{eq:def_OathcalIssprime2}
    &+\! \! \sum_{k \neq l}^r
    \Bigl[
    \Pi_k \mathcal{D}_{ij}^{\vr} \Pi_l
    +
    \frac{ 4\lambda_{k}\lambda_{l}}{(\lambda_k \!+\! \lambda_l)^2}
    \Pi_k \mathcal{D}_{ij}^l \Pi_k
    \Bigr]
    \\
    &- \sum_{k \neq l \neq m}^r
    \eta_{km} \eta_{lm}
    \Pi_k \mathcal{D}_{ij}^m \Pi_l
    \Biggr\},
    \label{eq:def_OathcalIssprime3}
    \\
    \label{eq:def_OathcalIsk}
    [\mathcal{I}_{\rm sk}(\vr_{\vec{\theta}})]_{ij}
    &= - 4 \sum_{k \neq l}^r \frac{\lambda_{k} - \lambda_{l}}{\lambda_k + \lambda_l}
    \Pi_k \mathcal{D}_{ij}^l \Pi_{\vr}^\perp,
    \\
    \label{eq:def_OathcalIks}
    [\mathcal{I}_{\rm ks}(\vr_{\vec{\theta}})]_{ij}
    &=  - 4 \sum_{k \neq l}^r \frac{\lambda_{k} - \lambda_{l}}{\lambda_k + \lambda_l}
    \Pi_{\vr}^\perp \mathcal{D}_{ij}^l \Pi_k,
    \\
    \label{eq:def_OathcalIkk}
    [\mathcal{I}_{\rm kk}(\vr_{\vec{\theta}})]_{ij}
    &= 4 \Pi_{\vr}^\perp \mathcal{D}_{ij}^{\vr} \Pi_{\vr}^\perp,
    \\
    \label{eq:def_OathcalJsk}
    [\mathcal{J}_{\rm sk}(\vr_{\vec{\theta}})]_{ij}
    &= 2i \Pi_\vr \mathcal{H}_i \mathcal{L}_j^{{\rm kk}},
    \\
    \label{eq:def_OathcalJks}
    [\mathcal{J}_{\rm ks}(\vr_{\vec{\theta}})]_{ij}
    &= -2i \mathcal{L}_i^{{\rm kk}} \mathcal{H}_j \Pi_\vr,
    \\
    \label{eq:def_OathcalJkk}
    [\mathcal{J}_{\rm kk} (\vr_{\vec{\theta}})]_{ij}
    &=\mathcal{L}_i^{{\rm kk}} \mathcal{L}_j^{{\rm kk}} - \mathcal{L}_j^{{\rm kk}} \mathcal{L}_i^{{\rm kk}}.
    \end{align}
    \end{subequations}
    Here, $\eta_{kl} = (\lambda_{k} - \lambda_{l})/(\lambda_{k} + \lambda_{l})$, $\mathcal{D}_{ij}^{\vr} \equiv \mathcal{H}_i \Pi_\vr \mathcal{H}_j - \mathcal{H}_j \Pi_\vr \mathcal{H}_i$, and $\mathcal{D}_{ij}^k \equiv \mathcal{H}_i \Pi_k \mathcal{H}_j - \mathcal{H}_j \Pi_k \mathcal{H}_i$ with $\mathcal{D}_{ij}^\vr = \sum_{k=1}^r \mathcal{D}_{ij}^k$. The operator $\mathcal{L}_i^{{\rm kk}}$ acts within the kernel space of $\vr$ in Eq.~(\ref{eq:SLDcomponents}) via $\mathcal{L}_i = U_{\vec{\theta}}^\dagger L_i U_{\vec{\theta}}$, which is not uniquely determined by the definition of the SLD operator.
\end{observation}

\begin{proof}
    The derivation of Eq.~(\ref{eq:GE:USU^daggwe}) is given in Appendix~\ref{ap:GE:USU^daggwe}. We notice that Eqs.~(\ref{eq:GE:UPU^daggwe}, \ref{eq:GE:UMU^daggwe}) directly follow from Eq.~(\ref{eq:GE:USU^daggwe}), using the relation $\Pi_{\vr} X \Pi_{\vr} = 0$ for $X = \mathcal{I}_{\rm sk}, \mathcal{I}_{\rm ks}, \mathcal{I}_{\rm kk}, \mathcal{J}_{\rm sk}, \mathcal{J}_{\rm ks}, \mathcal{J}_{\rm kk}$ and $Y \Pi_{\vr} = 0$ for $Y = \mathcal{I}_{\rm sk}, \mathcal{I}_{\rm kk}, \mathcal{J}_{\rm sk}, \mathcal{J}_{\rm kk}$.
\end{proof}

Here we have three remarks on Observation~\ref{ob:GE:UPUandUMUandUSU}. First, when $\mathcal{H}_i \mathcal{H}_j - \mathcal{H}_j \mathcal{H}_i = 0$ for all $i,j$, then $\mathcal{I}_{\rm ss}(\vr_{\vec{\theta}}) = 0$ holds, but the other terms do not vanish in general. Second, similarly to $W(\vr_{\vec{\theta}})$ in Eq.~(\ref{eq:Gformdecomp}), the calculation of $\mathcal{P} (\vr_{\vec{\theta}})$ only requires contributions from the support space of $\vr$. Also many terms in $\mathcal{O} (\vr_{\vec{\theta}})$ and $\mathcal{S} (\vr_{\vec{\theta}})$ only contain support contributions without knowledge from the kernel since $\Pi_{\vr}^\perp = \eins - \Pi_\vr$, but the terms $\mathcal{J}_{\rm sk}(\vr_{\vec{\theta}}), \mathcal{J}_{\rm ks}(\vr_{\vec{\theta}}),$ and $\mathcal{J}_{\rm kk}(\vr_{\vec{\theta}})$ cannot be uniquely determined due to the presence of $\mathcal{L}_i^{{\rm kk}}$. This identifies explicit differences between these commutativity conditions, as discussed in Section~\ref{subsubsec:remarks}.

Third, $\mathcal{I}_{\rm ss}^\prime(\vr_{\vec{\theta}})$ allows us to reveal a distinction between the WC and PC conditions. To see this, let us recall that $W(\vr_{\vec{\theta}}) = \Gamma(\vr_{\vec{\theta}}) + \Delta(\vr_{\vec{\theta}})$ presented in Observation~\ref{ob:Gformsummary}. Using the relation $W(\vr_{\vec{\theta}}) = \tr[\vr \mathcal{P}(\vr_{\vec{\theta}})]$, we find that 
\begin{equation}
    \Gamma(\vr_{\vec{\theta}}) = \tr[\vr \mathcal{I}_{\rm ss}(\vr_{\vec{\theta}})],
    \quad
    \Delta(\vr_{\vec{\theta}}) = \tr[\vr \mathcal{I}_{\rm ss}^\prime(\vr_{\vec{\theta}})].
\end{equation}
Since $\tr[\vr \Pi_k \mathcal{D}_{ij}^k \Pi_k]  = \tr[\vr \Pi_k \mathcal{D}_{ij}^m \Pi_l] = \tr[\vr \Pi_k \mathcal{D}_{\vr} \Pi_l] = 0$ holds for $k \neq l$, the terms in Eqs.~(\ref{eq:def_OathcalIssprime1}, \ref{eq:def_OathcalIssprime2}) vanish in $\Delta(\vr_{\vec{\theta}})$, while the term in Eq.~(\ref{eq:def_OathcalIssprime3}) appear in $\Delta(\vr_{\vec{\theta}})$. In other words, the terms in Eqs.~(\ref{eq:def_OathcalIssprime1}, \ref{eq:def_OathcalIssprime2}) can identify a difference between $W(\vr_{\vec{\theta}})$ and $\mathcal{P}(\vr_{\vec{\theta}})$.

\subsection{Converses of commutativity chain of Eq.~(\ref{eq:several_hierarchy})}\label{sec:IV.B}
Here we present the following:
\begin{observation} \label{ob:conversearrow}
     Consider commuting Hamiltonians $\{H_i\}$ in Eq.~(\ref{eq:unitary_specific}).
     (a)~There exist rank-deficient quantum states such that converse implications within the hierarchical chain in Eq.~(\ref{eq:several_hierarchy}) do not hold, i.e.,
     \begin{widetext}
     \begin{subequations}
     \begin{align}
         \label{eq:arrow:WnotP_PnotO}
         &\{ \exists \vec{L}^{\rm kk}| O(\vr_{\vec{\theta}}) = 0\}
         \ \stackrel{\rm (III\inv)}{\nLeftarrow} \
         P (\vr_{\vec{\theta}}) = 0
         \ \stackrel{\rm (IV\inv)}{\nLeftarrow} \
         W (\vr_{\vec{\theta}}) = 0,
         \\
         \label{eq:arrow:OnotS}
         &\{ \exists \vec{L}^{\rm kk}| O(\vr_{\vec{\theta}}) = 0\}
         \ \ \, \nRightarrow \ \ \,
         \{ \exists \vec{L}^{\rm kk}| S(\vr_{\vec{\theta}}) = 0\}.
     \end{align}
     \end{subequations} 
     \end{widetext}
     (b)~There exist full-rank states such that the WC condition holds but $\mathcal{Q}(\vr_{\vec{\theta}})$ in Eq.~(\ref{eq:mathcalQ}) does not vanish, i.e.,
    \begin{equation}
        \label{eq:WCnotimplyQ}
        W (\vr_{\vec{\theta}}) = 0
        \ \nRightarrow \
        \mathcal{Q} (\vr_{\vec{\theta}}) = 0.
    \end{equation}
\end{observation}

Observation~\ref{ob:conversearrow} is supported by explicit examples: (III$\inv$) and (IV$\inv$) in Eq.~(\ref{eq:arrow:WnotP_PnotO}) are shown in Examples~7~and~8; Eq.~(\ref{eq:arrow:OnotS}) is shown in Example~9; Eq.~(\ref{eq:WCnotimplyQ}) is shown in Example~10. There, we demonstrate the above coincidence and distinction between commutativity conditions more explicitly.

To this end, we note that for a rank-two mixed state $\vr = \lambda \ket{\psi_1} \! \bra{\psi_1} + (1-\lambda) \ket{\psi_2} \! \bra{\psi_2}$, the terms $\mathcal{I}_{\rm ss}^\prime(\vr_{\vec{\theta}})$ and $\mathcal{I}_{\rm ks}(\vr_{\vec{\theta}})$ are given by
\begin{align}
    \nonumber
    [\mathcal{I}_{\rm ss}^\prime(\vr_{\vec{\theta}})]_{ij}
    &= 
    -4\Bigl[
    \Pi_1 \mathcal{D}_{ij}^{1} \Pi_1
    + \Pi_2 \mathcal{D}_{ij}^{2} \Pi_2
    \\ \nonumber
    &\quad
    + \Pi_1 \mathcal{D}_{ij}^{\vr} \Pi_2
    + \Pi_2 \mathcal{D}_{ij}^{\vr} \Pi_1
    \\
    &\quad
    + 4 \lambda(1-\lambda)
    (\Pi_1 \mathcal{D}_{ij}^2 \Pi_1 + \Pi_2 \mathcal{D}_{ij}^1 \Pi_2)
    \Bigr],
    \label{eq:ranktwo_ssprime}
    \\
    [\mathcal{I}_{\rm ks}(\vr_{\vec{\theta}})]_{ij}
    &=4 (1 - 2 \lambda ) \Bigl[
    \Pi_{\vr}^\perp \mathcal{D}_{ij}^2 \Pi_1
    - \Pi_{\vr}^\perp \mathcal{D}_{ij}^1 \Pi_2
    \Bigr],
    \label{eq:ranktwo_ks}
\end{align}
where $\Pi_{\vr}^\perp = \eins - \Pi_1 - \Pi_2$ and $\Pi_k = \ket{\psi_k}\! \bra{\psi_k}$ for $k=1,2$.

We remark that the statement of Eq.~(\ref{eq:WCnotimplyQ}) does not extend to the asymptotic limit: there exist \textit{no} full-rank states such that $W(\vr_{\vec{\theta}}^{\otimes \nu}) = 0$ but $\mathcal{Q}(\vr_{\vec{\theta}}^{\otimes \nu}) \neq 0$ in the limit $\nu \to \infty$~\cite{demkowicz2020multi,ragy2016compatibility}. This follows from the fact that the inequality $\mathcal{C}_{\rm MI} \geq \mathcal{C}_{\rm H}$ is always saturated asymptotically by collective measurements on infinitely many copies, i.e., $\mathcal{C}_{\rm MI} (\vr_{\vec{\theta}}^{\otimes \nu}) = \mathcal{C}_{\rm H} (\vr_{\vec{\theta}}^{\otimes \nu})$ for $\nu \to \infty$~\cite{kahn2009local, yamagata2013quantum, yang2019attaining, demkowicz2020multi}.

\subsubsection{Example 7}
Consider the rank-two state $\vr$ based on the eigenstates $\{ \ket{\psi_1}, \ket{\psi_2} \}$ presented at \textit{Example 3} in Section~\ref{sec:III.C}, with commuting Hamiltonians: $H_1 = a h  + a^\prime h^\prime $ and $H_2$ as in Eq.~(\ref{eq:hamH2qutrit}), where $a, a^\prime$ are real coefficients, and $h, h^\prime$ are given by
\begin{equation}
    \label{eq:hamhhprime}
    h
    = \sqrt{\frac{3}{2}}
    \begin{bmatrix}
    0 & 1 & 0 \\
    1 & 0 & 0 \\
    0 & 0 & 0
    \end{bmatrix},
    \quad
    h^\prime
    = \sqrt{\frac{3}{2}}
    \begin{bmatrix}
    1 & 0 & 0 \\
    0 & -1 & 0 \\
    0 & 0 & 0
    \end{bmatrix}.
\end{equation}
In this case, we obtain that $W (\vr_{\vec{\theta}}) = 0$ for any $a,a^\prime$ but
\begin{equation}
    \label{eq:PC_value_C3}
    [\mathcal{P} (\vr_{\vec{\theta}})]_{12}
    = 12 a \begin{bmatrix}
    0 & s & t \\
    -s & 0 & 0 \\
    -t & 0 & 0
    \end{bmatrix},
\end{equation}
where the parameters $s \equiv - \sqrt{3} \cos^2(\alpha) \cos (2 \alpha)$ and $t \equiv \sqrt{6} \sin(\alpha) \cos^3(\alpha) \sqrt{1-\cot^2(\alpha)}$ depend on $\alpha$ via the eigenstates $\{ \ket{\psi_1}, \ket{\psi_2} \}$ of \textit{Example 3}. This result demonstrates~(IV$\inv$) in Observation~\ref{ob:conversearrow}, where the WC condition holds but the PC condition does not hold, even when the Hamiltonians commute. Note that $\mathcal{I}_{\rm kk}(\vr_{\vec{\theta}}) = 0$ holds.

In this case, taking $\alpha = \pi/4$ yields $\mathcal{P} (\vr_{\vec{\theta}}) = 0$ but
\begin{equation}
    [\mathcal{I}_{\rm ks}(\vr_{\vec{\theta}})]_{12}
    = \begin{bmatrix}
    0 & 0 & 0 \\
    0 & 0 & 3 a \sqrt{3} (1-2 \lambda ) \\
    0 & 0 & 0
    \end{bmatrix}.
\end{equation}
Now, since the dimension of the kernel space of $\vr$ is one, we can set $\mathcal{L}_i^{{\rm kk}} = \zeta_i \Pi_{\vr}^\perp$ for complex coefficients $\zeta_i$. Then we obtain that $[\mathcal{J}_{\rm ks}(\vr_{\vec{\theta}})]_{12} = -2i \zeta_1 \Pi_{\vr}^\perp \mathcal{H}_2 \Pi_\vr = 0$. Thus, $\mathcal{O} (\vr_{\vec{\theta}})$ does not vanish for any choice of $\{ \mathcal{L}_i^{{\rm kk}} \}$. This demonstrates~(III$\inv$) in Observation~\ref{ob:conversearrow}, where the PC condition holds but the OC condition does not hold.

Moreover, in the case with $\alpha = \pi/4$ and $a=0$, we can see that $\mathcal{O} (\vr_{\vec{\theta}}) = \mathcal{S} (\vr_{\vec{\theta}}) = 0$. Note that in this case with $a^\prime = 1$, the QFIM in Eq.~(\ref{eq:QFImatrix}) is given by
\begin{equation}
    \label{eq:QFIM_qutrit}
    F_Q(\vr_{\vec{\theta}})
    = \begin{bmatrix}
    3(1-2\lambda)^2 & 3\sqrt{3}(1-2\lambda)^2 \\
    3\sqrt{3}(1-2\lambda)^2 & 9(1-2\lambda)^2
    \end{bmatrix}.
\end{equation}

\subsubsection{Example 8}
Consider the two-qubit rank-three Bell-diagonal state $\vr_{\rm BD} \in (\mathbb{C}^2)^{\otimes 2}$ in Eq.~(\ref{eq:belldiagonalstate}) with local Hamiltonians: $H_1 = \sigma_A \otimes \eins$ and $H_2 = \eins \otimes \sigma_B$, where $\sigma_A = a_x \sigma_x + a_z \sigma_z$ and $\sigma_B = b_x \sigma_x + b_z \sigma_z$. Choosing the eigenstates $\{\ket{\psi_1}, \ket{\psi_2}, \ket{\psi_3} \}$ with the eigenvalues $\{\lambda_1, \lambda_2, 1-\lambda_1-\lambda_2\}$ in Eqs.~(\ref{eq:BDS_d=2_v1}, \ref{eq:BDS_d=2_v2}), we notice that $W(\vr_{\vec{\theta}}) = 0$ as discussed in Observation~\ref{ob:local_ham_weak_always_twoqubit} but
\begin{equation}
    \label{eq:PC_value_C22Rank3}
    [\mathcal{P} (\vr_{\vec{\theta}})]_{12}
    = \begin{bmatrix}
    0 & f & f & 0 \\
    -f & 0 & 0 & f \\
    -f & 0 & 0 & f \\
    0 & -f & -f & 0
    \end{bmatrix},
\end{equation}
where $f \!=\! [4 (1 \!-\! \lambda_1) \lambda_1 (a_x b_z \!-\! a_z b_x)]/[(1 \!-\! \lambda_2) (\lambda_1 \!+\! \lambda_2)]$. This demonstrates~(IV$\inv$) in Observation~\ref{ob:conversearrow}, where the WC condition holds but the PC condition does not hold, even for local Hamiltonians. Note that $\mathcal{I}_{\rm kk}(\vr_{\vec{\theta}}) = 0$ holds.

In this case, taking $a_x b_z = a_z b_x$ yields $\mathcal{P} (\vr_{\vec{\theta}}) = 0$ but 
\begin{equation}
    [\mathcal{I}_{\rm ks}(\vr_{\vec{\theta}})]_{12}
    = \begin{bmatrix}
    0 & 0 & 0 & 0 \\
    g & 0 & 0 & g \\
    -g & 0 & 0 & -g \\
    0 & 0 & 0 & 0
    \end{bmatrix},
\end{equation}
where $g \!=\! [4 \lambda_1 (1 \!-\! \lambda_1 \!-\! 2 \lambda_2) (a_x b_z \!+\! a_z b_x)]/[(1 \!-\! \lambda_2) (\lambda_1 \!+\! \lambda_2)]$. Similarly to \textit{Example~7}, setting $\mathcal{L}_i^{{\rm kk}} = \zeta_i \Pi_{\vr}^\perp$ for complex $\zeta_i$ since the dimension of the kernel space is one, we obtain that
\begin{subequations}
\begin{align}
    [\mathcal{J}_{\rm ks}(\vr_{\vec{\theta}})]_{12}
    &= (-i\zeta_1 )
    \begin{bmatrix}
    0 & 0 & 0 & 0 \\
    - a_x & a_z & a_z & a_x \\
    a_x & -a_z & -a_z & - a_x \\
    0 & 0 & 0 & 0
    \end{bmatrix},
    \\
    [\mathcal{J}_{\rm ks}(\vr_{\vec{\theta}})]_{21}
    &= i\zeta_2
    \begin{bmatrix}
    0 & 0 & 0 & 0 \\
    - b_x & b_z & b_z & b_x \\
    b_x & -b_z & -b_z & -b_x \\
    0 & 0 & 0 & 0
    \end{bmatrix}.
\end{align}
\end{subequations}
Thus, $\mathcal{O} (\vr_{\vec{\theta}})$ does not vanish in nonzero $a_x, a_z, b_x, b_z$ for any choice of $\{ \mathcal{L}_i^{{\rm kk}} \}$. This demonstrates~(III$\inv$) in Observation~\ref{ob:conversearrow}, where the PC condition holds but the OC condition does not hold, even for local Hamiltonians.

Furthermore, in the case with $a_z=b_z=0$, we can see that $\mathcal{O} (\vr_{\vec{\theta}}) = \mathcal{S} (\vr_{\vec{\theta}}) = 0$. Note that in this case, with $a_x=b_x=1$, the QFIM in Eq.~(\ref{eq:QFImatrix}) is given by
\begin{equation}
    \label{eq:QFIM_twoqubit}
    F_Q(\vr_{\vec{\theta}})
    \!=\! \frac{8}{1-\lambda_2}
    \begin{bmatrix}
     c & c \!-\! 2 (1 \!-\! \lambda_2) \lambda_2 \\
    c \!-\! 2 (1 \!-\! \lambda_2) \lambda_2 & c
    \end{bmatrix},
\end{equation}
where $c = 1 - \lambda_2 - 4 \lambda_1 (1-\lambda_1 -\lambda_2)$.

\subsubsection{Example 9}
Consider the two-\textit{qutrit} rank-two Bell-diagonal state $\vr_{\rm BD} \in (\mathbb{C}^3)^{\otimes 2}$ in Eq.~(\ref{eq:belldiagonalstate}) with local Hamiltonians: $H_1 = \eta \otimes \eins$ and $H_2 = \eins \otimes \eta$, where $\eta = a h  + a^\prime h^\prime $ for $h, h^\prime$ as in Eq.~(\ref{eq:hamhhprime}). Taking the eigenstates $\ket{\psi_1} = (1/\sqrt{2}) (\ket{01}+\ket{10})$ and $\ket{\psi_2} = (1/\sqrt{2}) (\ket{12}+\ket{21})$ with $\braket{\psi_1|\psi_2} = 0$, we obtain that $W(\vr_{\vec{\theta}}) = \mathcal{P}(\vr_{\vec{\theta}}) = 0$. Also, choosing $\mathcal{L}_i^{{\rm kk}} = 0$, we have that $\mathcal{O}(\vr_{\vec{\theta}}) = 0$ due to $\mathcal{I}_{\rm ks}(\vr_{\vec{\theta}}) = 0$.

However, $\mathcal{I}_{\rm kk}(\vr_{\vec{\theta}})$ does not vanish for nonzero $a$. This demonstrates Eq.~(\ref{eq:arrow:OnotS}) in Observation~\ref{ob:conversearrow}, where the OC condition holds but the SC condition does not hold. In the case with $a=0$ and $a^\prime = 1$, we have that $\mathcal{S} (\vr_{\vec{\theta}}) = 0$ while the QFIM in Eq.~(\ref{eq:QFImatrix}) is given by
\begin{equation}
    \label{eq:QFIM_twoqutrit}
    F_Q(\vr_{\vec{\theta}})
    \!=\! \frac{3(1+3\lambda)}{2}
    \begin{bmatrix}
     1 & -1\\
     -1 & 1
    \end{bmatrix}.
\end{equation}

\subsubsection{Example 10}
Consider the full-rank state $\vr \in \mathbb{C}^D$ that has an eigenstate $\ket{\psi_{1}} = \ket{\psi}$ with the eigenvalue $\lambda$ and all the other eigenstates $\ket{\psi_{k}}$ for $k \in [2, D]$ with the same eigenvalue $\lambda^* = (1-\lambda)/(D-1)$. Note that this state is the generalization of the mixed state with white noise as in Eq.~(\ref{eq:mixed_whitenoise}). In this case, we have that
\begin{align}
    \nonumber
     [\mathcal{S} (\vr_{\vec{\theta}})]_{ij}
    &= 4 \left( \frac{\lambda - \lambda^*}{\lambda + \lambda^*} \right)^2
    \bigl[
    \Pi_\psi (\mathcal{H}_i \mathcal{H}_j - \mathcal{H}_j \mathcal{H}_i) \Pi_\psi
    \\
    &\quad
    + \mathcal{D}_{ij}^{\psi} - \Pi_\psi \mathcal{D}_{ij}^{\psi} - \mathcal{D}_{ij}^{\psi} \Pi_\psi
    \Bigr],
    \label{eq:fullrank_S}
\end{align}
where $\Pi_\psi = \ket{\psi}\! \bra{\psi}$ and $\mathcal{D}_{ij}^{\psi} \equiv \mathcal{H}_i \Pi_\psi \mathcal{H}_j - \mathcal{H}_j \Pi_\psi \mathcal{H}_i$. This can be derived using the form of $\mathcal{L}_i = 2 i [( \lambda - \lambda^* ) / (\lambda + \lambda^*)] \left( \Pi_\psi \mathcal{H}_i - \mathcal{H}_i \Pi_\psi \right)$ and $\sum_{k=2}^D \ket{\psi_k}\!\bra{\psi_k} = \eins - \Pi_\psi$, since the kernel space of $\vr$ does not exist. Now, using $W (\vr_{\vec{\theta}}) = \tr[\vr \mathcal{S} (\vr_{\vec{\theta}})]$, we obtain
\begin{equation}
    [W (\vr_{\vec{\theta}})]_{ij} 
    = 4 \frac{(\lambda - \lambda^*)^3}{(\lambda + \lambda^*)^2}
    \braket{\psi|\mathcal{H}_i \mathcal{H}_j - \mathcal{H}_j \mathcal{H}_i|\psi},
    \label{eq:fullrank_W}
\end{equation}
where we used that $\tr(\Pi_\psi \mathcal{D}_{ij}^{\psi}) = 0$.

Hence, for commuting Hamiltonians $\{H_i\}$, the WC condition holds, i.e., $W (\vr_{\vec{\theta}}) = 0$, while the SC condition holds if and only if $\mathcal{D}_{ij}^{\psi} - \Pi_\psi \mathcal{D}_{ij}^{\psi} - \mathcal{D}_{ij}^{\psi} \Pi_\psi = 0$ holds. In fact, choosing the Bell state $\ket{\psi} = (1/\sqrt{2})(\ket{01}+\ket{10})$ and local Hamiltonians $H_1 = \sigma_A \otimes \eins$ and $H_2 = \eins \otimes \sigma_A$ for $\sigma_A= a_x \sigma_x + a_z \sigma_z$, we have
\begin{equation}
    \mathcal{D}_{ij}^{\psi} - \Pi_\psi \mathcal{D}_{ij}^{\psi} - \mathcal{D}_{ij}^{\psi} \Pi_\psi
    = a_x a_z
    \begin{bmatrix}
    0 & -1 & 1 & 0 \\
    1& 0 & 0 & 1 \\
    -1 & 0 & 0 & -1 \\
    0 & -1 & 1 & 0
    \end{bmatrix}.
\end{equation}
Let us recall that for full-rank states, the SC condition becomes both necessary and sufficient to have $\mathcal{Q}(\vr_{\vec{\theta}}) = 0$, as discussed in Eq.~(\ref{eq:S=P=0_fullrank}). Thus, we can conclude that this example demonstrates Eq.~(\ref{eq:WCnotimplyQ}) in Observation~\ref{ob:conversearrow}.

\section{Conclusion} \label{sec:V}
In this manuscript, we have studied the hierarchical relations among several commutativity conditions relevant to the saturation of the QCR bound in multiparameter quantum metrology. In particular, we have considered the scenario in which parameter-encoding is implemented via a unitary transformation. In this setting, we showed that our observations and examples identify gaps in the hierarchical chain and demonstrate its irreversibilities, as summarized in Fig.~\ref{fig1}, even when all parameter-encoding generators commute. We note that related questions have been considered in the Bayesian framework, where commuting generators can yield measurement incompatibility~\cite{albarelli2025measurement}.

There are several open avenues for further research. First, although the EC condition is both necessary and sufficient for the saturation of the QCR bound, it does not, on its own, provide any criteria for choosing extended SLD operators from a given state. It would therefore be worthwhile to consider a closed-form expression for this purpose. Second, understanding how the gaps between several commutativity conditions relate to geometric perspectives and provide insights into physical phenomena such as entanglement would be interesting. Finally, although our results provide several counterintuitive examples, they may encourage the development of experimental protocols in noisy distributed quantum sensing.

\section*{Acknowledgments}
We thank
Francesco Albarelli,
Syed Assad,
Lorcan Conlon, and
Jun Suzuki,
for discussions.
S.I. acknowledges support from JST ASPIRE (JPMJAP2339).
J.Y. acknowledges support from Zhejiang University start-up grants, Zhejiang Key Laboratory of R\&D and Application of Cutting-edge Scientific Instruments, and Wallenberg Initiative on Networks and Quantum Information (WINQ).
L.P. acknowledges support from the QuantERA project SQUEIS (Squeezing enhanced inertial sensing), funded by the European Union's Horizon Europe Program and the Agence Nationale de la Recherche (ANR-22-QUA2-0006). This publication has received funding under Horizon Europe programme HORIZON-CL4-2022-QUANTUM-02-SGA via the project 10113690 PASQuanS2.1.

\vspace{0.25cm}
{\footnotesize
\hypertarget{email1}{}\noindent\textsuperscript{*} \href{mailto:satoyaimai@yahoo.co.jp}{satoyaimai@yahoo.co.jp} \\
\hypertarget{email2}{}\textsuperscript{\textdagger} \href{jing.yang.quantum@zju.edu.cn}{jing.yang.quantum@zju.edu.cn} \\
\hypertarget{email3}{}\textsuperscript{\textdaggerdbl} \href{mailto:luca.pezze@ino.cnr.it}{luca.pezze@ino.cnr.it}
}

\clearpage
\newpage
\appendix
\onecolumngrid

\section{Extended commutativity (EC) condition on an enlarged Hilbert space}\label{ap:commutativity_extension}
Here, for clarity and completeness, and for readers' convenience, we provide a self-contained discussion of the extended commutativity (EC) condition presented in Section~\ref{subsubsec:ECcondition}, building on previously known results (see Refs.~\cite{helstrom1969quantum,helstromBOOK1976,fujiwara1995geometricalPhD,fujiwara1995quantum,holevo2011probabilistic,matsumoto2005geometricalPhD,demkowicz2020multi,liu2020quantum,suzuki2020quantum,conlon2022gap}). In Appendix~\ref{ap:notations_results}, we outline several notations and existing results. In Appendix~\ref{ap:proofs}, we prove the inequalities in Eq.~(\ref{eq:QCRBinequality}) and these saturation conditions in Eqs.~(\ref{eq:QCRBinequality_saturation:i}, \ref{eq:QCRBinequality_saturation:ii}). In Appendix~\ref{ap:Derivation_extended_commutativity}, we show that the EC condition is both necessary and sufficient for the saturation of the QCR bound~\cite{suzuki2020quantum,conlon2022gap}.

\subsection{Several notations and existing results} \label{ap:notations_results}
Here, we summarize several notations that are used in the main text and will be used in Appendix~\ref{ap:commutativity_extension}:
\begin{itemize}
    \item 
    $n$ is the number of protocol repetitions, and $\nu$ is the number of identical copies of the state. In the following, for the sake of simplicity, we consider the case $n = \nu = 1$.
    \item
    $\vr_{\vec{\theta}}$ is the parameter-encoded quantum state, and $L_i \equiv L_i (\vr_{\vec{\theta}})$ is its SLD operator defined by Eq.~(\ref{eq:SLD_definition_derivative}).
    \item 
    $\mathsf{E} = \{ E_\omega \}$ is a POVM, and $\omega$ is a measurement outcome obtained with a probability $p(\omega | \vec{\theta}) = \tr(\vr_{\vec{\theta}} E_\omega)$.
    \item 
    $\Tilde{\theta}_i = \Tilde{\theta}_i(\omega)$ is the estimator of the parameter $\theta_i$, and $\braket{\Tilde{\theta}_i} = \sum_{\omega} p (\omega |\vec{\theta}) \Tilde{\theta}_i (\omega)$ is its statistical mean value. In the following, we consider the vector $\Tilde{\vec{\theta}} = \{ \Tilde{\theta}_1, \ldots, \Tilde{\theta}_m \}$ and the locally-unbiased estimator, i.e., $\braket{\Tilde{\theta}_i} = \theta_i$ and $\partial_i \braket{\Tilde{\theta}_j} = \delta_{ij}$.
    \item 
    ${\rm Cov}(\vr_{\vec{\theta}}, \mathsf{E}, \Tilde{\vec{\theta}})$ is the covariance matrix with elements $[{\rm Cov}(\vr_{\vec{\theta}}, \mathsf{E}, \Tilde{\vec{\theta}})]_{ij} = \sum_{\omega} p (\omega | \vec{\theta})  (\theta_i - \Tilde{\theta}_i) (\theta_j - \Tilde{\theta}_j)$.
    \item 
    $\mathfrak{C}(\vr_{\vec{\theta}}, \mathsf{E}, \Tilde{\vec{\theta}})$ is a matrix with elements $[\mathfrak{C}(\vr_{\vec{\theta}}, \mathsf{E}, \Tilde{\vec{\theta}})]_{ij} = (1/2) \tr[\vr_{\vec{\theta}} (\Tilde{E}_i \Tilde{E}_j + \Tilde{E}_j \Tilde{E}_i)]$, where $\Tilde{E}_i \equiv \sum_\omega E_\omega [\tilde{\theta}_i - \theta_i]$.
    \item 
    $F_Q(\vr_{\vec{\theta}})$ is the QFIM with the elements $[F_Q(\vr_{\vec{\theta}})]_{ij} = (1/2) \tr [\vr_{\vec{\theta}} (L_i L_j + L_j L_i) ]$.
    \item 
    $\mathcal{Q}(\vr_{\vec{\theta}}) = \mathcal{C}_{\rm MI} (\vr_{\vec{\theta}}) - \mathcal{C}_{\rm QCR} (\vr_{\vec{\theta}})$, where $\mathcal{C}_{\rm MI} (\vr_{\vec{\theta}}) = \min_{\Tilde{\vec{\theta}}, \mathsf{E}} \tr[M {\rm Cov}(\vr_{\vec{\theta}}, \mathsf{E}, \Tilde{\vec{\theta}})]$ and $\mathcal{C}_{\rm QCR} (\vr_{\vec{\theta}}) = \tr[M F_Q^{-1} (\vr_{\vec{\theta}})]$ for $M$ being an $m \times m$ real, symmetric, and positive-definite matrix.
\end{itemize}

Next, we summarize several existing results previously considered in the literature. The following matrix inequalities are known to hold for all $\vr_{\vec{\theta}}, \mathsf{E}$, and $\Tilde{\vec{\theta}}$~\cite{helstrom1969quantum,helstromBOOK1976,demkowicz2020multi,liu2020quantum}:
\begin{equation}
    \label{eq:QCRBinequality}
    {\rm Cov}(\vr_{\vec{\theta}}, \mathsf{E}, \Tilde{\vec{\theta}})
    \geq 
    \mathfrak{C}(\vr_{\vec{\theta}}, \mathsf{E}, \Tilde{\vec{\theta}})
    \geq
    F_Q^{-1} (\vr_{\vec{\theta}}).
\end{equation}
The equality ${\rm Cov}(\vr_{\vec{\theta}}, \mathsf{E}, \Tilde{\vec{\theta}}) =  F_Q^{-1} (\vr_{\vec{\theta}})$ holds if and only if the following two conditions are satisfied:
\begin{subequations}
\begin{align}
    \label{eq:QCRBinequality_saturation:i}
    {\rm (i)}
    \ \
    &\Tilde{E}_i \vr_{\vec{\theta}} = \sum_j [F_Q^{-1}(\vr_{\vec{\theta}})]_{ij} L_j \vr_{\vec{\theta}},
    \ \ 
    \forall i,
    \\
    \label{eq:QCRBinequality_saturation:ii}
    {\rm (ii)}
    \ \
    &E_\omega [ \Tilde{E}_i - \tilde{e}_i (\omega) \eins ] \vr_{\vec{\theta}} = 0,
    \ \ 
    \forall i, \omega,
\end{align}

\end{subequations}
where
\begin{equation}
    \label{eq:tildeE_iL_ie_i}
    \Tilde{E}_i \equiv \sum_\omega \tilde{e}_i(\omega) E_\omega,
    \quad
    \tilde{e}_i(\omega) \equiv \tilde{\theta}_i(\omega) - \theta_i.
\end{equation}
In Appendix~\ref{ap:proofs}, we will prove the inequalities in Eq.~(\ref{eq:QCRBinequality}) and the saturation conditions in Eqs.~(\ref{eq:QCRBinequality_saturation:i}, \ref{eq:QCRBinequality_saturation:ii}), for clarity and completeness of this manuscript.

We have several remarks. First, taking $\Tilde{E}_i = \sum_j [F_Q^{-1}]_{ij} L_j$ is sufficient to satisfy the condition in Eq.~(\ref{eq:QCRBinequality_saturation:i}) but not necessary in general. For a full-rank state $\vr_{\vec{\theta}}$, the condition in Eq.~(\ref{eq:QCRBinequality_saturation:i}) reduces to $\Tilde{E}_i = \sum_j [F_Q^{-1}]_{ij} L_j$. Second, if a POVM is a projection-valued measure (PVM), i.e., $E_\omega E_{\eta} = \delta_{\omega \eta}E_\omega$ for all $\omega, \eta$, the condition in Eq.~(\ref{eq:QCRBinequality_saturation:ii}) is automatically satisfied. More generally, the condition in Eq.~(\ref{eq:QCRBinequality_saturation:ii}) is satisfied if and only if, on the support of $\vr_{\vec{\theta}}$, each POVM element $E_\omega$ is supported on a simultaneous eigenspace of the operators $\Tilde{E}_i$ with eigenvalue vector $\vec{\tilde{e}}(\omega) = \vec{\tilde{\theta}}(\omega) - \vec{\theta}$. Finally, we note that $\mathcal{Q}(\vr_{\vec{\theta}}) = 0$ holds for any $M$ if and only if there exist a POVM $\mathsf{E}$ and an estimator $\Tilde{\vec{\theta}}$ to achieve the equality ${\rm Cov}(\vr_{\vec{\theta}}, \mathsf{E}, \Tilde{\vec{\theta}}) =  F_Q^{-1} (\vr_{\vec{\theta}})$.

\subsection{Proof of Eqs.~(\ref{eq:QCRBinequality}, \ref{eq:QCRBinequality_saturation:i}, \ref{eq:QCRBinequality_saturation:ii})}\label{ap:proofs}
Here we first prove the inequality
\begin{equation}\label{eq:CgeqF-1}
    \mathfrak{C}(\vr_{\vec{\theta}}, \mathsf{E}, \Tilde{\vec{\theta}})
    \geq F_Q^{-1} (\vr_{\vec{\theta}}),
\end{equation}
and show that this saturation condition is given by Eq.~(\ref{eq:QCRBinequality_saturation:i}). Then, we prove the inequality
\begin{equation}\label{eq:CovgeqC}
    {\rm Cov}(\vr_{\vec{\theta}}, \mathsf{E}, \Tilde{\vec{\theta}})
    \geq \mathfrak{C}(\vr_{\vec{\theta}}, \mathsf{E}, \Tilde{\vec{\theta}}),
\end{equation}
and show that this saturation condition is given by Eq.~(\ref{eq:QCRBinequality_saturation:ii}).
\subsubsection*{Proof of Eqs.~(\ref{eq:QCRBinequality_saturation:i}, \ref{eq:CgeqF-1})}
\begin{proof}
    Let us begin by considering the matrix $\mathcal{B}$ with elements
    \begin{equation}
        \mathcal{B}_{ij} = \frac{1}{2} \tr[ (\vr_{\vec{\theta}} L_i + L_i \vr_{\vec{\theta}}) \Tilde{E}_j],
    \end{equation}
    where $\Tilde{E}_j = \sum_\omega \tilde{e}_j (\omega) E_\omega$ in Eq.~(\ref{eq:tildeE_iL_ie_i}). Using the locally-unbiased estimator condition, we obtain that $\mathcal{B} = \eins$, since
    \begin{subequations}
    \begin{align}
        \delta_{ij}
        &= \partial_i \braket{\Tilde{\theta}_j}
        \\
        &= \partial_i \left[ \sum_{\omega} p (\omega |\vec{\theta}) \Tilde{\theta}_j (\omega) \right]
        \\
        &= \tr\left[ (\partial_i  \vr_{\vec{\theta}}) \sum_{\omega} E_\omega \Tilde{\theta}_j (\omega)  \right] 
        \\
        &= \tr\left\{ (\partial_i  \vr_{\vec{\theta}}) \sum_{\omega} E_\omega [\Tilde{\theta}_j (\omega) -\theta_j]  \right\} 
        \\
        &= \tr\left[ (\partial_i  \vr_{\vec{\theta}}) \Tilde{E}_j \right] 
        \\
        &= \mathcal{B}_{ij},
    \end{align}
    \end{subequations}
    where we used that $p (\omega |\vec{\theta}) = \tr(\vr_{\vec{\theta}} E_\omega)$, $\tr(\partial_i  \vr_{\vec{\theta}}) = 0$, and $\tilde{e}_i(\omega) = \tilde{\theta}_i(\omega) - \theta_i$.

    Next, we consider Hermitian matrices $X = \sum_{i=1}^m x_i L_i$ and $Y = \sum_{i=1}^m y_i \Tilde{E}_i$, where $x_i, y_i$ are elements of vectors $\vec{x}, \vec{y} \in \mathbb{R}^m$, and $m$ is the number of parameters to be estimated. From a straightforward calculation, we have
    \begin{equation}
        \tr (\vr_{\vec{\theta}} X^2) = \vec{x}^\trans F_Q \vec{x},
        \quad
        \tr (\vr_{\vec{\theta}} Y^2) = \vec{y}^\trans \mathfrak{C} \vec{y},
        \quad
        \tr\left(\vr_{\vec{\theta}} \tfrac{XY + YX}{2} \right) = \vec{x}^\trans \mathcal{B} \vec{y} = \vec{x}^\trans \vec{y},
    \end{equation}
    where the dependencies on $\vec{\theta}$, $ \Tilde{\vec{\theta}}$, the state, and the POVM are omitted here and below. Now, let us apply the Cauchy-Schwarz inequality: $\left[ \tr\left( \sigma \tfrac{AB + BA}{2} \right ) \right]^2 \leq \tr(\sigma A^2) \tr(\sigma B^2)$ holds for any positive-semidefinite matrix $\sigma$ and Hermitian matrices $A, B$. Setting $\sigma = \vr_{\vec{\theta}}$, $A = X$, and $B = Y$ directly leads to
    \begin{equation}
        (\vec{x}^\trans \vec{y})^2 \leq (\vec{x}^\trans F_Q \vec{x}) (\vec{y}^\trans \mathfrak{C} \vec{y}),
    \end{equation}
    for any $\vec{x}, \vec{y} \in \mathbb{R}^m$. Hence, choosing $\vec{x} = F_Q^{-1}  \vec{y}$, we obtain the inequality $\mathfrak{C} \geq F_Q^{-1}$ in Eq.~(\ref{eq:CgeqF-1}).

    To prove the saturation condition in Eq.~(\ref{eq:QCRBinequality_saturation:i}), we remark that the above Cauchy-Schwarz inequality is saturated if and only if there exists a real coefficient $\mu$ such that $\tr[\sigma (A - \mu B)^2] = 0$, equivalently, $(A - \mu B) \sigma = 0$, since $\sigma$ is positive-semidefinite. Accordingly, the inequality $\mathfrak{C} \geq F_Q^{-1}$ is saturated if and only if there exists $\mu \in \mathbb{R}$ such that
    \begin{equation} \label{eq:saturation_condition_i_v1}
        \left( \sum_{i=1}^m x_i L_i - \mu \sum_{i=1}^m y_i \Tilde{E}_i \right) \vr_{\vec{\theta}} = 0.
    \end{equation}
    Since we have chosen $\vec{x} = F_Q^{-1}  \vec{y}$, Eq.~(\ref{eq:saturation_condition_i_v1}) is rewritten as
    \begin{equation}
        \Tilde{E}_i \vr_{\vec{\theta}} = \mu \sum_j [F_Q^{-1}]_{ij} L_j \vr_{\vec{\theta}},
        \ \ 
        \forall i,
    \end{equation}
    where $[F_Q]_{ij} = [F_Q]_{ji}$. Finally, we show that $\mu = 1$:
    \begin{subequations}
    \begin{align}
        \delta_{ij}
        &= \mathcal{B}_{ij}
        \\
        &=\frac{1}{2} \tr (L_i \Tilde{E}_j \vr_{\vec{\theta}} + L_i \vr_{\vec{\theta}} \Tilde{E}_j )
        \\
        &= \frac{\mu}{2} \sum_k
        [F_Q^{-1}]_{jk}
        \tr (L_i L_k \vr_{\vec{\theta}} + L_i \vr_{\vec{\theta}} L_k )
        \\
        &= \mu \sum_k
        [F_Q^{-1}]_{jk} [F_Q]_{ik}
        \\
        &= \mu \, \delta_{ij},
    \end{align}
    \end{subequations}
    where we used that $\sum_{j} [F_Q]_{ik} [F_Q^{-1}]_{kj} = \delta_{ij}$ and $[F_Q]_{ij} = [F_Q]_{ji}$. Hence, we complete the proof.
\end{proof}
\subsubsection*{Proof of Eqs.~(\ref{eq:QCRBinequality_saturation:ii}, \ref{eq:CovgeqC})}
\begin{proof}  
    Let us begin by considering the Hermitian matrix $X = \sum_{i} x_i \Tilde{E}_i$ with $\Tilde{E}_i = \sum_\omega \tilde{e}_i(\omega) E_\omega$ in Eq.~(\ref{eq:tildeE_iL_ie_i}), where $\tilde{e}_i(\omega)$ and $x_i$ are respectively elements of real vectors $\vec{\tilde{e}}(\omega)$ and $\vec{x}$. From a straightforward calculation, we have
    \begin{equation}
        \vec{x}^\trans {\rm Cov} \vec{x}
        = \sum_\omega \tr(\vr_{\vec{\theta}} E_\omega) [\vec{x}^\trans \vec{\tilde{e}}(\omega)]^2,
        \quad
        \vec{x}^\trans \mathfrak{C} \vec{x}
        = \tr(\vr_{\vec{\theta}} X^2),
    \end{equation}
    where the dependencies on $\vec{\theta}$, $ \Tilde{\vec{\theta}}$, the state, and the POVM are omitted here and below.
    
    Now, we introduce the Hermitian matrix
    \begin{equation}
        A_\omega \equiv \sum_{i} x_i [\tilde{e}_i(\omega) \eins - \Tilde{E}_i].
    \end{equation}
    Since $X = \sum_\omega [\vec{x}^\trans \vec{\tilde{e}}(\omega)] E_\omega$ and then $A_\omega = [\vec{x}^\trans \vec{\tilde{e}}(\omega)] \eins - X$, we have that $\sum_\omega A_\omega E_\omega A_\omega = \sum_\omega [\vec{x}^\trans \vec{\tilde{e}}(\omega)]^2 E_\omega - X^2$. Using this relation, the difference between ${\rm Cov}$ and $\mathfrak{C}$ is given by
    \begin{equation} \label{eq:Cov-C:gap}
        \vec{x}^\trans ( {\rm Cov} - \mathfrak{C}) \vec{x}
        = \sum_\omega \tr(\vr_{\vec{\theta}} A_\omega E_\omega A_\omega)
        = \sum_\omega \tr(\mathcal{A}_\omega^\dagger \mathcal{A}_\omega),
    \end{equation}
    where we denoted $\mathcal{A}_\omega \equiv \sqrt{E_\omega} A_\omega \sqrt{\vr_{\vec{\theta}}}$. Since $\mathcal{A}_\omega^\dagger \mathcal{A}_\omega$ is positive-semidefinite, Eq.~(\ref{eq:Cov-C:gap}) is also nonnegative. Thus we arrive at the inequality ${\rm Cov} \geq \mathfrak{C}$ in Eq.~(\ref{eq:CovgeqC}).

    To prove the saturation condition in Eq.~(\ref{eq:QCRBinequality_saturation:ii}), we notice that the inequality ${\rm Cov} \geq \mathfrak{C}$ is saturated if and only if $\mathcal{A}_\omega$ vanish for all $\omega$, equivalently, 
    \begin{equation} \label{eq:saturation:ii_v1}
         \sum_{i} x_i
         E_\omega [\tilde{e}_i(\omega) \eins - \Tilde{E}_i] \vr_{\vec{\theta}} = 0,
         \ \ 
         \forall \omega,
    \end{equation}
    where we used the fact that $E_\omega$ and $\vr_{\vec{\theta}}$ are positive-semidefinite. Note that Eq.~(\ref{eq:saturation:ii_v1}) holds for all $x_i$ if and only if the componentwise condition as Eq.~(\ref{eq:QCRBinequality_saturation:ii}) holds. Hence, we complete the proof. 
\end{proof}
\subsection{Proof of necessity and sufficiency for the EC condition}\label{ap:Derivation_extended_commutativity}
Before demonstrating that the EC condition discussed in Section~\ref{subsubsec:ECcondition} is both necessary and sufficient for the saturation of the QCR bound~\cite{matsumoto2005geometricalPhD,suzuki2020quantum,conlon2022gap}, we summarize several notations that are used in the main text and below:
\begin{itemize}
    \item 
    $\mathcal{H}$ is a Hilbert space on which the state $\vr_{\vec{\theta}}$ and the POVM $\mathsf{E} = \{E_\omega\}$ are defined.
    \item
    $\mathcal{H}^\prime$ is an extended Hilbert space with $\mathcal{H}^\prime \supseteq \mathcal{H}$, and $\mathcal{V}$ is an isometry from $\mathcal{H}$ to $\mathcal{H}^\prime$ with $\mathcal{V}^\dagger \mathcal{V} = \eins \in \mathcal{H}$.
    
    \item 
    $\vr_{\vec{\theta}}^\prime = \mathcal{V} \vr_{\vec{\theta}} \mathcal{V}^\dagger$ is the isometric extension of the state $\vr_{\vec{\theta}} \in \mathcal{H}$ to $\mathcal{H}^\prime$.
    
    \item 
    $\Pi_\mathcal{V} \equiv \mathcal{V} \mathcal{V}^\dagger$ is an orthogonal projector (i.e., $\Pi_\mathcal{V}^2 = \Pi_\mathcal{V} = \Pi_\mathcal{V}^\dagger$) onto the subspace $\mathcal{V}(\mathcal{H}) \subseteq \mathcal{H}^\prime$ corresponding to the isometric image of $\mathcal{H}$ into $\mathcal{H}^\prime$. Here, $\mathcal{V}(\mathcal{H}) = \mathcal{V} ({\rm supp}(\vr_{\vec{\theta}}) ) + \mathcal{V} ( {\rm ker}(\vr_{\vec{\theta}}) )$ with $\mathcal{V} ( {\rm supp}(\vr_{\vec{\theta}}) ) = {\rm supp}(\vr_{\vec{\theta}}^\prime)$.
\end{itemize}

To proceed, let us note two facts that will be essential in the following:
\begin{itemize}
    \item 
    \textbf{Naimark's dilation theorem:} (see Ref.~\cite{holevo2011probabilistic,watrous2018theory})
    For every POVM $\mathsf{E} = \{E_\omega \} \in \mathcal{H}$, there exist an extended Hilbert space $\mathcal{H}^\prime$, an isometry $\mathcal{V}$, and a projection-valued measure (PVM) $\mathsf{E}^\prime = \{ E_\omega^\prime \}\in \mathcal{H}^\prime$ with the orthogonality relation $E_\omega^\prime E_\eta^\prime = \delta_{\omega \eta} E_\omega^\prime$ for all $\omega, \eta$, such that $E_\omega = \mathcal{V}^\dagger E_\omega^\prime \mathcal{V}$. This directly implies that
    \begin{equation} \label{eq:commute_E_w'}
        [E_\omega^\prime, E_\eta^\prime] = 0,
        \quad
        \forall \omega, \eta.
    \end{equation}
    Within Naimark's dilation, we notice that
    \begin{equation}
        \label{eq:P'=P}
        \tr_{\mathcal{H}^\prime} (\vr_{\vec{\theta}}^\prime E_\omega^\prime)
        = \tr_{\mathcal{H}^\prime} (\mathcal{V} \vr_{\vec{\theta}} \mathcal{V}^\dagger E_\omega^\prime)
        = \tr_\mathcal{H} (\vr_{\vec{\theta}} \mathcal{V}^\dagger E_\omega^\prime \mathcal{V})
        = \tr_\mathcal{H} (\vr_{\vec{\theta}} E_\omega),
    \end{equation}
    where $\tr_{\mathcal{H}}$ and $\tr_{\mathcal{H}^\prime}$ respectively represent the partial traces over $\mathcal{H}$ and $\mathcal{H}^\prime$.

    \item 
    \textbf{Extended SLD operator:}
    Extending Eq.~(\ref{eq:def_SLD_another}) in the main text to $\mathcal{H}^\prime$, we have that an Hermitian operator $\mathfrak{L}_i^\prime \in \mathcal{H}^\prime$ is called the SLD operator of $\vr_{\vec{\theta}}^\prime$ if and only if it satisfies the condition
    \begin{equation} \label{eq:def_extended_SLD_another}
        \mathfrak{L}_i^\prime \vr_{\vec{\theta}}^\prime = \mathcal{V} L_i \mathcal{V}^\dagger \vr_{\vec{\theta}}^\prime,
    \end{equation}
    corresponding to Eq.~(\ref{eq:H'space}) of the main text. We have that $\mathfrak{L}_i^\prime$ satisfies the bona-fide property $\partial_i \vr_{\vec{\theta}}^\prime = (1/2) (\mathfrak{L}_i^\prime \vr_{\vec{\theta}}^\prime +  \vr_{\vec{\theta}}^\prime \mathfrak{L}_i^\prime)$ if and only if Eq.~(\ref{eq:def_extended_SLD_another}) holds, see Eqs.~(\ref{eq:SLD'extendedwithK'}, \ref{eq:SLD'_block}) for more details. Here we also notice that
    \begin{equation}
        \label{eq:inv_L'L'=LL}
        \tr_{\mathcal{H}^\prime} (\vr_{\vec{\theta}}^\prime \mathfrak{L}_i^\prime \mathfrak{L}_j^\prime)
        = \tr_{\mathcal{H}} (\vr_{\vec{\theta}} L_i L_j).
    \end{equation}
    Then the QFIM remains invariant under the isometric extension of the state $\vr_{\vec{\theta}}$ if and only if Eq.~(\ref{eq:def_extended_SLD_another}) holds, i.e., $F_Q(\vr_{\vec{\theta}}^\prime) = F_Q(\vr_{\vec{\theta}})$.
    \end{itemize}

We remark that Naimark's dilation is not an isometric extension. In fact, $E_\omega^\prime$ is \textit{not} defined as the isometric extension of $E_\omega$, unlike $\vr_{\vec{\theta}}^\prime$ and $L_i^\prime$. That is, one \textit{cannot} write $E_\omega^\prime$ as $E_\omega^\prime = \mathcal{V} E_\omega \mathcal{V}^\dagger$ in general, because the isometry $\mathcal{V}$ is not unitary, i.e., $\mathcal{V} \mathcal{V}^\dagger$ is not the identity.

\subsubsection*{Proof of necessity} 
\begin{proof}
    Here we prove that if $\mathcal{Q}(\vr_{\vec{\theta}}) = 0$ holds, then there exist an isometry $\mathcal{V}$ and a set of SLD operators $\{ \mathfrak{L}_i^\prime \}$ on the extended Hilbert space $\mathcal{H}^\prime$, satisfying $\mathfrak{L}_i^\prime \vr_{\vec{\theta}}^\prime = \mathcal{V} L_i \mathcal{V}^\dagger \vr_{\vec{\theta}}^\prime$ in Eq.~(\ref{eq:def_extended_SLD_another}), such that $[\mathfrak{L}_i^\prime, \mathfrak{L}_j^\prime] = 0$ for all $i,j$. Let us begin by recalling that $E_\omega^\prime \in \mathcal{H}^\prime$ are orthogonal projectors according to Naimark's dilation theorem discussed above. We can thus extend Eq.~(\ref{eq:tildeE_iL_ie_i}) to $\mathcal{H}^\prime$ as follows:
    \begin{equation} \label{eq:def_tildeE_i'}
        \Tilde{E}_i^\prime = \sum_\omega \tilde{e}_i(\omega) E_\omega^\prime \in \mathcal{H}^\prime,
    \end{equation} 
    with $\Tilde{E}_i = \mathcal{V}^\dagger \Tilde{E}_i^\prime \mathcal{V}$ and $E_\omega = \mathcal{V}^\dagger E_\omega^\prime \mathcal{V}$. We notice that $[\Tilde{E}_i^\prime, \Tilde{E}_j^\prime] = 0$, due to Eq.~(\ref{eq:commute_E_w'}).

    Next, we define a Hermitian operator 
    \begin{equation} \label{eq:SLD_ansatz}
        \mathfrak{L}_i^\prime \equiv \sum_{j} [F_Q]_{ij} \Tilde{E}_j^\prime.
    \end{equation}
    We can immediately see that $[\mathfrak{L}_i^\prime, \mathfrak{L}_j^\prime] = 0$, since $[\Tilde{E}_i^\prime, \Tilde{E}_j^\prime] = 0$. In the following, we show that the operator $\mathfrak{L}_i^\prime$ in Eq.~(\ref{eq:SLD_ansatz}) is an SLD operator in the extended space, namely, it satisfies the condition Eq.~(\ref{eq:def_extended_SLD_another}). To this end, let us first note that Eq.~(\ref{eq:SLD_ansatz}) obeys the following relation that we will derive later:
    \begin{equation}
        \label{eq:PiLPirho=Lrho}
        \mathfrak{L}_i^\prime \vr_{\vec{\theta}}^\prime
        = \Pi_{\mathcal{V}} \mathfrak{L}_i^\prime \vr_{\vec{\theta}}^\prime.
    \end{equation}
    
    Inserting Eq.~(\ref{eq:SLD_ansatz}) into Eq.~(\ref{eq:PiLPirho=Lrho}) and using $\vr_{\vec{\theta}}^\prime = \Pi_{\mathcal{V}} \vr_{\vec{\theta}}^\prime$, we obtain
    \begin{subequations}
    \begin{align}
        \mathfrak{L}_i^\prime \vr_{\vec{\theta}}^\prime
        &= \sum_{j} [F_Q]_{ij} \Pi_{\mathcal{V}} \Tilde{E}_j^\prime \Pi_{\mathcal{V}} \vr_{\vec{\theta}}^\prime
        \\
        &= \sum_{j} [F_Q]_{ij} \mathcal{V} (\mathcal{V}^\dagger \Tilde{E}_j^\prime \mathcal{V}) \mathcal{V}^\dagger \vr_{\vec{\theta}}^\prime
        \\
        &= \sum_{j} [F_Q]_{ij} \mathcal{V} \Tilde{E}_j \vr_{\vec{\theta}} \mathcal{V}^\dagger
        \\
        &= \sum_{j,k} [F_Q]_{ij} [F_Q^{-1}]_{jk}\mathcal{V} L_k \vr_{\vec{\theta}} \mathcal{V}^\dagger
        \\
        &= \mathcal{V} L_i \vr_{\vec{\theta}} \mathcal{V}^\dagger
        \\
        &= \mathcal{V} L_i \mathcal{V}^\dagger \vr_{\vec{\theta}}^\prime,
    \end{align}
    \end{subequations}
    where we used the saturation condition in Eq.~(\ref{eq:QCRBinequality_saturation:i}) and $\sum_{j} [F_Q]_{ij} [F_Q^{-1}]_{jk} = \delta_{ik}$.

    Finally, we below show that Eq.~(\ref{eq:SLD_ansatz}) satisfies Eq.~(\ref{eq:PiLPirho=Lrho}), equivalently $\Pi_{\mathcal{V}}^\perp \mathfrak{L}_i^\prime \vr_{\vec{\theta}}^\prime = 0$, where $\Pi_{\mathcal{V}}^\perp = \eins - \Pi_{\mathcal{V}}$ is the projector orthogonal to $\Pi_{\mathcal{V}}$. To show this, it is enough to prove that 
    \begin{equation}
        \label{eq:consition_orthogonal_PErho=0}
        \Pi_{\mathcal{V}}^\perp \Tilde{E}_i^\prime \vr_{\vec{\theta}}^\prime = 0.
    \end{equation}
    Let us begin by rewriting the covariance matrix as
    \begin{subequations}
    \label{eq:Cov_refomulation}
    \begin{align}
        [{\rm Cov}(\vr_{\vec{\theta}}, \mathsf{E}, \Tilde{\vec{\theta}})]_{ij}
        &=  \sum_{\omega} \tr(\vr_{\vec{\theta}}^\prime E_\omega^\prime)
        \tilde{e}_i(\omega) \tilde{e}_j(\omega)
        \\
        &=  \sum_{\omega,\eta} \tr(\vr_{\vec{\theta}}^\prime E_\omega^\prime E_\eta^\prime)
        \tilde{e}_i(\omega) \tilde{e}_j(\eta)
        \\
        &= \tr(\vr_{\vec{\theta}}^\prime \Tilde{E}_i^\prime \Tilde{E}_j^\prime)
        \\
        \label{eq:Cov_decomposite}
        &= \tr(\vr_{\vec{\theta}}^\prime \Tilde{E}_i^\prime \Pi_{\mathcal{V}} \Tilde{E}_j^\prime )
        + \tr(\vr_{\vec{\theta}}^\prime \Tilde{E}_i^\prime  \Pi_{\mathcal{V}}^\perp \Tilde{E}_j^\prime ),
    \end{align}
    \end{subequations}
    where we used Eq.~(\ref{eq:P'=P}), $E_\omega^\prime E_\eta^\prime = \delta_{\omega \eta} E_\omega^\prime$, and Eq.~(\ref{eq:def_tildeE_i'}).
    
    Now we rewrite the first term in Eq.~(\ref{eq:Cov_decomposite}) as
    \begin{subequations}
    \label{eq:QFIM_diagonal_refor}
    \begin{align}
        \tr(\vr_{\vec{\theta}}^\prime \Tilde{E}_i^\prime  \Pi_{\mathcal{V}} \Tilde{E}_j^\prime )
        &= \tr(\vr_{\vec{\theta}}^\prime \Pi_{\mathcal{V}} \Tilde{E}_i^\prime  \Pi_{\mathcal{V}} \Tilde{E}_j^\prime  \Pi_{\mathcal{V}})
        \\
        &= \tr[\mathcal{V}^\dagger \vr_{\vec{\theta}}^\prime \mathcal{V} (\mathcal{V}^\dagger \Tilde{E}_i^\prime  \mathcal{V}) (\mathcal{V}^\dagger \Tilde{E}_j^\prime  \mathcal{V})]
        \\
        &= \tr(\vr_{\vec{\theta}} \Tilde{E}_i \Tilde{E}_j)
        \\
        &= \sum_{k,l} [F_Q^{-1}]_{ik} [F_Q^{-1}]_{jl} \tr(\vr_{\vec{\theta}} L_k L_l)
        \\
        &= \sum_{k,l} [F_Q^{-1}]_{ik} [F_Q^{-1}]_{jl} [F_Q]_{kl}
        \\
        &= [F_Q^{-1}]_{ij}.
    \end{align}
    \end{subequations}
    Here we used $\vr_{\vec{\theta}}^\prime = \Pi_{\mathcal{V}} \vr_{\vec{\theta}}^\prime \Pi_{\mathcal{V}}$ and the saturation condition in Eq.~(\ref{eq:QCRBinequality_saturation:i}) together with the trace cyclicity, and employed the fact that $\tr(\vr_{\vec{\theta}}^\prime \Tilde{E}_i^\prime  \Pi_{\mathcal{V}} \Tilde{E}_j^\prime ) = \tr(\vr_{\vec{\theta}}^\prime \Tilde{E}_j^\prime  \Pi_{\mathcal{V}} \Tilde{E}_i^\prime )$ due to the symmetry of the covariance matrix. Also, the second term in Eq.~(\ref{eq:Cov_decomposite}) are given by
    \begin{equation}
        \label{eq:anothergap_diagonal_refor}
        \tr(\vr_{\vec{\theta}}^\prime \Tilde{E}_i^\prime  \Pi_{\mathcal{V}}^\perp \Tilde{E}_j^\prime )
        = \tr\left[ \left(\Pi_{\mathcal{V}}^\perp \Tilde{E}_i^\prime \sqrt{\vr_{\vec{\theta}}^\prime} \right)^\dagger \left(\Pi_{\mathcal{V}}^\perp \Tilde{E}_j^\prime \sqrt{\vr_{\vec{\theta}}^\prime} \right) \right]
        = \tr(\mathcal{Z}_i^\dagger \mathcal{Z}_j),
    \end{equation}
    where we denoted $\mathcal{Z}_i \equiv \Pi_{\mathcal{V}}^\perp \Tilde{E}_i^\prime \sqrt{\vr_{\vec{\theta}}^\prime}$.

    Since we have supposed that $\mathcal{Q}(\vr_{\vec{\theta}}) = 0$ holds, i.e., ${\rm Cov}(\vr_{\vec{\theta}}, \mathsf{E}, \Tilde{\vec{\theta}}) = F_Q^{-1} (\vr_{\vec{\theta}})$, both diagonal elements must coincide with each other. According to Eqs.~(\ref{eq:Cov_refomulation}, \ref{eq:QFIM_diagonal_refor}, \ref{eq:anothergap_diagonal_refor}), we have that $\tr(\mathcal{Z}_i^\dagger \mathcal{Z}_i) = 0$ for all $i$. This implies that $\mathcal{Z}_i = 0$, because $\mathcal{Z}_i^\dagger \mathcal{Z}_i$ is positive-semidefinite. Finally, multiplying $\sqrt{\vr_{\vec{\theta}}^\prime}$ by $\mathcal{Z}_i$ from the right-hand side leads to $\Pi_{\mathcal{V}}^\perp \Tilde{E}_i^\prime \vr_{\vec{\theta}}^\prime = 0$. We thus arrive at Eq.~(\ref{eq:consition_orthogonal_PErho=0}). Hence, we complete the proof.
\end{proof}
\subsubsection*{Proof of sufficiency}
\begin{proof}
    Here we prove that if there exist an isometry $\mathcal{V}$ and a set of SLD operators $\{ \mathfrak{L}_i^\prime \}$ on the extended Hilbert space $\mathcal{H}^\prime$, which satisfy $\mathfrak{L}_i^\prime \vr_{\vec{\theta}}^\prime = \mathcal{V} L_i \mathcal{V}^\dagger \vr_{\vec{\theta}}^\prime$ in Eq.~(\ref{eq:def_extended_SLD_another}) and commute, namely $[\mathfrak{L}_i^\prime, \mathfrak{L}_j^\prime] = 0$ for all $i,j$, then $\mathcal{Q}(\vr_{\vec{\theta}}) = 0$ holds. Let us begin by defining a Hermitian operator
    \begin{equation} \label{eq:tildeLprime}
        \mathcal{E}_i^\prime \equiv \sum_j [F_Q^{-1}]_{ij} \mathfrak{L}_j^\prime \in \mathcal{H}^\prime.
    \end{equation}
    We notice that $[\mathcal{E}_i^\prime, \mathcal{E}_j^\prime] = 0$ for all $i,j$ due to $[\mathfrak{L}_i^\prime, \mathfrak{L}_j^\prime] = 0$. Thus, according to the spectral theorem, every $\mathcal{E}_i^\prime$ can be written in the same basis
    \begin{equation} \label{eq:spectral_tildeL_i}
        \mathcal{E}_i^\prime
        = \sum_{\omega} \lambda_{i}(\omega) E_\omega^\prime,
    \end{equation}
    where $\lambda_{i}(\omega)$ are real eigenvalues with orthogonal projectors $E_\omega^\prime$ according to Naimark’s dilation theorem. We emphasize that $\lambda_{i}(\omega)$ depends on the parameter $\vec{\theta}$ since $F_Q$ and $\mathfrak{L}_j^\prime$ depend on $\vec{\theta}$ via the state $\vr_{\vec{\theta}}$. However, no assumption about the explicit dependence of $\lambda_{i}(\omega)$ on $\vec{\theta}$ can be made at this stage. 

    Next, we define a real and symmetric matrix $\Lambda$ with elements
    \begin{equation} \label{eq:def_Lambda_ij}
        [\Lambda]_{ij} \equiv \sum_{\omega} p (\omega |\vec{\theta}) \lambda_{i}(\omega) \lambda_{j}(\omega),
    \end{equation}
    where $p (\omega |\vec{\theta}) = \tr(\vr_{\vec{\theta}} E_\omega)$. In the following, we will show two relations
    \begin{subequations}
    \begin{align}
        \label{eq:Lambda=F_Q}
        \Lambda &= F_Q^{-1},
        \\
        \label{eq:Lambda=F_Q-1FCFQ-1}
        \Lambda &= F_Q^{-1} F_C F_Q^{-1},
    \end{align}
    \end{subequations}
    where $F_C$ is the CFIM in Eq.~(\ref{eq:CFIM_form}) in the main text with elements $[F_C]_{ij} = \sum_{\omega} [1/p(\omega | \vec{\theta})] [\partial_i p(\omega | \vec{\theta})] [\partial_j p(\omega | \vec{\theta})]$. Combining Eqs.~(\ref{eq:Lambda=F_Q}, \ref{eq:Lambda=F_Q-1FCFQ-1}) directly leads to $F_C = F_Q$. Thus, we can conclude that $\mathcal{Q}(\vr_{\vec{\theta}}) = 0$. 

    Eq.~(\ref{eq:Lambda=F_Q}) is shown as follows: Let us begin by noting Eqs.~(\ref{eq:P'=P}, \ref{eq:inv_L'L'=LL}) and the condition $\mathfrak{L}_i^\prime \vr_{\vec{\theta}}^\prime = \mathcal{V} L_i \mathcal{V}^\dagger \vr_{\vec{\theta}}^\prime$ in Eq.~(\ref{eq:def_extended_SLD_another}). Thus we can rewrite the $(i,j)$-element of $\Lambda$ as
    \begin{subequations}
    \begin{align}
        [\Lambda]_{ij}
        &=\sum_{\omega} \tr(\vr_{\vec{\theta}}^\prime E_\omega^\prime) \lambda_{i}(\omega) \lambda_{j}(\omega)
        \\
        &= \sum_{\omega,\eta} \tr(\vr_{\vec{\theta}}^\prime E_\omega^\prime E_\eta^\prime) \lambda_{i}(\omega) \lambda_{j}(\eta)
        \\
        &= \tr(\vr_{\vec{\theta}}^\prime \mathcal{E}_i^\prime \mathcal{E}_j^\prime)
        \\
        &= \sum_{k,l} [F_Q^{-1}]_{ik} [F_Q^{-1}]_{jl} \tr(\vr_{\vec{\theta}}^\prime \mathfrak{L}_k^\prime \mathfrak{L}_l^\prime)
        \\
        &= \sum_{k,l} [F_Q^{-1}]_{ik} [F_Q^{-1}]_{jl} \tr(\vr_{\vec{\theta}} L_k L_l)
        \\
        &= [F_Q^{-1}]_{ij},
    \end{align}
    \end{subequations}
    where we used $E_\omega^\prime E_\eta^\prime = \delta_{\omega \eta} E_\omega^\prime$, Eqs.~(\ref{eq:tildeLprime}, \ref{eq:spectral_tildeL_i}), and the symmetry of the matrix $\Lambda$. Then we obtain Eq.~(\ref{eq:Lambda=F_Q}).

    Eq.~(\ref{eq:Lambda=F_Q-1FCFQ-1}) is shown as follows: To proceed, let us show that $\lambda_{i}(\omega)$ is given by
    \begin{equation} \label{eq:lambda_i_form}
        \lambda_{i}(\omega) = \frac{1}{p (\omega |\vec{\theta})} \sum_{j} [F_Q^{-1}]_{ij} \partial_j p (\omega |\vec{\theta}).
    \end{equation}
    In fact, inserting Eq.~(\ref{eq:lambda_i_form}) into Eq.~(\ref{eq:def_Lambda_ij}) directly leads to $\Lambda = F_Q^{-1} F_C F_Q^{-1}$. To derive Eq.~(\ref{eq:lambda_i_form}), we first multiply $[F_Q^{-1}]_{ij}$ by $\partial_j \vr_{\vec{\theta}}^\prime$ and then take a sum over $j$:
    \begin{subequations}
    \begin{align}
        \sum_j [F_Q^{-1}]_{ij} \partial_j \vr_{\vec{\theta}}^\prime
        &= \frac{1}{2} \sum_j [F_Q^{-1}]_{ij} (\mathfrak{L}_j^\prime \vr_{\vec{\theta}}^\prime + \vr_{\vec{\theta}}^\prime \mathfrak{L}_j^\prime)
        \\
        &= \frac{1}{2} (\mathcal{E}_i^\prime \vr_{\vec{\theta}}^\prime + \vr_{\vec{\theta}}^\prime \mathcal{E}_i^\prime)
        \\
        &= \frac{1}{2} \sum_{\omega} \lambda_{i}(\omega) (E_\omega^\prime \vr_{\vec{\theta}}^\prime + \vr_{\vec{\theta}}^\prime E_\omega^\prime),
    \end{align}
    \end{subequations}
    where we used the condition $\mathfrak{L}_i^\prime \vr_{\vec{\theta}}^\prime = \mathcal{V} L_i \mathcal{V}^\dagger \vr_{\vec{\theta}}^\prime$ and Eqs.~(\ref{eq:tildeLprime}, \ref{eq:spectral_tildeL_i}). 
    
    Next, multiplying $E_\eta^\prime$ by this and taking the trace yields
    \begin{subequations}
    \begin{align}
        \sum_j [F_Q^{-1}]_{ij} \tr[(\partial_j \vr_{\vec{\theta}}^\prime) E_\eta^\prime]
        &= \frac{1}{2}\sum_{\omega} \lambda_{j}(\omega) \tr(E_\eta^\prime E_\omega^\prime \vr_{\vec{\theta}}^\prime + E_\eta^\prime \vr_{\vec{\theta}}^\prime E_\omega^\prime)
        \\
        &= \lambda_{j}(\eta) \tr(\vr_{\vec{\theta}}^\prime E_\eta^\prime)
        \\
        &= \lambda_{j}(\eta) p (\eta |\vec{\theta}).
    \end{align}
    \end{subequations}
    Using $\tr[(\partial_j \vr_{\vec{\theta}}^\prime) E_\eta^\prime] = \partial_j p (\eta |\vec{\theta})$ and relabeling $\eta$ with $\omega$, we obtain Eq.~(\ref{eq:lambda_i_form}).
    
    Finally, we can now demonstrate the explicit functional dependence of $\lambda_{i}(\omega)$ on $\vec{\theta}$: it \textit{must} be
    \begin{equation}
        \label{eq:lambda_form_must}
        \lambda_{i}(\omega) = \Tilde{\theta}_i (\omega) - \theta_i.
    \end{equation}
    To see this, we note that the optimal locally-unbiased estimator $\Tilde{\vec{\theta}}$ to achieve the saturation of the classical Cramér-Rao inequality ${\rm Cov}(\vr_{\vec{\theta}}, \mathsf{E}, \Tilde{\vec{\theta}}) = F_C^{-1}$ can be given by~\cite{kay1993fundamentals,lehmann1998theory}:
    \begin{equation}\label{eq:optimal_estimator_form_FC}
        \Tilde{\theta}_i (\omega) = \theta_i + \sum_{j} [F_C^{-1}]_{ij} \partial_j {\rm ln} p(\omega| \vec{\theta}).
    \end{equation}
    Since we have shown that $\mathcal{Q}(\vr_{\vec{\theta}}) = 0$, we can replace $F_Q$ in Eq.~(\ref{eq:lambda_i_form}) with $F_C$. Comparing such a form with Eq.~(\ref{eq:optimal_estimator_form_FC}), we can arrive at Eq.~(\ref{eq:lambda_form_must}).
\end{proof}

\section{Derivation of Eq.~(\ref{eq:FC<=FQ})} \label{ap:derivation_FC<=FQ}
Here we show that the matrix inequality $F_C(\vr_{\vec{\theta}}, \mathsf{E}) \leq F_Q(\vr_{\vec{\theta}})$ in Eq.~(\ref{eq:FC<=FQ}) in the main text can be derived only using the classical and quantum Cramér-Rao (CCR and QCR) inequalities: ${\rm Cov}(\vr_{\vec{\theta}}, \mathsf{E}, \Tilde{\vec{\theta}}) \geq (1/n) F_C^{-1}(\vr_{\vec{\theta}}, \mathsf{E})$ and ${\rm Cov}(\vr_{\vec{\theta}}, \mathsf{E}, \Tilde{\vec{\theta}}) \geq (1/n) F_Q^{-1}(\vr_{\vec{\theta}})$, where $\Tilde{\vec{\theta}}(\vec{\omega}) = \{ \Tilde{\theta}_1, \ldots, \Tilde{\theta}_m \}$ is the locally-unbiased estimator and $p(\vec{\omega} | \vec{\theta}) = \prod_{k=1}^{n} p(\omega_k | \vec{\theta})$ is the joint probability distribution, corresponding to a sequence of $n$ independent outcomes $\vec{\omega} = \{ \omega_1, \ldots, \omega_n \}$ in the single-copy case $\nu=1$, as explained in Section~\ref{sec:II.A}. (also see Appendix~\ref{ap:notations_results} for other notations).

\begin{proof}
    We can prove $F_C(\vr_{\vec{\theta}}, \mathsf{E}) \leq F_Q(\vr_{\vec{\theta}})$ by contradiction as follows. First, recall that the optimal estimator $\Tilde{\vec{\theta}}$ to achieve the CCR inequality can be given by $\Tilde{\theta}_i = \theta_i + \sum_{j=1}^m [F_C^{-1}(\vr_{\vec{\theta}}, \mathsf{E})]_{ij} \partial_j {\rm ln} \prod_{k=1}^{n} p(\omega_k | \vec{\theta})$~\cite{kay1993fundamentals,lehmann1998theory}. This means that there always exists $\Tilde{\vec{\theta}}$ such that the CCR inequality is saturated for all $\vr_{\vec{\theta}}$ and $\mathsf{E}$. Now, let us take such an optimal estimator $\Tilde{\vec{\theta}}_*$, i.e., ${\rm Cov}(\vr_{\vec{\theta}}, \mathsf{E}, \Tilde{\vec{\theta}}_*) = (1/n) F_C^{-1}(\vr_{\vec{\theta}}, \mathsf{E})$.
    
    To proceed, suppose that there are certain POVMs and quantum states such that $\tr[M F_Q^{-1} (\vr_{\vec{\theta}})] > \tr[M F_C^{-1}(\vr_{\vec{\theta}}, \mathsf{E})]$, where $M$ is an $m \times m$ real, symmetric, and positive-definite matrix. However, since $\tr[M {\rm Cov}(\vr_{\vec{\theta}}, \mathsf{E}, \Tilde{\vec{\theta}}_*)] = (1/n) \tr[M F_C^{-1}(\vr_{\vec{\theta}}, \mathsf{E})]$ holds, this is in contradiction to the QCR inequality. Thus, $\tr[M F_Q^{-1} (\vr_{\vec{\theta}})] \leq \tr[M F_C^{-1}(\vr_{\vec{\theta}}, \mathsf{E})]$ should hold for all POVMs and quantum states. This scalar inequality allows us to arrive at the matrix inequality $F_C(\vr_{\vec{\theta}}, \mathsf{E}) \leq F_Q(\vr_{\vec{\theta}})$ by following Lemma~2 in Ref.~\cite{yang2019optimal}. 
\end{proof}

\section{Converse implication (I$\inv$) in Eq.~(\ref{eq:several_hierarchy})}
\subsection{Weak commutativity implies strong commutativity for pure states}\label{ap:proof_Iinv}
Here we prove that the converse implication (I$\inv$) in Eq.~(\ref{eq:several_hierarchy}) in the main text holds for pure states. Equivalently, it is sufficient to show that, for pure states, the WC condition implies the SC condition. To this end, we recall the results of Refs.~\cite{matsumoto2002new,pezze2017optimal}, which establish that whenever the WC condition holds for a pure state, there exists a set of optimal projective measurements $\{E_\omega\}$ such that $\mathcal{Q}(\ket{\psi_{\vec{\theta}}}) = 0$. Since these measurement operators are mutually orthogonal projectors, the argument used in the \textit{Proof of necessity} in Appendix~\ref{ap:Derivation_extended_commutativity} can be applied directly, without invoking Naimark's dilation theorem, to show that the corresponding SLD operators commute.

More explicitly, for the Hermitian operators $\mathfrak{L}_i \equiv \sum_{j} [F_Q]_{ij} \Tilde{E}_j$ where $\Tilde{E}_j$ is defined in Eq.~(\ref{eq:tildeE_iL_ie_i}), one finds that they satisfy $\mathfrak{L}_i \ket{\psi_{\vec{\theta}}} = L_i \ket{\psi_{\vec{\theta}}}$ due to the saturation condition in Eq.~(\ref{eq:QCRBinequality_saturation:i}). Hence, $\mathfrak{L}_i$ are valid SLD operators. Since $\Tilde{E}_j$ mutually commute, $\mathfrak{L}_i$ also commute. Therefore, the SC condition holds whenever the WC condition is satisfied, establishing the converse implication (I$\inv$) for pure states.

\subsection{Extended commutativity cannot imply strong commutativity for rank-deficient states} \label{ap:Iinv_counterexample}
Here we present a counterexample demonstrating that the converse implication (I$\inv$) in Eq.~(\ref{eq:several_hierarchy}) in the main text does not generally hold for rank-deficient states. Consider the rank-two state in a four-dimensional system: $\vr_{\rm I} = p \ket{00}\! \bra{00} + (1-p) \ket{01}\! \bra{01} \in \mathbb{C}^4$ for $p \in [0,1]$. The two-dimensional kernel space of $\vr_{\rm I}$ is spanned by $\{ \ket{10}, \ket{11} \}$. Following an approach similar to that of Ref.~\cite{conlon2025role}, we introduce two SLD operators ($L_1, L_2$) and two extended SLD operators ($\mathfrak{L}_1^\prime, \mathfrak{L}_2^\prime$), written in the block matrix forms of Eqs.~(\ref{eq:SLDcomponents}, \ref{eq:SLD'_block}):
\begin{subequations}
\begin{alignat}{3}
    \label{eq:Li_line}
    L_1 &= \left[
    \begin{array}{cccc}
    0 & 0 & -\tfrac{1}{\sqrt{3}} & 0 \\
    0 & 0 & \tfrac{2}{\sqrt{15}} & -2 \sqrt{\tfrac{3}{5}} \\
    -\tfrac{1}{\sqrt{3}} & \tfrac{2}{\sqrt{15}} & u_1 & x_1 + i y_1 \\
    0 & -2 \sqrt{\tfrac{3}{5}} & x_1 -i y_1 & v_1 \\
    \end{array}
    \right],
    &\quad
    L_2 &= \left[
    \begin{array}{cccc}
    0 & 0 & -\tfrac{2}{\sqrt{3}} & 0 \\
    0 & 0 & \tfrac{4}{\sqrt{15}} & \sqrt{\tfrac{3}{5}} \\
    -\tfrac{2}{\sqrt{3}} & \tfrac{4}{\sqrt{15}} & u_2 & x_2 + i y_2 \\
    0 & \sqrt{\tfrac{3}{5}} & x_2 -i y_2 & v_2 \\
    \end{array}
    \right],
    \\
    \label{eq:Li'_line}
    \mathfrak{L}_1^\prime &= \left[
    \begin{array}{ccccc}
    0 & 0 & -\tfrac{1}{\sqrt{3}} & 0 & 0 \\
    0 & 0 & \tfrac{2}{\sqrt{15}} & -2 \sqrt{\tfrac{3}{5}} & 0 \\
    -\tfrac{1}{\sqrt{3}} & \tfrac{2}{\sqrt{15}} & 0 & 0 & -\sqrt{\tfrac{2}{5}} \\
    0 & -2 \sqrt{\tfrac{3}{5}} & 0 & 0 & -2 \sqrt{\tfrac{2}{5}} \\
    0 & 0 & -\sqrt{\tfrac{2}{5}} & -2 \sqrt{\tfrac{2}{5}} & 0 \\
    \end{array}
    \right],
    &\quad
    \mathfrak{L}_2^\prime &= \left[
    \begin{array}{ccccc}
    0 & 0 & -\tfrac{2}{\sqrt{3}} & 0 & 0 \\
    0 & 0 & \tfrac{4}{\sqrt{15}} & \sqrt{\tfrac{3}{5}} & 0 \\
    -\tfrac{2}{\sqrt{3}} & \tfrac{4}{\sqrt{15}} & 0 & 0 & -2 \sqrt{\tfrac{2}{5}} \\
    0 & \sqrt{\tfrac{3}{5}} & 0 & 0 & \sqrt{\tfrac{2}{5}} \\
    0 & 0 & -2 \sqrt{\tfrac{2}{5}} & \sqrt{\tfrac{2}{5}} & 0 \\
    \end{array}
    \right].
\end{alignat}
\end{subequations}
Here, the parameters $u_i, v_i, x_i, y_i \in \mathbb{R}$ for $i=1,2$ in Eq.~(\ref{eq:Li_line}) are free to choose, while the entries involving the one-dimensional ancillary space in Eq.~(\ref{eq:Li'_line}) are fixed.

A straightforward algebraic calculation shows that $[\mathfrak{L}_1^\prime, \mathfrak{L}_2^\prime] = 0$, whereas there exists no choice of the parameters $u_i, v_i, x_i, y_i \in \mathbb{R}$ for which $[L_1, L_2]=0$. Therefore, the EC condition holds while the SC condition does not hold. This demonstrates that the converse implication (I$\inv$) in Eq.~(\ref{eq:several_hierarchy}) does not hold in general. Note that $\tr(\vr_{\rm I} L_i) = 0$ and the corresponding  QFIM is given by
\begin{equation}
    F_Q = \frac{1}{3} \left[
    \begin{array}{cc}
    8-7 p & 2 (2 p-1) \\
    2 (2 p-1) & 5-p \\
    \end{array}
    \right].
\end{equation}

\section{Derivation of Eqs.~(\ref{eq:unitaryGform_basisInd}, \ref{eq:unitaryGform_basisInd_hierarchical})}\label{ap:proof_another_Gform}
Here, without loss of generality, we can set $L_i^{\rm kk}$ to zero, as discussed in Section~\ref{sec:II.C} in the main text.
\subsection{Derivation of Eq.~(\ref{eq:unitaryGform_basisInd})}
\begin{proof}
We begin by noting that the SLD operator $L_i = L_i (\vr_{\vec{\theta}})$ for any state $\vr_{\vec{\theta}}$ can be written as the Lyapunov representation~\cite{paris2009quantum}:
\begin{equation}
    L_i = 2 \int_{0}^{\infty} ds \, e^{- \vr_{\vec{\theta}} s}
    (\partial_i \vr_{\vec{\theta}}) e^{- \vr_{\vec{\theta}} s}.
\end{equation}
For $\vr_{\vec{\theta}} = U_{\vec{\theta}} \vr U_{\vec{\theta}}^\dagger$, using $\partial_i \vr_{\vec{\theta}} = (\partial_i U_{\vec{\theta}}) \vr U_{\vec{\theta}}^\dagger + U_{\vec{\theta}} \vr (\partial_i U_{\vec{\theta}}^\dagger)$ and $\mathcal{H}_i = - i (\partial_i U_{\vec{\theta}}^\dagger) U_{\vec{\theta}} = i U_{\vec{\theta}}^\dagger (\partial_i U_{\vec{\theta}})$, the operator $\mathcal{L}_i \equiv U_{\vec{\theta}}^\dagger L_i U_{\vec{\theta}}$, defined in the proof of Observation~\ref{ob:Gformsummary} in the main text, is thus given by
\begin{equation}
    \mathcal{L}_i
    = 2i \int_{0}^{\infty} ds \, e^{- \vr s}
    (\mathcal{H}_i \vr -\vr \mathcal{H}_i) e^{- \vr s}.
\end{equation}
From a straightforward calculation, we obtain
\begin{equation}
    \tr (\vr \mathcal{L}_i \mathcal{L}_j )
    = 4 \int_{0}^{\infty} \! \! \! ds \int_{0}^{\infty} \! \! \! dt \,
    \tr\Bigl[
    \vr^3 e^{- \vr (s + t)} \mathcal{H}_i e^{- \vr (s + t)} \mathcal{H}_j 
    + \vr e^{- \vr (s + t)} \mathcal{H}_i \vr^2 e^{- \vr (s + t)} \mathcal{H}_j
    \Bigr],
\end{equation}
where we used $\vr e^{- \vr (s + t)} = e^{- \vr (s + t)} \vr$ and $\tr[ABC] = \tr[CAB]$ for operators $A,B,C$. By denoting
\begin{equation}
    \mathcal{T}_\vr
    \equiv
    \int_{0}^{\infty} ds \int_{0}^{\infty} dt \, e^{- \vr (s + t)} \otimes e^{- \vr (s + t)},
\end{equation}
we have that $\tr (\vr \mathcal{L}_i \mathcal{L}_j ) = 4 \tr[ \swap (\vr^3 \otimes \eins \!+\! \vr \otimes \vr^2 ) \mathcal{T}_\vr (\mathcal{H}_i \otimes \mathcal{H}_j ) ]$. Here we used the SWAP trick that $\tr[\swap(A \otimes B)] = \tr(AB)$ for operators $A, B$, and exchanged the integrals with the trace. Recalling that $[W(\vr_{\vec{\theta}})]_{ij} = \tr [\vr (\mathcal{L}_i \mathcal{L}_j - \mathcal{L}_i \mathcal{L}_j) ]$, we can arrive at Eq.~(\ref{eq:unitaryGform_basisInd}) in the main text.
\end{proof}
\subsection{Derivation of Eq.~(\ref{eq:unitaryGform_basisInd_hierarchical})}
\begin{proof}
We begin by noting that the form of $W(\vr_{\vec{\theta}})$, presented in Observation~\ref{ob:Gformsummary} and its proof in the main text, can be rewritten as
\begin{equation}
    [W(\vr_{\vec{\theta}})]_{ij}
    = 4 \sum_{k, l}
    \frac{(\lambda_k - \lambda_{l})^3}{(\lambda_k +\lambda_{l})^2}
    \tr[\swap \Pi_{kl} \mathcal{H}_i  \otimes  \mathcal{H}_j],
\end{equation}
where $\Pi_{kl} = \ket{\psi_k} \! \bra{\psi_k} \otimes \ket{\psi_l} \! \bra{\psi_l}$. To proceed, we denote
\begin{equation}
    \mathcal{W}
    \equiv \sum_{k, l}
    \frac{(\lambda_k - \lambda_{l})^3}{(\lambda_k +\lambda_{l})^2} \Pi_{kl}.
\end{equation}
Then, we have that $[W(\vr_{\vec{\theta}})]_{ij} = 4 \tr[\swap \mathcal{W} \mathcal{H}_i  \otimes  \mathcal{H}_j]$.

Now we notice that $\mathcal{W}$ can be rewritten as
\begin{equation}
    \mathcal{W} = \Bigg[ \sum_{k, l} (\lambda_k \!-\! \lambda_{l})^3 \ \Pi_{kl} \Bigg]
    \Bigg[\sum_{k,l} ( \lambda_{k} \!+\! \lambda_{l} )^{-2} \ \Pi_{kl} \Bigg].
\end{equation}
Using $( \lambda_{k} +  \lambda_{l} )^{-2} = \sum_{x = 0}^\infty (1 - \lambda_{k} - \lambda_{l})^{2x}$ for $\lambda_{k} \in [0,1]$ and $\lambda_{k} + \lambda_{l} > 0$, we obtain that $\mathcal{W} = \lim_{\alpha \to \infty} \mathcal{W}_\alpha$, where
\begin{equation}
     \mathcal{W}_\alpha
     \equiv
     \sum_{x = 0}^\alpha 
     (\vr \otimes \eins - \eins \otimes \vr)^3
     (\eins \otimes \eins - \vr \otimes \eins - \eins \otimes \vr)^{2x}.
\end{equation}
This leads to $[W(\vr_{\vec{\theta}})]_{ij} = 4 \lim_{\alpha \to \infty} [\mathcal{G}_\alpha (\vr_{\vec{\theta}})]_{ij}$, where $[\mathcal{G}_\alpha(\vr_{\vec{\theta}})]_{ij} \equiv \tr[\swap \mathcal{W}_\alpha \mathcal{H}_i  \otimes  \mathcal{H}_j]$.

Also, we can show the identity $\swap \mathcal{W}_\alpha \swap = - \mathcal{W}_\alpha$. This follows from the properties that $\swap (A \otimes B) \swap = B \otimes A$ and $\swap A^k \swap = (\swap A \swap)^k$ for operators $A,B$ and integers $k$, due to $\swap^2 = \eins$. The identity can imply that $\mathcal{W}_\alpha$ can be written in linear combinations of $\vr^a \otimes \vr^b - \vr^b \otimes \vr^a$ as follows: 
\begin{equation}
    \mathcal{W}_\alpha = \sum_{a,b=0}^{2\alpha} \xi_{ab}^{(\alpha)}
    (\vr^a \otimes \vr^b - \vr^b \otimes \vr^a),
\end{equation}
with some coefficients $\xi_{ab}^{(\alpha)}$ for integers $a,b=0,1,\ldots,2 \alpha$. Using $\tr[\swap (\vr^b \otimes \vr^a) (\mathcal{H}_i  \otimes  \mathcal{H}_j)] = \tr[\swap (\vr^a \otimes \vr^b) (\mathcal{H}_j  \otimes  \mathcal{H}_i)]$, we can thus arrive at Eq.~(\ref{eq:unitaryGform_basisInd_hierarchical}) in the main text. As examples, $\mathcal{W}_0$ and $ \mathcal{W}_1$ are given by
\begin{subequations}
    \begin{align} \nonumber
    \mathcal{W}_0
    &= (\vr \otimes \eins - \eins \otimes \vr)^3
    \\
    &= (\vr^3 \otimes \eins - \eins \otimes \vr^3)
    - 3 (\vr^2 \otimes \vr - \vr \otimes \vr^2),
    \\
    \mathcal{W}_1 \nonumber
    &=(\vr \otimes \eins - \eins \otimes \vr)^3
    (\eins \otimes \eins - \vr \otimes \eins - \eins \otimes \vr)^{2}
    \\ \nonumber
    &=(\vr^3 \otimes \eins - \eins \otimes \vr^3)
    - 2 (\vr^4 \otimes \eins- \eins \otimes \vr^4 )
    + (\vr^5 \otimes \eins - \eins \otimes \vr^5)
    - 3 (\vr^2 \otimes \vr - \vr \otimes \vr^2)
    \\ \nonumber
    &\quad
    + 4 (\vr^3 \otimes \vr - \vr \otimes \vr^3)
    -2 (\vr^3 \otimes \vr^2 - \vr^2 \otimes \vr^3)
     - (\vr^4 \otimes \vr - \vr \otimes \vr^4).
     \end{align}
\end{subequations}
\end{proof}
\section{Derivation of Eq.~(\ref{eq:GE:USU^daggwe})} \label{ap:GE:USU^daggwe}
We begin by considering a rank-$r$ state $\vr$ in the Hilbert space with dimension $D$ and recalling the operator $\mathcal{L}_i$ in Eq.~(\ref{eq:mathcal_L_i}) in the main text is given by
\begin{equation}
    \mathcal{L}_i
    = 2 i \sum_{\substack{k,l=1\\\lambda_k + \lambda_l > 0}}^D
    \frac{\lambda_{k} - \lambda_{l}}{\lambda_k + \lambda_l}
    \mathfrak{h}_{kl}^{(i)} 
    \ket{\psi_k} \! \bra{\psi_l}
    + \mathcal{L}_i^{{\rm kk}},
\end{equation}
where $\mathfrak{h}_{kl}^{(i)} = \braket{\psi_k|\mathcal{H}_i|\psi_l}$, and the term $\mathcal{L}_i^{{\rm kk}}$ acts within the kernel space of $\vr$ but is not uniquely determined from the definition of the SLD operator. Here, the sum runs over indices $k,l$ such that $\lambda_k + \lambda_l > 0$, and the number of nonzero eigenvalues is $r$. Letting $\Pi_k = \ket{\psi_k} \! \bra{\psi_k}$, we expand the term $\mathcal{L}_i \mathcal{L}_j$ as
\begin{align}
    \nonumber
    \mathcal{L}_i \mathcal{L}_j
    &=4
    \sum_{\substack{k,l,m=1
    \\ \lambda_k + \lambda_m > 0
    \\ \lambda_l + \lambda_m > 0}}^{D}
    \frac{\lambda_{k} - \lambda_{m}}{\lambda_k + \lambda_m}
    \frac{\lambda_{l} - \lambda_{m}}{\lambda_l + \lambda_m}
    \mathfrak{h}_{km}^{(i)} \mathfrak{h}_{ml}^{(j)}
    \ket{\psi_k} \! \bra{\psi_l}
    + \mathcal{L}_i \mathcal{L}_j^{{\rm kk}} + \mathcal{L}_i^{{\rm kk}} \mathcal{L}_j + \mathcal{L}_i^{{\rm kk}} \mathcal{L}_j^{{\rm kk}}
    \\
    \label{eq:L_iL_jexpansion}
    &= 4
    \left\{
    \sum_{\substack{k \neq l
    \\ \lambda_k + \lambda_l > 0}}^{D}
    \frac{(\lambda_{k} - \lambda_{l})^2}{(\lambda_k + \lambda_l)^2}
    \mathfrak{h}_{kl}^{(i)} \mathfrak{h}_{lk}^{(j)}
    \Pi_k
    +
    \sum_{\substack{k \neq l \neq m
    \\ \lambda_k + \lambda_m > 0
    \\ \lambda_l + \lambda_m > 0}}^{D}
    \frac{\lambda_{k} - \lambda_{m}}{\lambda_k + \lambda_m}
    \frac{\lambda_{l} - \lambda_{m}}{\lambda_l + \lambda_m}
    \mathfrak{h}_{km}^{(i)} \mathfrak{h}_{ml}^{(j)}
    \ket{\psi_k} \! \bra{\psi_l}
    \right\}
    \\
    &\quad
    +
    \mathcal{L}_i \mathcal{L}_j^{{\rm kk}} + \mathcal{L}_i^{{\rm kk}} \mathcal{L}_j + \mathcal{L}_i^{{\rm kk}} \mathcal{L}_j^{{\rm kk}},
\end{align}
where $k \neq l \neq m$ means $k \neq l$, $l \neq m$, and $k \neq m$.

In the following, we denote $\ket{\psi_{k}}$ for $k \in [r +1, D]$ as eigenstates with zero eigenvalue $\lambda_{k} = 0$ (there are $(D-r)$ number of zero eigenvalues since the state $\vr$ has rank $r$). Note that the kernel space of $\vr$ is spanned by $\{ \ket{\psi_{k}}\}$ for $k \in [r+1, D]$. Using the identity $(\lambda_{k} - \lambda_{l})^2 = (\lambda_{k} + \lambda_{l})^2 - 4 \lambda_{k} \lambda_{l}$, the first term in Eq.~(\ref{eq:L_iL_jexpansion}) is rewritten as
\begin{align}
    \label{eq:First_L_iL_jexpansion}
    \sum_{\substack{k \neq l
    \\ \lambda_k + \lambda_l > 0}}^D
    \frac{(\lambda_{k} - \lambda_{l})^2}{(\lambda_k + \lambda_l)^2}
    \mathfrak{h}_{kl}^{(i)} \mathfrak{h}_{lk}^{(j)}
    \Pi_k
    &=
    \sum_{k =1}^r \sum_{l=r+1}^D
    \Bigl[ 
    \mathfrak{h}_{kl}^{(i)} \mathfrak{h}_{lk}^{(j)}
    \Pi_k
    +
    \mathfrak{h}_{lk}^{(i)} \mathfrak{h}_{kl}^{(j)}
    \Pi_l
    \Bigr]
    +
    \sum_{k \neq l}^r
    \mathfrak{h}_{kl}^{(i)} \mathfrak{h}_{lk}^{(j)}
    \Pi_k
    - \sum_{k \neq l}^r
    \frac{4 \lambda_{k}\lambda_{l}}{(\lambda_k + \lambda_l)^2}
    \mathfrak{h}_{kl}^{(i)} \mathfrak{h}_{lk}^{(j)}
    \Pi_k.
\end{align}
Also the second term in Eq.~(\ref{eq:L_iL_jexpansion}) is rewritten as
\begin{align} \nonumber
    &\sum_{\substack{k \neq l \neq m
    \\ \lambda_k + \lambda_m > 0
    \\ \lambda_l + \lambda_m > 0}}^D
    \frac{\lambda_{k} - \lambda_{m}}{\lambda_k + \lambda_m}
    \frac{\lambda_{l} - \lambda_{m}}{\lambda_l + \lambda_m}
    \mathfrak{h}_{km}^{(i)} \mathfrak{h}_{ml}^{(j)}
    \ket{\psi_k} \! \bra{\psi_l}
    \\ \nonumber
    = \quad
    &\sum_{k \neq l}^r \sum_{m=r+1}^D
    \Bigl[
    \mathfrak{h}_{km}^{(i)} \mathfrak{h}_{ml}^{(j)}
    \ket{\psi_k} \! \bra{\psi_l}
    -
    \frac{\lambda_{k} - \lambda_{l}}{\lambda_k + \lambda_l}
    \bigl(
    \mathfrak{h}_{kl}^{(i)} \mathfrak{h}_{lm}^{(j)}
    \ket{\psi_k} \! \bra{\psi_m}
    +
    \mathfrak{h}_{ml}^{(i)} \mathfrak{h}_{lk}^{(j)}
    \ket{\psi_m} \! \bra{\psi_k}
    \bigr)
    \Bigr]
    \\
    + \
    &\sum_{m=1}^r \sum_{k \neq l}^D
    \mathfrak{h}_{km}^{(i)} \mathfrak{h}_{ml}^{(j)}
    \ket{\psi_k} \! \bra{\psi_l}
    + 
    \sum_{k \neq l \neq m}^r
    \frac{\lambda_{k} - \lambda_{m}}{\lambda_k + \lambda_m}
    \frac{\lambda_{l} - \lambda_{m}}{\lambda_l + \lambda_m}
    \mathfrak{h}_{km}^{(i)} \mathfrak{h}_{ml}^{(j)}
    \ket{\psi_k} \! \bra{\psi_l}.
    \label{eq:Second_L_iL_jexpansion}
\end{align}
Moreover, the term $\mathcal{L}_i \mathcal{L}_j^{{\rm kk}} + \mathcal{L}_i^{{\rm kk}} \mathcal{L}_j$ is rewritten as
\begin{align}
    \label{eq:prod_kernel-kernel_SLD}
    \mathcal{L}_i \mathcal{L}_j^{{\rm kk}} + \mathcal{L}_i^{{\rm kk}} \mathcal{L}_j
    = 2 i \sum_{k=1}^r \sum_{l=r+1}^D
    \mathfrak{h}_{kl}^{(i)} 
    \ket{\psi_k} \! \bra{\psi_l}
    \mathcal{L}_j^{{\rm kk}}
    -
    2 i \sum_{k=r+1}^D \sum_{l=1}^r
    \mathfrak{h}_{kl}^{(j)} 
    \mathcal{L}_i^{{\rm kk}}
    \ket{\psi_k} \! \bra{\psi_l}.
\end{align}

To proceed, we consider the projector onto the kernel space of $\vr$: $\Pi_{\vr}^\perp = \sum_{k=r+1}^D \ket{\psi_k} \! \bra{\psi_k}$, where $\vr \Pi_{\vr}^\perp = \Pi_{\vr}^\perp \vr = 0$. Notice that $\Pi_{\vr}^\perp \mathcal{L}_i^{{\rm kk}}\Pi_{\vr}^\perp = \mathcal{L}_i^{{\rm kk}}$. Employing the identity $\Pi_{\vr}^\perp = \eins - \Pi_\vr$ for $\Pi_\vr = \sum_{k =1}^r \Pi_k$ and recalling $\mathfrak{h}_{kl}^{(i)} = \braket{\psi_k|\mathcal{H}_i|\psi_l}$, the first, second, and third terms in Eq.~(\ref{eq:First_L_iL_jexpansion}) are given by
\begin{subequations}
\begin{align}
    \sum_{k =1}^r \sum_{l=r+1}^D
    \mathfrak{h}_{kl}^{(i)} \mathfrak{h}_{lk}^{(j)}
    \Pi_k
    &= \sum_{k =1}^r
    \braket{\psi_k|\mathcal{H}_i \Pi_{\vr}^\perp \mathcal{H}_j|\psi_k}
    \Pi_k
    = 
    \sum_{k = 1}^r
    \left[
    \Pi_k \mathcal{H}_i \mathcal{H}_j \Pi_k
    - \Pi_k \mathcal{H}_i \Pi_\vr \mathcal{H}_j \Pi_k
    \right],
    \\
    \sum_{k =1}^r \sum_{l=r+1}^D
    \mathfrak{h}_{lk}^{(i)} \mathfrak{h}_{kl}^{(j)}
    \Pi_l
    &= \sum_{l = r +1}^D
    \braket{\psi_l|\mathcal{H}_i \Pi_\vr \mathcal{H}_j|\psi_l}
    \Pi_l
    = 
    \sum_{k = r +1}^D \Pi_k
    \mathcal{H}_i \Pi_\vr \mathcal{H}_j
    \Pi_k
    \\
    \sum_{k \neq l}^r
    \mathfrak{h}_{kl}^{(i)} \mathfrak{h}_{lk}^{(j)}
    \Pi_k
    &= \sum_{k \neq l}^r
    \braket{\psi_k|\mathcal{H}_i|\psi_l} \!
    \braket{\psi_l|\mathcal{H}_j|\psi_k}
    \Pi_k
    = \sum_{k=1}^r 
    \Pi_k \mathcal{H}_i \Pi_\vr \mathcal{H}_j \Pi_k
    - 
    \sum_{k=1}^r
    \Pi_k \mathcal{H}_i \Pi_k \mathcal{H}_j \Pi_k.
\end{align}
\end{subequations}
Also several terms in Eq.~(\ref{eq:Second_L_iL_jexpansion}) are given by
\begin{subequations}
\begin{align}
    &\sum_{k \neq l}^r \sum_{m=r+1}^D
    \mathfrak{h}_{km}^{(i)} \mathfrak{h}_{ml}^{(j)}
    \ket{\psi_k} \! \bra{\psi_l}
    =
    \sum_{k \neq l}^r
    \braket{\psi_k|\mathcal{H}_i \Pi_{\vr}^\perp \mathcal{H}_j|\psi_l}
    \ket{\psi_k} \! \bra{\psi_l}
    =
    \sum_{k \neq l}^r
    \left[
    \Pi_k \mathcal{H}_i \mathcal{H}_j \Pi_l
    - \Pi_k \mathcal{H}_i \Pi_\vr \mathcal{H}_j \Pi_l
    \right],
    \\
    &\sum_{k \neq l}^r \sum_{m=r+1}^D
    \frac{\lambda_{k} - \lambda_{l}}{\lambda_k + \lambda_l}
    \bigl(
    \mathfrak{h}_{kl}^{(i)} \mathfrak{h}_{lm}^{(j)}
    \ket{\psi_k} \! \bra{\psi_m}
    +
    \mathfrak{h}_{ml}^{(i)} \mathfrak{h}_{lk}^{(j)}
    \ket{\psi_m} \! \bra{\psi_k}
    \bigr)
    =
    \sum_{k \neq l}^r
    \frac{\lambda_{k} - \lambda_{l}}{\lambda_k + \lambda_l}
    \left[
    \Pi_k \mathcal{H}_i \Pi_l \mathcal{H}_j \Pi_{\vr}^\perp
    +
    \Pi_{\vr}^\perp \mathcal{H}_i \Pi_l \mathcal{H}_j \Pi_k
    \right],
    \\
    &\sum_{m=1}^r \sum_{k \neq l}^D
    \mathfrak{h}_{km}^{(i)} \mathfrak{h}_{ml}^{(j)}
    \ket{\psi_k} \! \bra{\psi_l}
    =
    \sum_{m=1}^r \sum_{k \neq l}^D
    \braket{\psi_k|\mathcal{H}_i|\psi_m} \!
    \braket{\psi_m|\mathcal{H}_j|\psi_l}
    \ket{\psi_k} \! \bra{\psi_l}
    = \sum_{k \neq l}^D
    \Pi_k \mathcal{H}_i \Pi_\vr \mathcal{H}_j \Pi_l.
\end{align}
\end{subequations}
Furthermore, Eq.~(\ref{eq:prod_kernel-kernel_SLD}) is rewritten as
\begin{equation}
    \mathcal{L}_i \mathcal{L}_j^{{\rm kk}} + \mathcal{L}_i^{{\rm kk}} \mathcal{L}_j
    = 2 i 
    \Pi_\vr \mathcal{H}_i (\eins - \Pi_\vr) \mathcal{L}_j^{{\rm kk}}
    -
    2 i 
    \mathcal{L}_i^{{\rm kk}} (\eins - \Pi_\vr) \mathcal{H}_j \Pi_\vr
    = 2i \left(
    \Pi_\vr \mathcal{H}_i \mathcal{L}_j^{{\rm kk}}
    - \mathcal{L}_i^{{\rm kk}} \mathcal{H}_j \Pi_\vr
    \right),
\end{equation}
where $\Pi_\vr \mathcal{L}_j^{{\rm kk}} = \mathcal{L}_j^{{\rm kk}} \Pi_\vr = 0$.

To summarize these terms, let us denote 
\begin{equation}
    \mathcal{D}_{ij}^{\vr} \equiv \mathcal{H}_i \Pi_\vr \mathcal{H}_j - \mathcal{H}_j \Pi_\vr \mathcal{H}_i.
\end{equation}
Then we represent the term $[\mathcal{S} (\vr_{\vec{\theta}})]_{ij} =  \mathcal{L}_i \mathcal{L}_j - \mathcal{L}_j \mathcal{L}_i$ as follows:
\begin{align}
    \nonumber
    [\mathcal{S} (\vr_{\vec{\theta}})]_{ij}
    &= 4
    \Biggl\{
    \sum_{k =1}^r
    \left[
    \Pi_k (\mathcal{H}_i \mathcal{H}_j - \mathcal{H}_i \mathcal{H}_j) \Pi_k
    - \Pi_k \mathcal{D}_{ij}^{\vr} \Pi_k
    \right]
    +
    \sum_{k = r +1}^D \Pi_k \mathcal{D}_{ij}^{\vr} \Pi_k
    \\
    \nonumber
    &+
    \sum_{k=1}^r 
    \Pi_k \mathcal{D}_{ij}^{\vr} \Pi_k
    - 
    \sum_{k=1}^r
    \Pi_k (\mathcal{H}_i \Pi_k \mathcal{H}_j - \mathcal{H}_j \Pi_k \mathcal{H}_i) \Pi_k
    \\
    \nonumber
    &- 4
    \sum_{k \neq l}^r
    \frac{\lambda_{k}\lambda_{l}}{(\lambda_k + \lambda_l)^2}
    \Pi_k (\mathcal{H}_i \Pi_l \mathcal{H}_j - \mathcal{H}_j \Pi_l \mathcal{H}_i) \Pi_k
    +
    \sum_{k \neq l}^r
    \left[
    \Pi_k (\mathcal{H}_i \mathcal{H}_j - \mathcal{H}_j \mathcal{H}_i) \Pi_l
    - \Pi_k \mathcal{D}_{ij}^{\vr} \Pi_l
    \right]
    \\
    \nonumber
    &-
    \sum_{k \neq l}^r
    \frac{\lambda_{k} - \lambda_{l}}{\lambda_k + \lambda_l}
    \left[
    \Pi_k (\mathcal{H}_i \Pi_l \mathcal{H}_j - \mathcal{H}_j \Pi_l \mathcal{H}_i) \Pi_{\vr}^\perp
    +
    \Pi_{\vr}^\perp (\mathcal{H}_i \Pi_l \mathcal{H}_j - \mathcal{H}_j \Pi_l \mathcal{H}_i) \Pi_k
    \right]
    \\
    \nonumber
    &+
    \sum_{k \neq l}^D
    \Pi_k \mathcal{D}_{ij}^{\vr} \Pi_l
    +
    \sum_{k \neq l \neq m}^r
    \frac{\lambda_{k} - \lambda_{m}}{\lambda_k + \lambda_m}
    \frac{\lambda_{l} - \lambda_{m}}{\lambda_l + \lambda_m}
    \Pi_k \left(\mathcal{H}_i \Pi_m \mathcal{H}_j - \mathcal{H}_j \Pi_m \mathcal{H}_i\right) \Pi_l
    \Biggr\}
    \\
    &+
    2i \left(
    \Pi_\vr \mathcal{H}_i \mathcal{L}_j^{{\rm kk}}
    - \mathcal{L}_i^{{\rm kk}} \mathcal{H}_j \Pi_\vr
    \right)
    + \mathcal{L}_i^{{\rm kk}} \mathcal{L}_j^{{\rm kk}} - \mathcal{L}_j^{{\rm kk}} \mathcal{L}_i^{{\rm kk}},
\end{align}
where $\mathfrak{h}_{km}^{(i)} \mathfrak{h}_{ml}^{(j)} \ket{\psi_k} \! \bra{\psi_l} = \Pi_k \mathcal{H}_i \Pi_m \mathcal{H}_j \Pi_l$. To summarize this further, let us introduce
\begin{equation}
    \mathcal{D}_{ij}^k \equiv \mathcal{H}_i \Pi_k \mathcal{H}_j - \mathcal{H}_j \Pi_k \mathcal{H}_i,
    \quad
    \mathcal{D}_{ij}^\vr = \sum_{k=1}^r \mathcal{D}_{ij}^k.
\end{equation}
Using the relations
\begin{subequations}
    \begin{align}
        &\sum_{k =1}^r \Pi_k (\mathcal{H}_i \mathcal{H}_j - \mathcal{H}_i \mathcal{H}_j) \Pi_k
        + \sum_{k \neq l}^r \Pi_k (\mathcal{H}_i \mathcal{H}_j - \mathcal{H}_j \mathcal{H}_i) \Pi_l
        = \Pi_\vr (\mathcal{H}_i \mathcal{H}_j - \mathcal{H}_i \mathcal{H}_j) \Pi_\vr,
        \\
        &\sum_{k = r +1}^D \Pi_k \mathcal{D}_{ij}^{\vr} \Pi_k + \sum_{k \neq l}^D
        \Pi_k \mathcal{D}_{ij}^{\vr} \Pi_l = \Pi_{\vr}^\perp \mathcal{D}_{ij}^{\vr} \Pi_{\vr}^\perp,
    \end{align}
\end{subequations}
we obtain
\begin{align}
    \nonumber
    [\mathcal{S} (\vr_{\vec{\theta}})]_{ij}
    &= 4
    \Biggl\{
    \Pi_\vr (\mathcal{H}_i \mathcal{H}_j - \mathcal{H}_i \mathcal{H}_j) \Pi_\vr
    -
    \sum_{k \neq l}^r
    \Bigl[
    \Pi_k \mathcal{D}_{ij}^{\vr} \Pi_l
    +
    \frac{ 4\lambda_{k}\lambda_{l}}{(\lambda_k + \lambda_l)^2}
    \Pi_k \mathcal{D}_{ij}^l \Pi_k
    \Bigr]
    \\
    \nonumber
    &- 
    \sum_{k=1}^r
    \Pi_k \mathcal{D}_{ij}^k \Pi_k
    +
    \sum_{k \neq l \neq m}^r
    \frac{\lambda_{k} - \lambda_{m}}{\lambda_k + \lambda_m}
    \frac{\lambda_{l} - \lambda_{m}}{\lambda_l + \lambda_m}
    \Pi_k \mathcal{D}_{ij}^m \Pi_l
    \\
    \nonumber    
    &-
    \sum_{k \neq l}^r
    \frac{\lambda_{k} - \lambda_{l}}{\lambda_k + \lambda_l}
    \left[
    \Pi_k \mathcal{D}_{ij}^l \Pi_{\vr}^\perp
    +
    \Pi_{\vr}^\perp \mathcal{D}_{ij}^l \Pi_k
    \right]
    + \Pi_{\vr}^\perp \mathcal{D}_{ij}^{\vr} \Pi_{\vr}^\perp
    \Biggr\}
    \\
    &+2i \left(
    \Pi_\vr \mathcal{H}_i \mathcal{L}_j^{{\rm kk}}
    - \mathcal{L}_i^{{\rm kk}} \mathcal{H}_j \Pi_\vr
    \right)
    + \mathcal{L}_i^{{\rm kk}} \mathcal{L}_j^{{\rm kk}} - \mathcal{L}_j^{{\rm kk}} \mathcal{L}_i^{{\rm kk}}.
\end{align}
For simplicity, we denote that
$\mathcal{I}_{\rm ss}(\vr_{\vec{\theta}})$,
$\mathcal{I}_{\rm ss}^\prime(\vr_{\vec{\theta}})$,
$\mathcal{I}_{\rm sk}(\vr_{\vec{\theta}})$,
$\mathcal{I}_{\rm ks}(\vr_{\vec{\theta}})$,
$\mathcal{I}_{\rm kk}(\vr_{\vec{\theta}})$,
$\mathcal{J}_{\rm sk}(\vr_{\vec{\theta}})$,
$\mathcal{J}_{\rm ks}(\vr_{\vec{\theta}})$, and
$\mathcal{J}_{\rm kk}(\vr_{\vec{\theta}})$
with elements defined as
\begin{subequations}
\begin{align}
    [\mathcal{I}_{\rm ss}(\vr_{\vec{\theta}})]_{ij}
    &= 4 \Pi_\vr (\mathcal{H}_i \mathcal{H}_j - \mathcal{H}_i \mathcal{H}_j) \Pi_\vr,
    \\
    [\mathcal{I}_{\rm ss}^\prime(\vr_{\vec{\theta}})]_{ij}
    &= -4 \Biggl\{
    \sum_{k=1}^r
    \Pi_k \mathcal{D}_{ij}^k \Pi_k   
    + \sum_{k \neq l}^r
    \Bigl[
    \Pi_k \mathcal{D}_{ij}^{\vr} \Pi_l
    +
    \frac{ 4\lambda_{k}\lambda_{l}}{(\lambda_k + \lambda_l)^2}
    \Pi_k \mathcal{D}_{ij}^l \Pi_k
    \Bigr]
    - \sum_{k \neq l \neq m}^r
    \eta_{km} \eta_{lm}
    \Pi_k \mathcal{D}_{ij}^m \Pi_l
    \Biggr\},
    \\
    [\mathcal{I}_{\rm sk}(\vr_{\vec{\theta}})]_{ij}
    &= - 4 \sum_{k \neq l}^r \eta_{kl}
    \Pi_k \mathcal{D}_{ij}^l \Pi_{\vr}^\perp,
    \\
    [\mathcal{I}_{\rm ks}(\vr_{\vec{\theta}})]_{ij}
    &=  - 4 \sum_{k \neq l}^r \eta_{kl}
    \Pi_{\vr}^\perp \mathcal{D}_{ij}^l \Pi_k,
    \\
    [\mathcal{I}_{\rm kk}(\vr_{\vec{\theta}})]_{ij}
    &= 4 \Pi_{\vr}^\perp \mathcal{D}_{ij}^{\vr} \Pi_{\vr}^\perp,
    \\
    [\mathcal{J}_{\rm sk}(\vr_{\vec{\theta}})]_{ij}
    &= 2i \Pi_\vr \mathcal{H}_i \mathcal{L}_j^{{\rm kk}},
    \\
    [\mathcal{J}_{\rm ks}(\vr_{\vec{\theta}})]_{ij}
    &= -2i \mathcal{L}_i^{{\rm kk}} \mathcal{H}_j \Pi_\vr,
    \\
    [\mathcal{J}_{\rm kk} (\vr_{\vec{\theta}})]_{ij}
    &=\mathcal{L}_i^{{\rm kk}} \mathcal{L}_j^{{\rm kk}} - \mathcal{L}_j^{{\rm kk}} \mathcal{L}_i^{{\rm kk}},
\end{align}
\end{subequations}
where we denote that $\eta_{kl} = (\lambda_{k} - \lambda_{l})/(\lambda_{k} + \lambda_{l})$. Hence, we arrive at the form of $\mathcal{S} (\vr_{\vec{\theta}})$ in Eq.~(\ref{eq:GE:USU^daggwe}) in the main text.

\bibliography{ref.bib}
\end{document}